%% file: main.tex
\documentclass[11pt]{article}
\usepackage[english]{babel}
\usepackage[toc,page]{appendix}
\usepackage{bbm}
\usepackage{fullpage}
\usepackage[top=1in, bottom=1.25in, left=0.9in, right=0.9in]{geometry}
\usepackage{etaremune}
\usepackage{enumitem}

\usepackage{witharrows}

\usepackage{amssymb,amsmath,amsthm}

\usepackage{hyperref}

\DeclareSymbolFont{rsfs}{U}{rsfs}{m}{n}
\DeclareSymbolFontAlphabet{\mathscrsfs}{rsfs}
\usepackage{mathrsfs}

\usepackage{url}

\hypersetup{
  colorlinks   = true, 
  urlcolor     = blue, 
  linkcolor    = blue, 
  citecolor   = red 
}

\usepackage[labelformat=empty]{caption}

\newtheorem{theorem}{Theorem}[section]
\newtheorem{lemma}[theorem]{Lemma}

\newtheorem{proposition}[theorem]{Proposition}
\newtheorem{corollary}[theorem]{Corollary}

\theoremstyle{definition}
\newtheorem{definition}{Definition}
\newtheorem{remark}[theorem]{Remark}

\numberwithin{equation}{section}

\usepackage{times}

\input{commands}

\author{
    Brice Huang\thanks{Department of Electrical Engineering and Computer Science, MIT. Email: \texttt{bmhuang@mit.edu}. Supported by a Google Ph.D. Fellowship, NSF-Simons collaboration grant DMS-2031883, NSF CAREER grant DMS-1940092, and the Solomon Buchsbaum Research Fund at MIT.}
    \and
    Mark Sellke\thanks{Department of Statistics, Harvard University.
    Email: \texttt{msellke@fas.harvard.edu}}
}

\title{Strong Low Degree Hardness for Stable Local Optima in Spin Glasses}

\date{}

\begin{document}

\maketitle

\begin{abstract}
    \noindent It is a folklore belief in the theory of spin glasses and disordered systems that out-of-equilibrium dynamics fail to find stable local optima exhibiting e.g.\ local strict convexity on physical time-scales.
    In the context of the Sherrington--Kirkpatrick spin glass, \cite{behrens2022dis,minzer2023perfectly} have recently conjectured that this obstruction may be inherent to all efficient algorithms, despite the existence of exponentially many such optima throughout the landscape.
    We prove this search problem exhibits \emph{strong low degree hardness} for polynomial algorithms of degree $D\leq o(N)$: any such algorithm has probability $o(1)$ to output a stable local optimum.
    To the best of our knowledge, this is the first result to prove that even constant-degree polynomials have probability $o(1)$ to solve a random search problem without planted structure.
    To prove this, we develop a general-purpose enhancement of the ensemble overlap gap property, and as a byproduct improve previous results on spin glass optimization, maximum independent set, random $k$-SAT, and the Ising perceptron to strong low degree hardness.
    This technique yields lower bounds which are optimal under the low-degree heuristic.
    Finally for spherical spin glasses with no external field, we prove that Langevin dynamics does not find stable local optima within dimension-free time.
\end{abstract}


\section{Introduction}

Out of equilibrium dynamics in disordered systems have been actively studied for decades, in a wide range of models and regimes exhibiting both fast and slow mixing.
Advances in this direction have come from statistical physics, probability theory, and theoretical computer science
\cite{sompolinsky1981dynamic,sompolinsky1981time,sompolinsky1982relaxational,sommers1987path,crisanti1993spherical,cugliandolo1994out,bouchaud1995aging,grunwald1996sanov,biroli1999dynamical,arous1995large,arous1997symmetric,arous2002aging,cugliandolo2004course,ben2006cugliandolo,guionnet2007dynamics,ben2008universality,ben2012universality,gheissari2019spectral,ben2020bounding,eldan2021spectral,anari2021entropic,adhikari2024spectral,sellke2023threshold,anari2024universality,huang2024weak}.
While these dynamics can be naturally viewed as attempting to optimize the Hamiltonian or sample from the Gibbs measure, alternative algorithms which are specially designed to accomplish these tasks have emerged as well \cite{subag2018following,montanari2021optimization,ams20,alaoui2022sampling,montanari2023solving,chen2023local,huang2024sampling}.
In all, this line of work has made substantial progress towards characterizing the power and fundamental limits of efficient algorithms in non-convex, high-dimensional landscapes.

One may also ask a more general question: \emph{which regions in a high-dimensional random landscape are accessible by efficient algorithms}?
Motivated by a bold prediction from the statistical physics literature, we investigate a different aspect of this question than was studied in the aforementioned work.
This prediction states that low-temperature dynamics are unable to reach \emph{stable} local optima, but instead wander the ``manifold'' of marginally stable states (see e.g.\ \cite{biroli1999dynamical,muller2015marginal,parisi2017marginally,kent2024algorithm}).
Recently \cite{behrens2022dis} has proposed a substantial extension of this conjecture: \emph{``omnipresent marginal stability in glasses is a consequence of computational hardness''}.
In other words, stable local optima might be inaccessible to \emph{all efficient algorithms}.
This conjecture is also consistent with recent findings from \cite{subag2018following,huang2021tight,sellke2023threshold,huang2024optimization} on pure spherical spin glasses, where optimization beyond the marginal stability threshold is impossible within a broad class of algorithms with Lipschitz dependence on the disorder.
Another relevant connection is with the popular belief in deep learning that flat local optima generalize better \cite{keskar2016large}. 
From this point of view, it would be promising if algorithms naturally prefer flat optima to stable optima; see also \cite{baity2018comparing}.

We make progress on these conjectures in two paradigmatic mean-field disordered systems.
In the Sherrington--Kirkpatrick spin glass, we prove that stable local optima are inaccessible for low degree polynomial algorithms of degree $D\leq o(N)$.
According to standard computational complexity heuristics \cite{hopkins2018statistical,kunisky2019notes}, this suggests that locating such points may require $e^{\Omega(N)}$ time (the same as brute force search over the full state space).
For spherical spin glasses without external field, we prove that Langevin dynamics does not find stable local optima on dimension-free time-scales.

\subsection{Models and Main Results}

Our main results are for two closely related models.
The first is the Sherrington--Kirkpatrick (SK) model with Hamiltonian
$H_N:\bSig_N\equiv\{-1,1\}^N\to \bbR$ defined for IID Gaussians $(g_{i,j})_{1\leq i,j\leq N}$ by:
\begin{equation}
\label{eq:SK-def}
    H_N(\bsig)
    =
    \frac{1}{\sqrt{N}}
    \sum_{i,j=1}^N
    g_{i,j} \sigma_i\sigma_j.
\end{equation}
We also consider spherical mixed spin glasses, where the Hamiltonian takes a more general form. For $2\leq k\leq \bar k$, let $\bG_N^{(k)} \in \lt(\bbR^N\rt)^{\otimes k}$ be an order $k$ tensor with IID standard Gaussian entries $g_{i_1,\dots,i_k}\sim \cN(0,1)$.
A general mean-field spin glass Hamiltonian, also denoted $H_N$, is the random polynomial
\begin{equation}
    \label{eq:def-hamiltonian}
    H_N(\bsig)
    =
    \sum_{k=2}^{\bar k}
    \gamma_k
    N^{-\frac{k-1}{2}}
    \la \bG_N^{(k)}, \bsig^{\otimes k} \ra
    =
    \sum_{k=2}^{\bar k}
    \gamma_k
    N^{-\frac{k-1}{2}}
    \sum_{i_1,\dots,i_k=1}^N g_{i_1,\dots,i_k}\sigma_{i_1}\dots\sigma_{i_k}
    .
\end{equation}
An equivalent definition is that $H_N$ is a centered Gaussian process on $\bbR^N$ with covariance
\begin{equation}
    \label{eq:def-xi}
\bbE[H_N(\bsig)H_N(\brho)]
=
N\xi\lt(\frac{\la\bsig,\brho\ra}{N}\rt)
\equiv
N
\sum_{k=2}^{\bar k}
\gamma_k^2 \lt(\frac{\la \bsig,\brho\ra}{N}\rt)^k.
\end{equation}
We will always take $\cS_N=\{\bsig\in\bbR^N~:~\|\bsig\|=\sqrt{N}\}$ to be the domain of a mixed $p$-spin Hamiltonian.
We use the same notation $H_N$ in both cases, as the model should always be clear from context.
(In fact, aside from the different domains, the SK Hamiltonian is a special case of \eqref{eq:def-hamiltonian} with $\gamma_3=\dots=\gamma_{\bar k}=0$.)

In both models, we will be interested in \emph{stable local maxima}.
On the sphere, the natural ``strict'' definition of a stable local maximum is simply a critical point $\bsig\in\cS_N$ for which $\lambda_{\max}(\nabla_{\sph}^2 H_N(\bsig))\leq -\gamma$ for some $\gamma>0$ independent of $N$.
We often refer to such points as \emph{wells}.
For the SK model, the analogous definition of (stable) local maximum is based on single-spin flips.
Define the $i$-th local field
\begin{equation}
\label{eq:local-field-def}
L_i(\bsig)=
\sigma_i\cdot (H_N(\bsig)
-
H_N(\bsig\oplus e_i)) / 2
.
\end{equation}
Here $\bsig\oplus e_i\in\bSig_N$ disagrees with $\bsig$ precisely in the $i$-th coordinate.
By construction, any local maximum in the SK model (with respect to the nearest neighbor graph) satisfies $\sigma_i L_i(\bsig)\geq 0$ for all $\bsig$.
We say such a local maximum is a $\gamma$-gapped state if $\min_i \sigma_i L_i(\bsig)\geq \gamma>0$, i.e. the energy cost of any spin flip is uniformly positive.

In both settings, our hardness results actually extend to a more lenient definition of stable local optimum.
Namely, we allow a small $\delta$ fraction of the local fields or eigenvalues to violate the constraint, or even take the wrong sign. In Definition~\ref{def:well} just below, $\nabla_{\sph} H_N$ and $\nabla_{\sph}^2 H_N$ denote the Riemannian gradient and Hessian on $\cS_N$; these are defined precisely in \eqref{eq:spherical-calculus-def}, \eqref{eq:riemannian-hessian}.

\begin{definition}
\label{def:gapped}
    For $\gamma,\delta>0$, the point $\bsig\in\{-1,1\}^N$ is a \textbf{$(\gamma,\delta)$-gapped state} for the SK Hamiltonian $H_N$ if
    \[
    |\{i\in [N]~:~\sigma_i L_i(\bsig)<\gamma\}|\leq \delta N.
    \]
    We let $S(\gamma,\delta)=S(\gamma,\delta;H_N)\subseteq\bSig_N$ be the set of such points.
\end{definition}

\begin{definition}
\label{def:well}
$\bsig\in\cS_N$ is a \textbf{$(\gamma,\delta)$-well} for the mixed $p$-spin Hamiltonian $H_N$ if $\|\nabla_{\sph} H_N(\bsig)\|\leq \delta\sqrt{N}$ and
\[
\lambda_{\delta N}(\nabla_{\sph}^2 H_N(\bsig))\leq \gamma<0.
\]
We let $W(\gamma,\delta)=W(\gamma,\delta;H_N)\subseteq\cS_N$ be the set of such points.
\end{definition}



Our first main result, Theorem~\ref{thm:SK-hardness-LDP}, shows gapped states for the SK model cannot be found using large classes of efficient algorithms.
We provide three separate statements, which trade-off the power of the algorithm class with the bound obtained on the algorithm's success probability.
For our purposes, an ``algorithm'' means a measurable function $\cA^{\circ}:(H_N,\omega)\mapsto \bsig\in\bSig_N$, where $\omega$ is an independent random variable in some Polish space $\Omega_N$.
The variable $\omega$ enables $\cA$ to be randomized; we say $\cA$ is deterministic if it does not depend on $\omega$.
Since the most relevant classes of algorithms for us naturally have continuous outputs, we turn such an $\cA^{\circ}:(H_N,\omega)\mapsto \bx\in\bbR^N$ into a $\bSig_N$-valued algorithm $\cA$ using a randomized rounding scheme.
With $\vec U=(U_1,\dots,U_N)\stackrel{IID}{\sim}\Unif([-1,1])$ independent of $(H_N,\omega)$, we analyze the performance of
\begin{equation}
\label{eq:round-stable-alg}
\cA(H_N,\omega,\vec U)
\equiv
\round_{\vec U}(\cA^{\circ}(H_N,\omega)).
\end{equation}
Here we define the coordinate-wise function $\round_{\vec U}(\bx)=\big(\round_{U_1}(x_1),\dots,\round_{U_N}(x_N)\big)$ for
\begin{align*}
\round_U(x)
=
\begin{cases}
1,\quad x\geq U,
\\
-1 \quad x<U.
\end{cases}
\end{align*}
Below, we say $\cA^{\circ}$ is $C$-Lipschitz if it is $C$-Lipschitz for each fixed $\omega$, where the domain $H_N\simeq (\bG^{(k)})_{2\leq k\leq \bar k}$ and codomain are metrized by the unnormalized Euclidean norm.
We say $\cA^{\circ}$ is a degree $D$ polynomial if each of its $N$ output coordinates is a degree $D$ polynomial in the coefficients of $H_N$, for any fixed $\omega$.

\begin{theorem}
\label{thm:SK-hardness-LDP}
For any $\gamma>0$ there exists $\delta>0$ such that for any fixed $C>0$ and $N$ sufficiently large:
\begin{enumerate}[label=(\Roman*)]
    \item
    \label{it:lipschitz-alg-conseq}
    If $\cA^{\circ}$ is $C$-Lipschitz, then $\cA$ (as defined in \eqref{eq:round-stable-alg}) satisfies for some $c(\gamma,\delta,C)>0$:
    \[
    \bbP[\cA(H_N,\omega,\vec U)\in S(\gamma,\delta;H_N)]\leq e^{-cN}.
    \]
    \item
    \label{it:low-degree-strong-conseq}
    If $\cA^{\circ}$ is a deterministic degree $D\leq \frac{\log N}{11}$ polynomial with $\bbE^{H_N}[\|\cA^{\circ}(H_N)\|^2]\leq CN$, then $\cA$ satisfies
    \[
    \bbP[\cA(H_N,\vec U)\in S(\gamma,\delta;H_N)]\leq e^{-\Omega(DN^{1/3D})}
    \leq
    N^{-\Omega(1)}.
    \]
    \item
    \label{it:low-degree-ultimate-conseq}
    If $\cA^{\circ}$ is a deterministic degree $D\leq o_{N\to\infty}(N)$ polynomial with $\bbE^{H_N}[\|\cA^{\circ}(H_N)\|^2]\leq CN$, then
    \[
    \bbP[\cA(H_N,\vec U)\in S(\gamma,\delta;H_N)]\leq
    (D/N)^{c(\gamma,\delta,C)}
    \leq
    o_{N\to\infty}(1).
    \]
\end{enumerate}
\end{theorem}

Theorem~\ref{thm:SK-hardness-LDP}\ref{it:lipschitz-alg-conseq} is proved using an ensemble version of the overlap gap property (OGP); this technique was first used in \cite{chen2019suboptimality}, following the introduction of the OGP in \cite{gamarnik2014limits}.
As discussed in Subsection~\ref{subsec:strong-hardness-discussion}, parts \ref{it:low-degree-strong-conseq},\ref{it:low-degree-ultimate-conseq} establish \emph{strong low degree hardness} of finding gapped states, and are the first such results for any random search problem without planted structure.
We also mention that from part~\ref{it:low-degree-ultimate-conseq}, it trivially follows that for any \emph{randomized} degree $D\leq o(N)$ polynomial $\cA^{\circ}$ satisfying the averaged bound $\bbE^{H_N,\omega}[\|\cA^{\circ}(H_N,\omega)\|^2]\leq CN$, the rounded algorithm $\cA$ has success probability $o(1)$; just apply the Markov inequality to $\omega\mapsto \bbE[\|\cA^{\circ}(H_N,\omega)\|^2\,|\,\omega]$ to reduce to the deterministic case.

\begin{remark}
\label{rem:rounding-schemes}
    While we used randomized rounding above for concreteness, the precise choice of rounding scheme is not important: one just needs a random function from $\bbR^N$ to $\bSig_N$ which is stable with high probability.
    To illustrate the flexibility, we recall that in \cite[Definition 2.3]{gamarnik2020optimization} a $\bbR^N$-valued algorithm is said to solve a random optimization problem with domain $\bSig_N$ if its coordinate-wise sign solves the problem, and additionally only an $\eta$-fraction of coordinates have absolute value at most $\lambda$. (Here both $\eta$ and $\lambda$ are small dimension-free constants, and the latter is denoted $\gamma$ therein.)
    Given such an solution, one can multiply its output by $1/\lambda$, and then its randomized rounding will almost surely lie within $2\sqrt{\eta N}$ of its coordinate-wise sign.
    In both our setting and that of \cite{gamarnik2020optimization}, this $2\sqrt{\eta N}$ distance has a negligible effect on the algorithm's performance.
    Therefore our results obstruct low degree algorithmic solutions in this sense as well.
    Similar comments apply to the formalizations in \cite{wein2020independent,bresler2021ksat}, who treat coordinates of $\cA^{\circ}(\cdot)$ in a certain interval as ``errors''; this roughly corresponds to worst-case rounding rather than our randomized rounding (and so hardness in our sense implies hardness in theirs).
    Finally if the independent coordinates $U_i$ are non-uniform but drawn from a bounded density, all of our results still apply.
\end{remark}

Turning to spherical spin glasses, we consider the canonical Langevin dynamics (with initialization independent of $H_N$), and prove that this method cannot find wells on dimension-free time-scales.
The Langevin dynamics at inverse temperature $\beta\in\bbR$ are given by the SDE
\begin{equation}
    \label{eq:langevin-dynamics}
    \de \bx_t
    =
    \lt(\beta \nabla_{\sph} H_N(\bx_t)
    -
    \frac{(N-1)\bx_t}{2N}\rt)
    \de t
    +
    P_{\bx_t}^{\perp}\sqrt{2}~\de \bB_t.
\end{equation}
Here $\bx_0$ is independent of the disorder and $P_{\bx}^{\perp}=I_N-\frac{\bx\bx^{\top}}{N}$ is a rank $N-1$ projection matrix, while $\bB_t$ is an independent Brownian motion in $\bbR^N$.
It is a standard fact (see e.g. \cite{hsu2002stochastic}) that $\bx_t$ remains on $\cS_N$ almost surely, and has stationary measure given by the Gibbs measure with density proportional to $e^{\beta H_N(\bx)}$.

\begin{theorem}
\label{thm:langevin-fails-pseudo-wells}
    Let $H_N$ be a mixed $p$-spin model Hamiltonian with covariance function $\xi$ (recall \eqref{eq:def-xi}).
    Then for any $\gamma>0$ there exists $\delta=\delta(\xi,\gamma)>0$ such that the following holds.
    For any $T,\beta>0$ there exists  $c=c(\xi,T,\beta,\gamma,\delta)>0$, such that for $N$ large, if $\bx_0$ is independent of $H_N$ and $\bx_t$ solves \eqref{eq:langevin-dynamics}, then
    \[
    \bbP[\bx_T\in W(\gamma,\delta)]\leq e^{-cN}.
    \]
\end{theorem}

In the pure case $\xi(t)=t^k$, Theorem~\ref{thm:langevin-fails-pseudo-wells} is closely related to the threshold results from \cite{huang2021tight,huang2023algorithmic,sellke2023threshold}.
While these works are focused specifically on the energy value obtained, Theorem~\ref{thm:langevin-fails-pseudo-wells} studies the qualitative behavior of the algorithmically reachable part of the landscape (see \cite{alaoui2024near} for another recent work in a related spirit).
We note that the threshold energy of pure spherical $p$-spin models was first predicted by physicists \cite{cugliandolo1994out,biroli1999dynamical} based on the idea that Langevin dynamics cannot find wells.
Theorem~\ref{thm:langevin-fails-pseudo-wells} therefore confirms the original physical intuition behind these predictions.

In both models, the landscape of local optima has been thoroughly studied in its own right.
On the sphere, the existence of wells has been understood for many $\xi$ via the Kac--Rice formula, which gives a general method to count critical points of random landscapes \cite{auffinger2013random,auffinger2013complexity,fyodorov2014topology,subag2017complexity,subag2017geometry,arous2021exponential,belius2022triviality,huang2023strong}.
In particular, this line of work showed that the ground states of spherical spin glasses are wells when $\xi$ is in the interior of the $1$-RSB or topologically trivial phase at zero temperature.
Recently in \cite{sellke2024marginal}, the second author characterized which spherical spin glasses have ground states that are wells, in terms of the RSB behavior in the Parisi formula.
On the cube, the analogous full RSB behavior is always true for mean-field spin glasses \cite{auffinger2020sk}, which may suggest that ground states are typically not gapped.
Indeed it was only shown recently that gapped states for the SK model exist at all by \cite{minzer2023perfectly,dandi2023maximally}, which in fact determined the maximum possible $\gamma$ (see also \cite{ferber2022friendly}).
Both \cite{behrens2022dis,minzer2023perfectly} conjectured the algorithmic hardness of finding gapped states in the SK model.

\subsection{Discussion of Strong Low Degree Hardness}
\label{subsec:strong-hardness-discussion}

Theorem~\ref{thm:SK-hardness-LDP}\ref{it:low-degree-strong-conseq},\ref{it:low-degree-ultimate-conseq} show \emph{strong low degree hardness} of finding gapped states in the following sense.

\begin{definition}
\label{def:SLDH}
A sequence of random search problems, formally given by an $N$-indexed sequence of random input vectors $\by_N\in\bbR^{d_N}$ and random subsets $S_N=S_N(\by_N)\subseteq\bSig_N$, is said to exhibit \textbf{strong low degree hardness up to degree $D\leq o(D_N)$} if the following holds as $N\to\infty$.
If $\cA^{\circ}=\cA_N^{\circ}(H_N,\omega)\mapsto\bsig\in\bbR^N$ is a sequence of polynomial algorithms of degree $o(D_N)$ obeying $\bbE[\|\cA^{\circ}(\by_N,\omega)\|^2]\leq O(N)$, then for $\cA_N(\by_N,\omega,\vec U)$ as in \eqref{eq:round-stable-alg},
\[
\bbP[\cA_N(\by_N,\omega,\vec U)\in S_N(\by_N)]\leq o(1).
\]
\end{definition}

Theorem~\ref{thm:SK-hardness-LDP}\ref{it:low-degree-strong-conseq}, \ref{it:low-degree-ultimate-conseq} are shown using the same problem-specific strategy as part \ref{it:lipschitz-alg-conseq}, but with a new general-purpose enhancement in implementing the ensemble OGP.
As a result both of these latter statements are qualitatively much stronger than any previous low degree hardness results for random search problems without planted structure.\footnote{
For problems with planted structure, e.g.\ tensor PCA with a hidden signal vector, a convenient approach to low degree hardness is to study the hypothesis-testing problem for the presence/absence of signal.
This can be studied via the low degree likelihood ratio test \cite{hopkins2018statistical,kunisky2019notes}, which gives evidence of hardness in a sense which is
somewhat different from
our results.
The related sum-of-squares lower bounds can also provide evidence for computational hardness, see e.g.\ \cite{hopkins2015tensor,barak2019nearly}.
}
Previous results such as \cite{gamarnik2020optimization,wein2020independent,bresler2021ksat,gamarnik2022algorithms,li2024discrepancy} could at best establish that the success probability is bounded above by $1-f(D)$, where $f(D)>0$ for fixed $D$ but $\lim_{D\to\infty} f(D)=0$.
Our method seems to be compatible with essentially any stability argument based on the discrete-time intermediate value theorem.
Indeed in Section~\ref{sec:low-degree}, we improve results from all five of the papers just mentioned to strong low degree hardness.

\begin{theorem}{(Informal version of results in Subsections~\ref{subsec:prior-ogp-improve}, \ref{subsec:CSP}, \ref{subsec:dense-max-ind-set})}
\label{thm:informal}
    The following problems exhibit strong low degree hardness in a suitable range of parameters:
    \begin{enumerate}[label=(\alph*)]
    \item
    \label{it:informal-1}
    Optimization of pure Ising $k$-spin glasses for $k\geq 4$ even, up to $D\leq o(N)$.
    \item
    \label{it:informal-2}
    The symmetric binary perceptron, up to $D\leq o(N/\log N)$.
    \item
    \label{it:informal-3}
    The asymmetric Ising perceptron, up to $D\leq o(N/\log N)$.
    \item
    \label{it:informal-4}
    Large independent sets in sparse Erdős--Rényi graphs, up to $D\leq o(N)$.
    \item
    \label{it:informal-5}
    Random $k$-SAT, up to $D\leq o(N)$.
    \item
    \label{it:informal-6}
    Large independent sets in $G(N,1/2)$, up to $D\leq o(\log^2 N)$.
    \end{enumerate}
\end{theorem}

As discussed below, the low-degree prediction suggests that the above ranges of $D$ are optimal, except for the logarithmic factors in \ref{it:informal-2}, \ref{it:informal-3}.
These factors stem from Ramsey-theoretic arguments used in the more complicated multi-ensemble OGP, introduced by \cite{gamarnik2021partitioning,gamarnik2022algorithms}.
For case \ref{it:informal-1} of mean-field spin glass optimization, we also prove unconditional optimality within the class of low-degree algorithms, by constructing degree $O(N)$ polynomials which find approximate ground states.

We do not know how to combine our method with the branching overlap gap property \cite{huang2021tight,huang2023algorithmic}, which gives a general technique to compute exact thresholds in mean-field random optimization problems, for the more restricted class of Lipschitz algorithms.
It would be extremely interesting to deduce strong low degree hardness from the branching OGP.

\paragraph{On the importance of high-probability guarantees in OGPs.}

As discussed further in Subsection~\ref{subsec:proof-ideas}, an important step in our implementation of the ensemble OGP is to replace certain costly union bounds by correlation inequalities.
This lower-bounds the probability that a putative algorithm succeeds and is stable across the full ensemble, by a quantity $q^K$ which decays exponentially in the number $K$ of correlated Hamiltonians used in the OGP obstruction.
Consequently the probability that the OGP obstruction holds is crucial to our argument.
In particular if such an obstruction holds with probability $1-p_{\OGP}$, then one can run the desired OGP argument so long as $K\ll \log(1/p_{\OGP})$.
In our arguments, this turns out to correspond to hardness up to degree $D\ll \log(1/p_{\OGP})$.
Thanks to concentration of measure, it is a ubiquitous feature of mean-field disordered systems that rare events have exponentially small probability of order $e^{-\Omega(N)}$.
This explains why we are able to obtain strong low degree hardness up to $D\leq o(N)$ in many settings.

Heuristically, degree $D$ polynomials are believed to often serve as a good proxy for the class of $e^{\tilde O(D)}$-time algorithms \cite{hopkins2018statistical,kunisky2019notes}, though this is known to be false in general even under favorable natural conditions \cite{buhai2025quasi}.
Thus, our results may be taken as evidence that exponential $e^{\tilde\Omega(N)}$ time is required to find gapped states, and to solve the other problems mentioned in \ref{it:informal-1}-\ref{it:informal-5} above.
Conversely since solving any search problem over $\bSig_N$ in exponential time is trivial by brute-force, this heuristic suggests that the range $D\leq o(N)$ may be best possible in such problems.
Given the previous paragraph, our results lead us to propose the following rule of thumb:
\[
    \parbox{11cm}{\emph{If a random search problem exhibits an OGP with probability $1-p_{\OGP}$, then algorithms may require running time $p_{\OGP}^{-\Omega(1)}$ to find a solution.}}
\]
This rule of thumb is also consistent with the maximum clique problem in $G(N,1/2)$.
Here an overlap gap property holds with probability $1-e^{-\Theta(\log^2 N)}$, which leads to Theorem~\ref{thm:informal}\ref{it:informal-6}, while the best known algorithms (based on brute-force search) require running time $e^{\Theta(\log^2 N)}$.
However algebraic algorithms such as Gaussian elimination are known to contradict a fully general statement of the above rule of thumb in certain special cases such as XOR-SAT.

Finally we note that obtaining strong low degree hardness for random search problems has been highlighted recently as a missing link in our understanding of average-case hardness.
In December 2024, it was investigated as an open problem in the American Institute of Mathematics workshop \emph{Low degree polynomial methods in average-case complexity}.
\cite{li2024some} has argued that the lack of strong low degree hardness is a major drawback in the existing OGP literature, and demonstrated that the shortest path problem on randomly weighted graphs (which of course admits efficient algorithms) exhibits an overlap gap property with probability $1-O\big(\frac{\log\log N}{\log N}\big)$.
They ask whether it is possible to ``rule out degree-$O(\log N)$ algorithms which succeed with constant probability'' in emblematic random search problems such as maximum independent set, and whether there are ``additional conditions under which we expect that OGP is a good heuristic for hardness''.
Our results give a resoundingly positive resolution to the first question, and suggest ``OGP with probability $1-N^{-\omega(1)}$'' as a tentative answer to the second.

\subsection{Matching Upper Bound for Spin Glass Optimization}

Although our hardness results are tight under the low-degree heuristic, they do not actually show that low-degree polynomials have good performance when the degree is large enough.
As pointed out in the recent survey \cite{wein2025computational}, such ``low-degree upper bounds'' are important to validate the low-degree heuristic.
For the problem of optimizing an Ising spin glass, we obtain a matching upper bound to Theorem~\ref{thm:informal}\ref{it:informal-1}: degree $O(N)$ algorithms suffice to reach approximate ground states.
\begin{theorem}[Informal, see Theorem~\ref{thm:low-degree-upper-bound-GS}]\label{thm:informal-ub}
    For any Ising mixed $p$-spin model $H_N$ and constant $\eps > 0$, there exists randomized algorithms $\cA_N^\circ:(H_N,\omega)\mapsto \bbR^N$ such that, for each fixed realization of the auxiliary randomness $\omega$, the map $H_N\mapsto \cA_N^\circ(H_N,\omega)$ is a degree $O_{\eps}(N)$ polynomial in the disorder coefficients of $H_N$ such that the following holds.
    Define the rounding
    \[
    \cA_N(H_N,\omega)
    \equiv
    \sign\big(\cA_N^\circ(H_N,\omega)\big)\in \bSig_N\,.
    \]
    Then with high probability $H_N(\cA_N(H_N,\omega))$ $(1-\eps)$-approximates the ground state $\max_{\bsig \in \bSig_N} H_N(\bsig)$.
\end{theorem}
Thus the set of degrees $D$ for which Theorem~\ref{thm:informal}\ref{it:informal-1} provides a lower bound is sharp.
We view this as further evidence for the heuristic correspondence between OGP probability, degree, and runtime described above.
We emphasize that in Theorem~\ref{thm:informal-ub}, the value attained by the degree $O_{\eps}(N)$ algorithm approximates the true maximum of $H_N$, as opposed to the algorithmic maximum predicted by OGP \cite{huang2021tight}.

An analogous result to Theorem~\ref{thm:informal-ub} holds for \emph{spherical} mixed $p$-spin models; see Remark~\ref{rmk:low-deg-upper-bound-sphere}.

\subsection{Proof Ideas}
\label{subsec:proof-ideas}

\paragraph{Lower bounds}

We prove Theorem~\ref{thm:SK-hardness-LDP} via the ensemble overlap gap property \cite{chen2019suboptimality,gamarnik2021survey}.
Following this strategy, we will construct a correlated jointly Gaussian ensemble $(H_N^{(0)},\dots,H_N^{(K)})$ of SK Hamiltonians, and suppose for sake of contradiction that $\cA$ finds a gapped state in all of them.
We then argue that these solutions cannot have intermediate distance between each other, but also that $H_N^{(K)}$ has no gapped state close to $\cA(H_N^{(0)})$.
An unusual feature of our proof is that the main OGP obstruction only applies at very small perturbation levels, but restricts the movement of $\cA$ very harshly.

An interesting conceptual difference from previous work is that we centrally use the fact that $\cA(H_N^{(i)})$ depends only on $H_N^{(i)}$.
In all OGP arguments including ours, the first step is to argue that a putative algorithm $\cA$ should be able to construct a rich constellation of solutions, either for a single Hamiltonian or for a correlated ensemble.
Previous arguments then typically prove a \emph{global landscape obstruction}: with high probability, such ensembles do not exist at all.\footnote{Both \cite{wein2020independent,bresler2021ksat} use a ``chaos'' result of a similar conditional form, but the primary ``OGP'' obstruction is global.}
By contrast we establish what might be called a \emph{conditional landscape obstruction}: for any \emph{fixed} $\cA(H_N^{(i)})\in \bSig_N$, with high conditional probability given $H_N^{(i)}$, the correlated Hamiltonian $H_N^{(j)}$ has no solutions at medium distance from $\cA(H_N^{(i)})$.
As outlined in Remark~\ref{rem:BOGP}, in principle such arguments can also be expressed as global landscape obstructions. However this requires a more complicated ensemble similar to the branching OGP, and does not simplify the proof otherwise.

As mentioned previously, the strong low degree hardness of Theorem~\ref{thm:SK-hardness-LDP}\ref{it:low-degree-strong-conseq},\ref{it:low-degree-ultimate-conseq} comes via another, broadly applicable improvement in the analysis of the ensemble overlap gap property.
Namely a naive implementation of the argument described above requires costly union bounds over the events that $\cA$ succeeds on each instance $H_N^{(i)}$, as well as stability of $\cA$ on adjacent pairs of instances.
Such arguments cannot establish success probabilities smaller than $1-O(1/K)$.
Works including \cite{gamarnik2014limits,rahman2017independent,gamarnik2020optimization,wein2020independent} improved upon this method by establishing positive-correlation properties of the ``success'' events, or of the ``stability'' events.
These improvements cannot be directly combined, because one cannot fruitfully union bound the ``all success'' and ``all stable'' events together.
(Indeed, \cite{gamarnik2014limits,rahman2017independent} use a positive correlation bound over ``success'' events but take a union bound over ``stability'' events, while \cite{gamarnik2020optimization,wein2020independent} do the reverse.)
We identify a new positive-correlation property which simultaneously includes both success and stability events, thus circumventing this difficulty (see Lemma~\ref{lem:grand-correlation}).
Unlike previous implementations of the ensemble OGP, it is crucial that our sequence of Hamiltonians forms a reversible Markov chain.

Theorem~\ref{thm:langevin-fails-pseudo-wells} is proved in a somewhat different way.
Instead of varying the disorder, we vary both the initialization $\bx_0$ and driving Brownian motion $\bB_{[0,T]}$.
A first attempt is to argue that if $\bx_T$ were in a well, then the endpoint $\wt\bx_{T,i}$ of some perturbed dynamics must have different gradient norm, which would contradict known concentration properties of Langevin dynamics.
This works if one uses the stronger definition of ``well'' which requires \emph{all} eigenvalues of the Hessian $\nabla_{\sph}^2 H_N(\bx_T)$ to be smaller than $-\gamma$.
In the presence of Hessian outliers, the same argument implies that the movement of perturbed dynamics away from $\bx_T$ must align strongly with the span of the outlier eigenvectors, which is a constant-dimensional subspace.
Hence if one considers many independently perturbed dynamics, there must exist some pair $\wt\bx_{T,i}, \wt\bx_{T,j}$ of perturbed endpoints which are much closer to each other than to $\bx_T$.
However such behavior turns out to be exponentially unlikely, thanks to a monotonicity property for expected overlaps of correlated trajectories.

Finally we mention a technical step for the analysis of Langevin dynamics, which is needed to show the perturbed dynamics actually move away from $\bx_T$: if $(\bx_0,\bB_{[0,T]})$ and $(\wt\bx_0,\wt\bB_{[0,T]})$ are independent, then the corresponding pairs of dynamics (for the same $H_N$) remain approximately orthogonal on bounded time-scales.
This is proved in Lemma~\ref{lem:langevin-orthogonal}, and is the only step in our proof of Theorem~\ref{thm:langevin-fails-pseudo-wells} which requires the external field to vanish.
This requirement is necessary in the sense that Theorem~\ref{thm:langevin-fails-pseudo-wells} is false for external fields $\gamma_1^2>\xi''(1)-\xi'(1)$ above the topological trivialization threshold; here low-temperature Langevin dynamics does rapidly reach a well, namely the global optimum of $H_N$ (\cite[Section 1.4]{huang2023strong}).
The case of non-zero external field below the topological trivialization threshold is left as an interesting open question, as is the extension to a broader class of algorithms.

\paragraph{Upper bound for spin glass optimization}

The proof of Theorem~\ref{thm:informal-ub} is based on approximating the barycenter of a low-temperature (perturbed) Gibbs measure by a degree $O(N)$ polynomial, by showing the noise stability of this barycenter.
To be precise, let the auxiliary randomness be $\omega = \bh \sim \cN(0, I_N)$.
For $\eta > 0$ small, consider the perturbed Hamiltonian
\[
    H_N^\eta(\bsig) = H_N(\bsig) + \eta \langle \bsig, \bh \ra\,.
\]
For large $\beta > 0$, the Gibbs measure at inverse temperature $\beta$ is
\[
    \mu_N(\bsig) = \frac{\exp(H_N^\eta(\bsig))}{Z_{N,\eta}}\,, \qquad \bsig \in \bSig_N\,.
\]
Let $\bm = \bm(H_N,\bh)$ denote the mean of $\mu_{\beta,\eta}$.
We prove Theorem~\ref{thm:informal-ub} by combining two ingredients:
\begin{itemize}
    \item It follows from standard spin glass arguments that the mean $\bm$ of the low-temperature perturbed model is an approximate ground state of the original model $H_N$.
    That is, $H_N(\bm)/N \ge \GS - o_{\eta,\beta}(1)$, where $\GS$ is the in-probability limiting ground state energy of $H_N$, and $\bm \in [-1,1]^N$ is close to $\bSig_N$ in the sense that $\|\bm\|^2/N = 1 - o_{\eta,\beta}(1)$.
    \item \cite{alaoui2023shattering} proved the following noise stability estimate for the Gibbs measure $\mu_N$.
    Let $(H_N,H_{N,p})$ are $p$-correlated Hamiltonians for $p = 1 - c/N$, where $c$ is a small constant.
    Let $H_N^\eta, H_{N,p}^\eta$ be the associated perturbed Hamiltonians, defined with the same $\bh$, and $\mu_N,\mu_{N,p}$ be the resulting Gibbs measures.
    Then, $\TV(\mu_N,\mu_{N,p})$ is bounded by a small constant $\tau$.
    By considering the optimal coupling of $\mu_N,\mu_{N,p}$, this implies
    \[
        \bbE[\|\bm(H_N,\bh)-\bm(H_{N,p},\bh)\|^2]/N \le 4\tau\,,
    \]
    which implies that $\bm$ is well approximated by its Hermite expansion in the disorder coefficients of $H_N$ up to degree $D = O(N)$.
\end{itemize}
Theorem~\ref{thm:informal-ub} follows by setting $\cA^\circ(H_N,\omega)$ to be the Hermite expansion of $\bm$ up to degree $D$.

\subsection{Notations and Preliminaries}
\label{subsec:notation}

We will sometimes identify $H_N$ with the vector consisting of its disorder coefficients, i.e. $(g_{i,j})_{1\le i,j\le N}$ in the case of \eqref{eq:SK-def}, and $(g_{i_1,\ldots,i_k})_{2\le k\le \bar k, 1\le i_1,\ldots,i_k\le N}$ in the case of \eqref{eq:def-hamiltonian}.
We thus let $\sH_N=\bbR^{N^2}$ or $\bbR^{N^2+\dots+N^{\bar k}}$ be the state space for $H_N$.

For vectors in $\bbR^N$, we write $\|\cdot\|_2=\|\cdot\|$ for the Euclidean norm, and $\|\cdot\|_1$ for the $\ell^1$ norm.
We use $\plim$ to denote in-probability limits.

For $p\in [0,1]$ we say the centered Gaussian vector $(Z,Z_p)$ is a $p$-correlated pair if $Z$ and $Z_p$ have the same marginal law, and all $1$-dimensional projections $\la Z,v\ra, \la Z_p,v\ra$ have covariance matrix proportional to $\big(\begin{smallmatrix} 1& p\\ p& 1\end{smallmatrix}\big)$.
Similarly, we say a pair of Hamiltonians $(H_N,H_{N,p})$ is $p$-correlated if they are $p$-correlated as elements of $\sH_N$.

We will use the multivariate Hermite polynomials, which form an orthonormal basis for the space of $L^2$ functions on $\bbR^d$ relative to the standard Gaussian measure; see \cite[Chapter 11]{donnell2014analysis} for an introduction.
These polynomials are denoted $\He_{\vec\alpha}(\vec x)=\prod_{i=1}^d \He_{\alpha_i}(x_i)$ for $\alpha=(\alpha_1,\dots,\alpha_d)\in\bbZ_{\geq 0}^d$ and $\vec x\in\bbR^d$, and have degree $|\vec\alpha|_1\equiv\sum_{i=1}^d \alpha_i$.
For $\bg,\bg'\sim\cN(0,I_d)$ which are $p$-correlated, they satisfy:
\begin{align}
\label{eq:hermite-gradient}
    \partial_{x_i} \He_{\vec \alpha}(\vec x)
    &=
    \alpha_i \He_{\alpha_1,\dots,\alpha_{i-1},\alpha_i - 1,\alpha_{i+1},\dots,\alpha_d}(\vec x)
    \\
\label{eq:hermite-L2-norm}
    \implies
    \bbE[\|\nabla \He_{\vec \alpha}(\bg)\|_2^2]
    &=
    \|\vec\alpha\|_2^2
    =
    \sum_{i=1}^d \alpha_i^2\geq \|\vec\alpha\|_1;
    \\
\label{eq:hermite-correlation}
    \bbE[\He_{\vec\alpha}(\bg)\He_{\vec\alpha}(\bg')]
    &=
    p^{|\vec\alpha|}.
\end{align}

We next state a standard, useful bound on the smoothness of a mixed $p$-spin Hamiltonian (which also applies to the SK model, since except for the different domain it is a special case of spherical mixed $p$-spin).
For a tensor $\bA \in (\bbR^N)^{\otimes k}$, define the operator norm
\[
    \tnorm{\bA}_\op =
    \max_{\|\bsig^1\|_2,\ldots,\|\bsig^k\|_2\leq 1}
    |\la \bA, \bsig^1 \otimes \cdots \otimes \bsig^k \ra|\,.
\]

\begin{proposition}{\cite[Corollary 59]{arous2020geometry}}
    \label{prop:gradients-bounded}
    For any $\xi$, there exists a constant $c=c(\xi)>0$ and sequence of constants $(C_k)_{k\ge 0}$ independent of $N$ such that the following holds for all $N$ sufficiently large.
    Defining the convex set
    \[
    K_N=\lt\{H_N\in\sH_N~:~\norm{\nabla^k H_N(\bsig)}_{\op} \le C_k N^{1-\frac{k}{2}}~~~~\forall\, k\geq 0,\norm{\bsig}_2 \le \sqrt{N}\rt\}\subseteq \sH_N,
    \]
    we have $\bbP[H_N \in K_N] \ge 1-e^{-cN}$.
\end{proposition}


On the sphere, the relevant derivatives are the Riemannian gradient and radial derivative, and the Riemannian Hessian. These are defined as follows.
For each $\bsig \in S_N$, let $\{e_1(\bsig),\ldots,e_N(\bsig)\}$ be an orthonormal basis of $\bbR^N$ with $e_1(\bsig) = \bsig / \sqrt{N}$.
Let $\cT = \{2,\ldots,N\}$.
Let $\nabla_\cT H_N(\bsig) \in \bbR^\cT$ denote the restriction of $\nabla H_N(\bsig) \in \bbR^N$ to the space spanned by $\{e_2(\bsig),\ldots,e_N(\bsig)\}$, and $\nabla^2_{\cT \times \cT} H_N(\bsig) \in \bbR^{\cT \times \cT}$ analogously.
Define the radial and tangential derivatives
\begin{equation}
\label{eq:spherical-calculus-def}
    \partial_\rd H_N(\bsig) = \lt\la e_1(\bsig), \nabla H_N(\bsig)\rt\ra, \qquad
    \nabla_\sph H_N(\bsig) = \nabla_{\cT} H_N(\bsig).
\end{equation}
Further, define the Riemannian Hessian
\begin{equation}
\label{eq:riemannian-hessian}
    \nabla^2_{\sph} H_N(\bsig) = \nabla^2_{\cT \times \cT} H_N(\bsig) - \fr{1}{\sqrt{N}} \partial_\rd H_N(\bsig) I_{\cT \times \cT}.
\end{equation}

The next standard fact concerns the Hessian of a spherical spin glass.
It shows that aside from a constant number of potential outliers, the edge eigenvalues of the Hessian are essentially determined by the radial derivative, uniformly over $\cS_N$.

\begin{lemma}[{\cite[Lemma 3]{subag2018following}}]
\label{lem:spectrum-approx}
    For any $\eps>0$ there are $K=K(\xi,\eps)$ and $\delta=\delta(\xi,\eps)>0$ such that with probability $1-e^{-cN}$, the estimate
    \[
    \big|
    \blambda_j\big(\nabla_{\sph}^2 H_N(\bsig)\big)
    -
    2\xi''(1)
    -
    \fr{1}{\sqrt{N}}
    \partial_{\rd}H_N(\bsig)
    \big|
    \leq
    \eps
    \]
    holds simultaneously for all $K\leq j\leq \delta N$ and $\bsig\in\cS_N$.
\end{lemma}

\subsection{Stability of Lipschitz and Low Degree Algorithms}
\label{subsec:stability}

Here we establish stability properties of randomly rounded Lipschitz and low degree polynomial algorithms, which are relevant for Theorem~\ref{thm:SK-hardness-LDP}.
Throughout the below, we always take $\vec U=(U_1,\dots,U_N)\stackrel{IID}{\sim}\Unif([-1,1])$ and define $\cA$ in terms of $\cA^{\circ}$ as in \eqref{eq:round-stable-alg}.
We recall the following result, which is shown by Hermite polynomial expansions and hypercontractivity.
(In translating the statements, we use that $(1-\eps)^D\geq 1-D\eps$ and $3^4\geq 6\text{e}$, and replace their $L$ by $L^4$).
Recalling the notation of $p$-correlated Gaussian processes in Subsection~\ref{subsec:notation}, below we let $(H_N,H_{N,p})$ be a $p$-correlated pair of Hamiltonians.

\begin{proposition}[{\cite[Lemma 3.4 and Theorem 3.1]{gamarnik2020optimization}}]
\label{prop:low-degree-stable-basic}
Suppose $\cA^{\circ}:\sH_N\mapsto \bbR^N$ is a deterministic degree $D$ polynomial with $\bbE[\|\cA^{\circ}(H_N)\|^2]\leq CN$.
Then:
\begin{align}
\label{eq:low-degree-l2-stable}
\bbE[\|\cA^{\circ}(H_N)-\cA^{\circ}(H_{N,1-\eps})\|^2]
&\leq
2CD\eps N;
\\
\label{eq:low-degree-hypercontractive}
\bbP[\|\cA^{\circ}(H_N)-\cA^{\circ}(H_{N,1-\eps})\|^2\geq 2CL^2D\eps N]
&\leq
\exp\Big(-\frac{D L^{4/D}}{10}\Big),\quad\forall\, L\geq 3^D.
\end{align}
\end{proposition}

The next easy estimate shows that randomized rounding is stable with high probability.

\begin{proposition}
\label{prop:randomized-rounding-stable}
Fix $\bx,\bx'\in\bbR^N$.
Then for any $\tau>0$:
\begin{align}
\label{eq:stability-rounding-l2}
\bbE[\|\round_{\vec U}(\bx)-\round_{\vec U}(\bx')\|^2]
&\leq
2\|\bx-\bx'\|\sqrt{N},
\\
\label{eq:stability-rounding-l4}
\bbE[\|\round_{\vec U}(\bx)-\round_{\vec U}(\bx')\|^4]
&\leq
4\|\bx-\bx'\|^2 N + 8 \|\bx-\bx'\|\sqrt{N},
\\
\label{eq:stability-rounding-high-prob}
\bbP\big[\|\round_{\vec U}(\bx)-\round_{\vec U}(\bsig')\|^2
\geq
2\|\bx-\bx'\|\sqrt{N}+\tau^2 N\big]
&\leq
e^{-\tau^4 N/100}.
\end{align}
\end{proposition}

\begin{proof}
    Let $p_i = |x_i - x'_i|/2$ and
    \[
        Y_i = \ind\{x_i \ge U_i \ge x'_i\,\text{or}\, x'_i\ge U_i\ge x_i\},
    \]
    so that $Y_i \sim \Ber(p_i)$ are mutually independent and
    \begin{equation}
    \label{eq:Y-sum}
        \|\round_{\vec U}(\bx)-\round_{\vec U}(\bx')\|^2 = 4(Y_1 + \cdots + Y_N).
    \end{equation}
    Then, \eqref{eq:stability-rounding-l2} follows from Cauchy--Schwarz:
    \[
        \bbE[\|\round_{\vec U}(\bx)-\round_{\vec U}(\bx')\|^2]
        = 4 \sum_{i=1}^N p_i
        = 2\|\bx-\bx'\|_1
        \le 2\|\bx-\bx'\|\sqrt{N}.
    \]
    Similarly, \eqref{eq:stability-rounding-l4} follows because
    \begin{align*}
        \bbE[\|\round_{\vec U}(\bx)-\round_{\vec U}(\bx')\|^4]
        &= 16 \lt(\sum_{i=1}^N p_i\rt)^2 + 16 \sum_{i=1}^N (p_i - p_i^2) \\
        &\le 4 \|\bx-\bx'\|_1^2 + 8\|\bx-\bx'\|_1 \\
        &\le 4 \|\bx-\bx'\|_2^2 N + 8\|\bx-\bx'\|_2 \sqrt{N}.
    \end{align*}
    Finally \eqref{eq:stability-rounding-high-prob} follows from the (one-sided) Hoeffding's inequality, because \eqref{eq:Y-sum} represents the expected squared norm $\|\round_{\vec U}(\bx)-\round_{\vec U}(\bx')\|^2$ as a sum of independent $[0,4]$-valued random variables.
\end{proof}

Combining the above estimates, we obtain high-probability stability for low degree polynomials.
We note that the quartic scaling below (as opposed to quadratic) comes from the randomized rounding.

\begin{proposition}
\label{prop:low-degree-stable}
Suppose $\cA^{\circ}:\sH_N\mapsto \bbR^N$ is a degree $D$ polynomial with $\bbE[\|\cA^{\circ}(H_N,\omega)\|^2]\leq CN$, and let $\cA(H_N,\omega,\vec U)=\round_{\vec U}(\cA^{\circ})$.
Then for all $\eps\in (0,1)$, and for $L\geq 1$ in the former and $L\geq 2\cdot 3^D$ in the latter:
\begin{align}
\label{eq:stability-of-rounded-low-degree-l2}
\bbE[\|\cA(H_N,\vec U)-\cA(H_{N,1-\eps},\vec U)\|^4]
&\leq
12 \lt(CD\eps N^2 + \sqrt{CD\eps} N\rt),
\\
\label{eq:stability-of-rounded-low-degree-hypercontractivity}
\bbP[\|\cA(H_N,\vec U)-\cA(H_{N,1-\eps},\vec U)\|/\sqrt{N}
\geq
LC^{1/4}D^{1/4}\eps^{1/4}+ \tau
]
&\leq
\exp(-DL^{4/D}/200) + e^{-\tau^4 N/2000}
.
\end{align}
\end{proposition}

\begin{proof}
    The estimate \eqref{eq:stability-of-rounded-low-degree-l2} follows from:
    \begin{align*}
        &\bbE[\|\cA(H_N,\vec U)-\cA(H_{N,1-\eps},\vec U)\|^4] \\
        &\stackrel{\eqref{eq:stability-rounding-l4}}{\le}
        4\bbE [\|\cA^\circ(H_N) - \cA^\circ(H_{N,1-\eps})\|^2] N
        + 8 \bbE [\|\cA^\circ(H_N) - \cA^\circ(H_{N,1-\eps})\|] \sqrt{N} \\
        &\stackrel{\eqref{eq:low-degree-l2-stable}}{\le} 8CD\eps N^2 + 8\sqrt{2CD\eps} N
        \le 12 \lt(CD\eps N^2 + \sqrt{CD\eps} N\rt).
    \end{align*}
    (Note that although \eqref{eq:low-degree-l2-stable} is stated for deterministic $\cA^\circ$, the result holds just as well for randomized $\cA^\circ$ because the right-hand side of \eqref{eq:low-degree-l2-stable} is linear in $CN$.)
    For \eqref{eq:stability-of-rounded-low-degree-hypercontractivity}, with probability $1 - \exp(-DL^{4/D}/10) - e^{-\tau^4 N/100}$ we have by \eqref{eq:low-degree-hypercontractive} and \eqref{eq:stability-rounding-high-prob}:
    \begin{align*}
    \|\cA(H_N,\vec U)-\cA(H_{N,1-\eps})\|^4
    &\leq
    (2\|\cA^{\circ}(H_N)-\cA^{\circ}(H_{N,1-\eps})\|\sqrt{N} + \tau^2 N)^2
    \\
    &\leq
    8\|\cA^{\circ}(H_N)-\cA^{\circ}(H_{N,1-\eps})\|^2 N
    +
    2\tau^4 N^2
    \\
    &\leq
    16CL^2 D\eps N^2+2\tau^4 N^2
    \\
    &\leq
    (2C^{1/4}LD^{1/4}\eps^{1/4}+ 2\tau)^4 N^2
    .
    \end{align*}
    Adjusting both constants $L$ and $\tau$ by a factor of $2$ gives the desired result.
\end{proof}

Next we show a stronger estimate when $\cA^{\circ}$ is Lipschitz.

\begin{proposition}
\label{prop:lipschitz-algs-stable}
Suppose $\cA^{\circ}:\sH_N\mapsto \bbR^N$ is $L$-Lipschitz.
Then
\begin{equation}
\label{eq:stability-of-Lipschitz}
\bbP
\lt[\|\cA(H_N,\vec U)-\cA(H_{N,1-\eps},\vec U)\|/\sqrt{N}
\geq
2L^{1/4}\eps^{1/4} + \tau
\rt]
\leq
e^{-\tau^4 N/128L^2}
+
e^{-\tau^4 N/1000}
.
\end{equation}
\end{proposition}

\begin{proof}
    We first estimate the expected distance on $\cA^{\circ}$ and then apply concentration of measure.
    Decompose
    \[
    \cA^{\circ}=\sum_{k\geq 0} \cA_k^{\circ}
    \]
    where $\cA_k^{\circ}$ is a linear combination of degree $k$ Hermite polynomials (in each of the $N$ output coordinates), and let $\alpha_k=\bbE[\|\cA_k^{\circ}(H_N)\|^2]$.
    Since $\cA^{\circ}$ is assumed $L$-Lipschitz, with $J\cA^{\circ}$ the Jacobian matrix and $\|\cdot\|_F$ the Frobenius norm we have:
    \[
    \|J\cA_k^{\circ}(H_N)\|_F\leq \|J\cA_k^{\circ}(H_N)\|_{\op}\sqrt{N}\leq L\sqrt{N}.
    \]
    (The Jacobian has smaller dimension $N$, and the operator norm is at least the maximum norm of these $N$ vectors, while the Frobenius norm is at most $N^{1/2}$ times larger.)
    Applying \eqref{eq:hermite-L2-norm} and similar relations to a random $1$-dimensional projection of $\cA^{\circ}$ and averaging shows
    \begin{align*}
    \sum_k
    \alpha_k
    &\leq
    \bbE[\|\cA_k^{\circ}(H_N)\|^2]
    \leq
    N
    ,\\
    \sum_k k^2 \alpha_k
    &=
    \bbE[\|J\cA_k^{\circ}(H_N)\|^2_F]
    \leq
    L^2 N
    \\
    \implies
    \sum_k k \alpha_k
    &\leq
    L N
    .
    \end{align*}
    Applying \eqref{eq:low-degree-l2-stable} term-by-term (the increments are orthogonal by Hermite orthogonality):
    \begin{align*}
    \bbE[\|\cA^{\circ}(H_N)-\cA^{\circ}(H_{N,1-\eps})\|^2]
    \leq
    2\eps \sum_{k} k\alpha_k
    &\leq
    2L\eps N
    \\
    \bbE[\|\cA^{\circ}(H_N)-\cA^{\circ}(H_{N,1-\eps})\|]
    \leq
    \sqrt{2 L\eps N}
    \end{align*}
    Next we note that $\|\cA^{\circ}(H_N)-\cA^{\circ}(H_{N,1-\eps})\|$ is a $2L$-Lipschitz function in the disorder.
    Therefore
    \[
    \bbP[\|\cA^{\circ}(H_N)-\cA^{\circ}(H_{N,1-\eps})\|
    \geq
    \sqrt{2 L\eps N}
    +
    \tau^2\sqrt{N}]
    \leq
    e^{-\tau^4 N/8L^2}.
    \]
    Combining with \eqref{eq:stability-rounding-high-prob}, we see that
    \[
    \bbP
    \lt[\|\cA(H_N,\vec U)-\cA(H_{N,1-\eps},\vec U)\|^2
    \geq
    2(\sqrt{2 L\eps }+\tau^2)N +\tau^2 N
    \rt]
    \leq
    e^{-\tau^4 N/8L^2}
    +
    e^{-\tau^4 N/100}.
    \]
    Finally, noting that
    \[
    2(\sqrt{2 L\eps }+\tau^2) +\tau^2
    \leq
    (2L^{1/4}\eps^{1/4} + 2\tau)^2
    \]
    and substituting $\tau/2$ for $\tau$ finishes the proof.
\end{proof}

\section{Hardness of Finding Gapped States for the SK Model}

In this section we prove Theorem~\ref{thm:SK-hardness-LDP}\ref{it:lipschitz-alg-conseq} using a version of the ensemble overlap gap property.
In fact, we show hardness for algorithms which are stable
in the sense below,
with suitably high probability
$1-\punstable$.

\begin{definition}
\label{def:approx-lipschitz}
    $\cA$ is $(L,\tau,\punstable)$-stable if for all $p\in [0,1]$ and $p$-correlated $(H_N,H_{N,p})$:
    \begin{equation}
    \label{eq:stability}
    \bbP^{H_N,\omega}[\|\cA(H_N)-\cA(H_{N,p})\|/\sqrt{N}\geq
    L(1-p)^{1/4}+\tau]
    \leq \punstable.
    \end{equation}
\end{definition}

\begin{theorem}
\label{thm:SK-hardness}
For any $\gamma>0$ there exists $\delta>0$ such that the following holds.
For $0<\tau\leq \tau_0(\gamma,\delta)$ sufficiently small and any $L>0$, for $N\geq N_0(\gamma,\delta,\tau,L,\punstable)$, any $(L,\tau,\punstable)$-stable algorithm satisfies
\[
\bbP[\cA(H_N)\in S(\gamma,\delta;H_N)]\leq e^{-cN} + 2\punstable^{c}
\]
for some $c=c(\gamma,\delta,\tau,L)>0$.
\end{theorem}

Theorem~\ref{thm:SK-hardness-LDP}\ref{it:lipschitz-alg-conseq} immediately follows from Theorem~\ref{thm:SK-hardness} by applying Proposition~\ref{prop:lipschitz-algs-stable} and adjusting constants.
However Theorem~\ref{thm:SK-hardness-LDP} parts \ref{it:low-degree-strong-conseq},\ref{it:low-degree-ultimate-conseq} are \textbf{not} direct consequences since $2\punstable^{c}$ may be too large; the value of $\punstable$ guaranteed by Proposition~\ref{prop:low-degree-stable} is a decreasing function of $L$, but $c$ also depends on $L$.
The refinement needed to obtain Theorem~\ref{thm:SK-hardness-LDP}\ref{it:low-degree-strong-conseq},\ref{it:low-degree-ultimate-conseq} is explained in the next section.

To prove Theorem~\ref{thm:SK-hardness}, we consider a jointly Gaussian sequence $(H_N^{(0)},H_N^{(1)},\dots,H_N^{(K)})$ of SK Hamiltonians in which $H_N^{(j)}$ and $H_N^{(i)}$ are $(1-\eps)^{|i-j|}$-correlated for each $i,j$.
This can be generated by a (reversible) Markov chain, where we first sample $H_N^{(0)}$, and then obtain $H_N^{(i+1)}$ from $H_N^{(i)}$ by running Ornstein--Uhlenbeck flow on its disorder coefficients for time $\log \fr{1}{1-\eps}$.

We let $K$ be a large constant (independent of $N$).
We choose $N$-independent parameters in the order:
\begin{equation}
\label{eq:parameter-ordering-SK}
0<c\ll\eps=L^{-5}\ll L^{-1}\ll \tau \ll \delta \ll \gamma < 1.
\end{equation}
This means given $\gamma$, we choose $\delta$ sufficiently small, then choose $\tau$ sufficiently small, and so on (with $N$ large depending on all of them).
For simplicity we assume $\cA=\cA(H_N)$ is stable and deterministic, and aim to show it fails to find $(\gamma,\delta)$-gapped states.
(The proof for randomized algorithms $\cA(H_N,\omega)$ requires only one slight change as explained in Remark~\ref{rmk:positive-correlation-with-rounding}.)
Given this ensemble, we will carry out the following steps:

\begin{enumerate}[label=(\alph*)]
    \item
    \label{eq:sk-step-1}
    Adjacent Hamiltonians $H_N^{(i)},H_N^{(i+1)}$ have correlation close to $1$.
    By Definition~\ref{def:approx-lipschitz}, this means a stable algorithm $\cA$ must give similar outputs on these two inputs, with high probability.
    \item
    \label{eq:sk-step-2}
    For some constant $1\leq k\leq K$, conditionally on $H_N^{(i)}$, Lemma~\ref{lem:perturbed-gapped-state-separation} shows it is overwhelmingly likely that all gapped states for $H_N^{(i+1)}, \dots, H_N^{(i+k)}$ are either extremely close to $\cA(H_N^{(i)})$, or rather far away.
    Using \ref{eq:sk-step-1}, the intermediate value theorem then implies that $\cA(H_N^{(i)})$ and $\cA(H_N^{(i+k)})$ are extremely close together.
    \item
    \label{eq:sk-step-3}
    $K/k$ is not too large. Then repeatedly applying \ref{eq:sk-step-2}, the triangle inequality shows $\cA(H_N^{(0)})$ and $\cA(H_N^{(K)})$ are relatively close together.
    \item
    \label{eq:sk-step-4}
    $H_N^{(0)},H_N^{(K)}$ have correlation close to $0$.
    By Lemma~\ref{lem:SK-chaos}, this implies that conditionally on $H_N^{(0)}$, it is very unlikely for $H_N^{(K)}$ to have a gapped state close to $\cA(H_N^{(0)})$.
    \item
    \label{eq:sk-step-5}
    The previous points together imply that it is very unlikely for $\cA$ to find a gapped state on each Hamiltonian, such that the stability guarantee of Definition~\ref{def:approx-lipschitz} holds on all adjacent pairs $H_N^{(i)},H_N^{(i+1)}$.
    Using the positive correlation statement in Lemma~\ref{lem:positive-correlation}, this allows us to conclude the main result.
\end{enumerate}

In the analysis, we will mainly use the rescaled symmetrized matrix $\bG^{\sym} = (g^{\sym}_{i,j})_{1\le i,j\le N}$ defined by:
\[
    g^{\sym}_{i,j}=\frac{(g_{i,j}+g_{j,i})\cdot 1_{i\neq j}}{\sqrt{N-1}}.
\]
With this definition, the (rescaled) local fields $\wt L_i(\bsig) = L_i(\bsig)\sqrt{\frac{N}{N-1}}$ satisfy
\begin{equation}
    \label{eq:rescaled-local-field}
    \wt L_i(\bsig)
    =
    \sum_{j=1}^N g^{\sym}_{i,j}\sigma_j
    =
    (\bG^{\sym}\bsig)_i.
\end{equation}
Note that the diagonal entries $g_{i,i}$ do not appear in $\bG^{\sym}$, as they do not affect the local fields.
The slight rescaling by $\sqrt{\frac{N}{N-1}}$ makes $\wt L_i(\bsig)$ have variance exactly $2$ for each $\bsig$.
Note that except for the diagonal entries, $\bG^{\sym}$ is a Wigner matrix rescaled by $\sqrt{2 \cdot \fr{N}{N-1}}$.

\subsection{Preliminary Estimates}


We will use the following simple estimates.
First, for any $I\subseteq [N]$, we let $\bG^{\sym}_{I\times I}$ be the corresponding submatrix.
Then, for $\bsig,\bsig' \in \{-1,1\}^N$ differing only on (a subset of) coordinates $I$,
\begin{equation}
\label{eq:spectral-bound-trivial}
\|(\bG^{\sym}(\bsig-\bsig'))_I\|_1
\leq
\|(\bG^{\sym}(\bsig-\bsig'))_I\|_2
\cdot \sqrt{|I|}
\leq
2\|\bG^{\sym}_{I\times I}\|_{\op} \cdot |I|.
\end{equation}
Second, for $h(p)=p\log(1/p)+(1-p)\log\lt(\frac{1}{1-p}\rt)$ the entropy, we have from \cite[Lemma 2.2]{csiszar2004information}:
\begin{equation}
\label{eq:binomial-bound}
\binom{N}{M}
\leq
e^{N h(M/N)},\quad \forall\, 1\leq M\leq N
.
\end{equation}

\begin{proposition}
\label{prop:restricted-norm-bound}
There is a universal (large) constant $C_{\ref{prop:restricted-norm-bound}} > 0$ such that the following holds.
For all $q\in (0,1/2)$, for $N$ sufficiently large, the matrix $\bG^{\sym}$ satisfies
\[
    \sup_{|I|\leq q N}
    \|\bG^{\sym}_{I\times I}\|_{\op}
    \leq
    C_{\ref{prop:restricted-norm-bound}} \sqrt{q\log(1/q)}.
\]
\end{proposition}

\begin{proof}
    Note that $\|\bG^{\sym}_{I'\times I'}\|_{\op}\leq \|\bG^{\sym}_{I\times I}\|_{\op}$ if $I'\subseteq I$, so it suffices to consider the supremum over $|I| = \lfloor qN \rfloor$.
    For a fixed such $I$, it is classical that $\bbE \|\bG^{\sym}_{I\times I}\|_{\op}$ is of order $O(\sqrt{q})$.
    Moreover, $\|\bG^{\sym}_{I\times I}\|_{\op}$ is a $\fr{2}{\sqrt{N-1}}$-Lipschitz function of $(\hat g_{i,j})_{1\le i<j\le N}$, where $\hat g_{i,j} = \sqrt{\fr{2}{N-1}} g^\sym_{i,j}$, and therefore is subgaussian with variance proxy $O(1/N)$.
    Union bounding over the $\binom{N}{\lfloor q N\rfloor}\leq e^{O(Nq\log(1/q))}$ possible values of $I$ concludes the proof.
\end{proof}

\subsection{Disappearance of Gapped States for Weakly Correlated Disorder}

The next preparatory lemma gives a convenient way to decouple the local fields $L_i(\bsig)$, which are slightly correlated due to the symmetrization defining $\bG^{\sym}$, by conditioning on a noisy Gaussian observation.
This will be convenient for first moment computations.

\begin{lemma}
\label{lem:make-fields-independent}
For any $\bsig\in\{-1,1\}^N$, conditionally on $Z=\frac{\la \bsig,\bG^{\sym}\bsig\ra}{2\sqrt{N-1}}+g\sqrt{\frac{N-2}{N-1}}$ for independent standard Gaussian $g$, the entries of $\bG^{\sym}\bsig$ are IID with mean $Z/\sqrt{N-1}$ and variance $V_N=\frac{2(N-2)}{N-1}\in [1,2]$.
\end{lemma}

\begin{proof}
By the transitive $\bbZ_2^N$-symmetry of $\bSig_N$ and joint Gaussianity, it suffices to set $\bsig=\bone$ and analyze the conditional joint distribution of any two distinct entries of $\bG^{\sym}\bsig$.
We consider the centered Gaussian vector
\[
(X,Y,Z)=\lt((\bG^{\sym}\bone)_i,(\bG^{\sym}\bone)_j,\frac{\la \bone,\bG^{\sym}\bone\ra}{2\sqrt{N-1}}+g\sqrt{\frac{N-2}{N-1}}\rt).
\]
A routine computation shows that its covariance matrix is
\[
\begin{pmatrix}
    2 & \frac{2}{N-1} & \frac{2}{\sqrt{N-1}} \\
    \frac{2}{N-1} & 2 & \frac{2}{\sqrt{N-1}} \\
    \frac{2}{\sqrt{N-1}} & \frac{2}{\sqrt{N-1}} & 2
\end{pmatrix}
.
\]
An easy Gram--Schmidt computation then shows that conditionally on $Z$, the pair $\big(X-\frac{Z}{\sqrt{N-1}},Y-\frac{Z}{\sqrt{N-1}}\big)$ is IID centered Gaussian with variance $\frac{2(N-2)}{N-1}$.
\end{proof}

Next we establish a ``chaos'' property: for $p$-correlated Hamiltonians $H_N,H_{N,p}$ with $p$ away from $1$, there is no way to predict the location of any gapped state for $H_{N,p}$ based only on $H_N$.
In the next lemma, $\omega$ plays the role of the internal randomness of the algorithm $\cA$, i.e. $(\omega,\vec U)$ in \eqref{eq:round-stable-alg}.

\begin{lemma}
    \label{lem:SK-chaos}
    There is a universal (small) constant $c_{\ref{lem:SK-chaos}} > 0$ such that the following holds.
    Let $\gamma \in (0,1)$ and $p\in [0,1/2]$.
    Let $(H_N,H_{N,p})$ be $p$-correlated SK Hamiltonians and $\omega$ be any random variable independent of $(H_N,H_{N,p})$.
    Suppose a $\bsig \in \{-1,1\}^N$ is a function of $(H_N,\omega)$.
    Then, with probability $1-e^{-cN}$, there does not exist $\bsig' \in S(\gamma,c_{\ref{lem:SK-chaos}};H_{N,p})$ with $\|\bsig'-\bsig\| \le \sqrt{c_{\ref{lem:SK-chaos}} N}$.
\end{lemma}
\begin{proof}
    Let $\bG^{\sym}_p$ denote the symmetrized disorder matrix of $H_{N,p}$, and write
    \[
        \bG^{\sym}_p = p \bG^\sym + \sqrt{1-p^2} \wt\bG^\sym
    \]
    for $\wt\bG^\sym$ an independent copy of $\bG^\sym$.
    It is classical that $\|\bG^{\sym}\|_{\op} \le 3$ with probability $1-e^{-cN}$.
    We condition on a realization of $\bG^{\sym}$ where this holds, so that for any $\bsig' \in \{-1,+1\}^N$,
    \[
        \|\bG^{\sym} \bsig'\|^2 \le 9N.
    \]
    Thus, there are at least $N/2$ coordinates $i\in [N]$ with $|\bG^{\sym} \bsig'| \le 6$.
    Let $I_0(\bsig') \subseteq [N]$ be the set of such coordinates.

    Consider any fixed $\bsig' \in \{-1,+1\}^N$.
    We will next control the probability over $\wt\bG^{\sym}$ that $\bsig' \in S(\gamma,c_{\ref{lem:SK-chaos}};H_{N,p})$.
    Let $\wt Z(\bsig') = \fr{\la \bsig', \wt\bG^{\sym} \bsig' \ra}{2\sqrt{N-1}} + g\sqrt{\fr{N-2}{N-1}}$ for independent $g\sim \cN(0,1)$.
    By Lemma~\ref{lem:make-fields-independent} (with $V_N \in [1,2]$ defined therein),
    \begin{align*}
        \bG^{\sym}_p \bsig'
        &= p \bG^\sym \bsig' + \sqrt{1-p^2} (\wt\bG^\sym \bsig') \\
        &\stackrel{d}{=}
        p \bG^\sym \bsig' + \sqrt{1-p^2} \cdot \fr{\wt Z(\bsig')}{\sqrt{N-1}} \bone + \sqrt{V_N(1-p^2)} \bg,
    \end{align*}
    where $\bg \sim \cN(0,I_N)$.
    Then, (recalling \eqref{eq:rescaled-local-field}) $\bG^{\sym}_p \bsig' \in S(\gamma,c_{\ref{lem:SK-chaos}};H_{N,p})$ if there exists $I \subseteq [N]$ with $|I| = \lceil (1-c_{\ref{lem:SK-chaos}}) N\rceil$ such that $(\bG^{\sym}_p \bsig')_i \ge \gamma \sqrt{N/(N-1)}$ for all $i\in I$.
    Then, (conditional on $\bG^{\sym}$)
    \begin{align*}
        &\bbP(\bsig' \in S(\gamma,c_{\ref{lem:SK-chaos}};H_{N,p}))
        \le \bbP(|\wt Z(\bsig')| \ge \sqrt{N-1}) \\
        &\qquad + \sum_{|I| = \lceil (1-c_{\ref{lem:SK-chaos}}) N\rceil}
        \bbP\lt((\bG^{\sym}_p \bsig')_i \ge \gamma \sqrt{N/(N-1)} \,\text{for all}\, i\in I \big| |\wt Z(\bsig')| \le \sqrt{N-1}\rt)
    \end{align*}
    Then $\bbP(|\wt Z(\bsig')| \ge \sqrt{N-1}) \le e^{-N/8}$, while the last probability is bounded by
    \begin{align*}
        &\bbP\lt((\bG^{\sym}_p \bsig')_i \ge 0 \,\text{for all}\, i\in I \cap I_0(\bsig') \big| |\wt Z(\bsig')| \le \sqrt{N-1}\rt) \\
        &\le \prod_{i\in I \cap I_0(\bsig')} \bbP\lt(
            \sqrt{V_N(1-p^2)} g_i \ge -6p - \sqrt{1-p^2}
        \rt)
        \le \Psi\lt(-2\sqrt{3} - 1\rt)^{|I\cap I_0(\bsig')|},
    \end{align*}
    where $\Psi$ is the complementary gaussian CDF.
    If we set $c_{\ref{lem:SK-chaos}} \le 1/4$, then $|I \cap I_0(\bsig')| \ge N/4$.
    Thus,
    \[
        \bbP(\bsig' \in S(\gamma,c_{\ref{lem:SK-chaos}};H_{N,p}))
        \le e^{-N/8} + \binom{N}{\lfloor c_{\ref{lem:SK-chaos}} N\rfloor} \Psi\lt(-2\sqrt{3} - 1\rt)^{N/4}.
    \]
    Finally, the binomial coefficient and the number of $\bsig' \in \{-1,1\}^N$ such that $\|\bsig'-\bsig\| \le \sqrt{c_{\ref{lem:SK-chaos}} N}$ are both $e^{O(c_{\ref{lem:SK-chaos}} \log (1/c_{\ref{lem:SK-chaos}})) N}$.
    Union bounding over such $\bsig'$ yields
    \[
        \bbP(\exists\,\bsig' \in S(\gamma,c_{\ref{lem:SK-chaos}};H_{N,p})\,\text{with}\,\|\bsig-\bsig'\| \le \sqrt{c_{\ref{lem:SK-chaos}} N})
        \le e^{O(c_{\ref{lem:SK-chaos}} \log (1/c_{\ref{lem:SK-chaos}})) N} \lt( e^{-N/8} + \Psi\lt(-2\sqrt{3} - 1\rt)^{N/4}\rt).
    \]
    This is exponentially small if $c_{\ref{lem:SK-chaos}}$ is a sufficiently small universal constant.
\end{proof}

\subsection{Main Argument}

We now turn to the main arguments toward Theorem~\ref{thm:SK-hardness}.
The following lemma shows a form of the overlap gap property for pairs of Hamiltonians with large correlation $p\approx 1$.
Informally, if $\bsig$ is a \emph{fixed} gapped state of $H_N$, we show that with high probability all gapped states of $H_{N,p}$ are either extremely close to $\bsig$, or are rather far away.
Crucially, ``extremely close'' can be taken to be exponentially small in $(1-p)^{-1}$ (see the lower endpoint of the interval in \eqref{eq:perturbed-gapped-state-separation} below).
By taking the correlation strength between adjacent Hamiltonians in our sequence very close to $1$, this implies that a hypothetical stable algorithm finding gapped states is unable to move macroscopic distances, contradicting Lemma~\ref{lem:SK-chaos} above.

\begin{lemma}
    \label{lem:perturbed-gapped-state-separation}
    There exists a (large) universal constant $C_{\ref{lem:perturbed-gapped-state-separation}}$ such that, for any $\gamma \in (0,1)$, there exists a (small) $c_{\ref{lem:perturbed-gapped-state-separation}}(\gamma) > 0$ such that the following holds for any $\delta \in (0,\min(\gamma,1/2))$ and $p \in [1/2,1]$.
    Let $(H_N,H_{N,p})$ be a pair of $p$-correlated SK Hamiltonians.
    Let $\bsig \in S(\gamma,\delta,H_N)$ be conditionally independent of $H_{N,p}$ given $H_N$.
    Finally, let
    \begin{align*}
        r_1 = r_1(\gamma,p) &= \exp\lt(-\fr{\gamma^2}{C_{\ref{lem:perturbed-gapped-state-separation}} (1-p)}\rt), &
        r_2 = r_2(\gamma,\delta) &= \fr{C_{\ref{lem:perturbed-gapped-state-separation}} (\delta \log(1/\delta))^{1/4}}{\gamma^{1/2}}.
    \end{align*}
    With probability $1-e^{-cN}$, there does not exist $\bsig' \in S(\gamma,\delta;H_{N,p})$ such that
    \begin{equation}
        \label{eq:perturbed-gapped-state-separation}
        \|\bsig-\bsig'\|/\sqrt{N} \in [\max(r_1,r_2), c_{\ref{lem:perturbed-gapped-state-separation}}(\gamma)].
    \end{equation}
\end{lemma}
We first prove a preparatory lemma, which bounds the total $L^1$ violation of the gappedness constraint.
\begin{lemma}
    \label{lem:l1-error-sum}
    With probability $1-e^{-cN}$ over an SK Hamiltonian $H_N$, the following holds for all $\gamma \in (0,1)$, $\delta > 0$.
    For all $\bsig \in S(\gamma,\delta;H_N)$, we have
    \[
        \sum_{i=1}^N (\gamma - (\bG^\sym \bsig)_i \bsig_i)_+ \le 4\sqrt{\delta} N.
    \]
\end{lemma}
\begin{proof}
    Let $J \subseteq [N]$ be the set of coordinates where $(\bG^\sym \bsig)_i \bsig_i < \gamma$, so that by definition $|J| \le \delta N$.
    Then, proceeding similarly to \eqref{eq:spectral-bound-trivial},
    \begin{align*}
        \sum_{i=1}^N (\gamma - (\bG^\sym \bsig)_i \bsig_i)_+
        &\le \gamma \delta
        + \|(\bG^\sym \bsig)_J\|_1 \\
        &\le \gamma |J|
        + \|(\bG^\sym \bsig)_J\|_2 \sqrt{|J|}
        \le \gamma |J|
        + \|\bG^\sym\|_{\op} \sqrt{|J|N}.
    \end{align*}
    With probability $1-e^{-cN}$, $\|\bG^\sym\|_{\op} \le 3$ and the result follows.
\end{proof}

\begin{proof}[Proof of Lemma~\ref{lem:perturbed-gapped-state-separation}]
    We condition on $H_N$ satisfying Lemma~\ref{lem:l1-error-sum}, and will bound the expected number of $\bsig' \in S(\gamma,\delta;H_{N,p})$ satisfying \eqref{eq:perturbed-gapped-state-separation}.
    To this end, we fix a $\bsig' \in \{-1,1\}^N$ satisfying \eqref{eq:perturbed-gapped-state-separation} and will control the (conditional on $H_N$) probability of $\bsig' \in S(\gamma,\delta;H_{N,p}))$.

    Let $I = I(\bsig,\bsig')\subseteq [N]$ denote the set of coordinates on which $\bsig,\bsig'$ disagree.
    Let $\eta = |I| / N = 4\|\bsig-\bsig'\|^2/N$, so that \eqref{eq:perturbed-gapped-state-separation} implies
    \begin{equation}
        \label{eq:eta-range}
        \eta \in [\max(4r_1^2,4r_2^2),4c_{\ref{lem:perturbed-gapped-state-separation}}(\gamma)^2].
    \end{equation}
    Let us abbreviate $\gamma_N = \gamma \sqrt{N/(N-1)}$.
    Note that
    \begin{align}
        \notag
        \sum_{i\in I} (\gamma_N + (\bG^\sym \bsig')_i \bsig'_i)_+
        &\le \sum_{i\in I} (\gamma_N + (\bG^\sym \bsig)_i \bsig'_i)_+ + \|(\bG^{\sym}(\bsig-\bsig'))_I\|_{L^1} \\
        \notag
        &= \sum_{i\in I} (\gamma_N - (\bG^\sym \bsig)_i \bsig_i)_+ + \|(\bG^{\sym}(\bsig-\bsig'))_I\|_{L^1} \\
        \label{eq:l1-error-sum-bound}
        &\le \lt(4\sqrt{\delta} + 2 C_{\ref{prop:restricted-norm-bound}}  \sqrt{\eta^3 \log (1/\eta)}\rt) N,
    \end{align}
    where we have bounded the first term using Lemma~\ref{lem:l1-error-sum}, and the second by \eqref{eq:spectral-bound-trivial} combined with Proposition~\ref{prop:restricted-norm-bound}.
    We set $C_{\ref{lem:perturbed-gapped-state-separation}}$ large enough and $c_{\ref{lem:perturbed-gapped-state-separation}}(\gamma)$ small enough such that the bounds $\eta \ge 4r_2^2$ and $\eta \le 4c_{\ref{lem:perturbed-gapped-state-separation}}(\gamma)^2$ imply (respectively)
    \[
        4\sqrt{\delta}, 2 C_{\ref{prop:restricted-norm-bound}}  \sqrt{\eta^3 \log (1/\eta)}
        \le \eta \gamma / 4
        \le \eta \gamma_N / 4.
    \]
    Thus \eqref{eq:l1-error-sum-bound} is bounded by $\eta \gamma_N/2$.
    So, there exist at least $\eta N / 2$ coordinates $i\in I$ such that
    \[
        (\gamma_N + (\bG^\sym \bsig')_i \bsig'_i)_+ \le \gamma_N,
    \]
    i.e. $(\bG^\sym \bsig')_i \bsig'_i \le 0$.
    Let $J\subseteq I$ be a set of such coordinates, with $|J| = \lfloor \eta N / 2\rfloor$.

    Next, let $\bG^{\sym}_p$ denote the symmetrized disorder matrix of $H_{N,p}$, and write $\bG^{\sym}_p = p \bG^\sym + \sqrt{1-p^2} \wt\bG^\sym$ as above.
    Then, for $i\in J$,
    \[
        (\bG^{\sym}_p \bsig')_i \bsig'_i
        \le \sqrt{1-p^2} (\wt \bG^{\sym}_p \bsig')_i \bsig'_i.
    \]
    Since $\bsig' \in S(\gamma,\delta;H_{N,p})$, there must exist $J' \subseteq J$ of size $|J'| = |J| - \lfloor \delta N \rfloor$, such that for all $i\in J'$,
    \[
        \sqrt{1-p^2}
        (\wt \bG^\sym \bsig')_i \bsig'_i
        \ge \gamma_N.
    \]
    Let $\wt Z(\bsig') = \fr{\la \bsig', \wt\bG^{\sym} \bsig' \ra}{2\sqrt{N-1}} + g\sqrt{\fr{N-2}{N-1}}$ for independent $g\sim \cN(0,1)$.
    Then, by Lemma~\ref{lem:make-fields-independent}, (for $V_N \in [1,2]$ defind therein)
    \[
        \wt \bG^\sym \bsig' \stackrel{d}{=} \fr{\wt Z(\bsig')}{\sqrt{N-1}} \bone + \sqrt{V_N} \bg
    \]
    for $\bg \sim \cN(0,I_N)$.
    Thus,
    \begin{align}
        \notag
        &\bbP(\bsig' \in S(\gamma,\delta;H_{N,p}))
        \le \bbP(\wt Z(\bsig') > \gamma_N \sqrt{N-1} / 2) \\
        \label{eq:ogp-prob-bsig-prime-gapped}
        &\qquad + \sum_{\substack{J' \subseteq J \\ |J'| = |J| - \lfloor \delta N \rfloor}}
        \bbP\lt(
            \sqrt{1-p^2}
            (\wt \bG^\sym \bsig')_i \bsig'_i
            \ge \gamma_N\,\forall\, i\in J'\,\,
            \big| \,\,
            \wt Z(\bsig') \le \gamma_N \sqrt{N-1} / 2
        \rt)
    \end{align}
    Then $\bbP(\wt Z(\bsig') > \gamma_N \sqrt{N-1} / 2) \le e^{-\gamma^2 N/16}$, while the last probability is bounded by
    \[
        \bbP\lt(\sqrt{(1-p^2) V_N} \bg_i \ge \gamma_N / 2 \,\forall\, i\in J'\rt)
        \le \exp\lt(-\fr{|J'| \gamma_N^2}{8 (1-p^2) V_N}\rt)
        \le \exp\lt(-\fr{|J'| \gamma^2}{32 (1-p)}\rt).
    \]
    Plugging into \eqref{eq:ogp-prob-bsig-prime-gapped}, and further taking a union bound over $\bsig'$ satisfying $I(\bsig,\bsig') = \eta N$, we find
    \[
        \bbE|\{\bsig' \in S(\gamma,\delta;H_{N,p}) : I(\bsig,\bsig') = \eta N\}|
        \le \binom{N}{\eta N} \lt(
            e^{-\gamma^2N/16}
            + \binom{|J|}{\lfloor \delta N \rfloor}
            \exp\lt(-\fr{|J'| \gamma^2}{32 (1-p)}\rt)
        \rt)
    \]
    By setting $c_{\ref{lem:perturbed-gapped-state-separation}}(\gamma)$ small enough, we can ensure that for $\eta \le 4c_{\ref{lem:perturbed-gapped-state-separation}}(\gamma)^2$,
    \[
        \binom{N}{\eta N}
        e^{-\gamma^2N/16}
        \stackrel{\eqref{eq:binomial-bound}}{\le}
        \exp\lt(
            O(\eta \log (1/\eta)) N - \gamma^2 N/16
        \rt)
        \le e^{-cN}.
    \]
    By setting $C_{\ref{lem:perturbed-gapped-state-separation}}$ large enough, we can ensure that for $\eta \ge 4r_2^2$, we have $\delta \le \eta/8$.
    Thus
    \[
        |J'| = |J| - \lfloor \delta N \rfloor
        = \lfloor \eta N / 2\rfloor - \lfloor \delta N \rfloor
        \ge \eta N / 4.
    \]
    Consequently
    \[
        \binom{N}{\eta N}
        \binom{|J|}{\lfloor \delta N \rfloor}
        \exp\lt(-\fr{|J'| \gamma^2}{32 (1-p)}\rt)
        \stackrel{\eqref{eq:binomial-bound}}{\le}
        \exp\lt(
            O(\eta \log (1/\eta)) N
            + O(\delta \log (\eta/\delta)) N
            - \fr{\eta \gamma^2 N}{128(1-p)}
        \rt)
    \]
    Since $\eta \ge 4r_1^2, 4r_2^2$, setting $C_{\ref{lem:perturbed-gapped-state-separation}}$ large enough ensures this is $e^{-cN}$.
    It follows that
    \[
        \bbE|\{\bsig' \in S(\gamma,\delta;H_{N,p}) : I(\bsig,\bsig') = \eta N\}| = e^{-cN}.
    \]
    Taking a final union bound over $\eta$ in the range \eqref{eq:eta-range} (which are multiples of $1/N$) completes the proof.
\end{proof}

The next lemma gives a positive correlation property which allows us to avoid union-bounding over success events for $\cA$.
Here, the set $S$ plays the role of the event $\cA(H_N,\omega,\vec U) \in S(\gamma,\delta;H_N)$ (see Remark~\ref{rmk:positive-correlation-with-rounding} below for how to handle the auxiliary randomness $(\omega,\vec U)$).
Recall that we study the correlated ensemble $H_N=H_N^{(0)},H_N^{(1)},\dots,H_N^{(K)}$, where $H_N^{(j)}$ and $H_N^{(i)}$ are $(1-\eps)^{|i-j|}$-correlated.

\begin{lemma}
\label{lem:positive-correlation}
    Let $j = \lfloor \log_2 K \rfloor$.
    For any set $S\subseteq\sH_N$ of SK Hamiltonians with $\bbP[H_N\in S]=q$, we have
    \[
        \bbP[H_N^{(i)}\in S\quad\forall\, 0\leq i\leq K]\geq
        q^{2^{j+1}} \ge q^{4K}.
    \]
\end{lemma}

\begin{proof}
    We induct on $j$.
    For the base case $j=0$ (i.e. $K=1$), we write
    \begin{align*}
        H_N^{(0)} &= \sqrt{1-\eps} \overline H_N + \sqrt{\eps} \wt H_N^{(0)}, &
        H_N^{(1)} &= \sqrt{1-\eps} \overline H_N + \sqrt{\eps} \wt H_N^{(1)},
    \end{align*}
    for IID copies $\overline H_N, \wt H_N^{(0)}, \wt H_N^{(1)}$ of $H_N$.
    Since $H_N^{(0)}$ and $H_N^{(1)}$ are conditionally independent given $\overline H_N$,
    \[
        \bbP(H_N^{(0)},H_N^{(1)} \in S)
        = \bbE\lt[\bbP(H_N^{(0)} \in S | \overline H_N)^2 \rt]
        \ge \bbE\lt[\bbP(H_N^{(0)} \in S | \overline H_N) \rt]^2
        = q^2,
    \]
    using Jensen's inequality.
    For the inductive step, define
    \begin{align*}
        Q_j &\equiv \bbP[H_N^{(i)}\in S\quad\forall\, 0\leq i\leq 2^j], &
        Q_j(H_N^{(0)}) &\equiv \bbP[H_N^{(i)}\in S\quad\forall\, 0\leq i\leq 2^j | H_N^{(0)} ].
    \end{align*}
    By reversibility of the Hamiltonian ensemble, the sequences
    \[
        (H^{(0)},H^{(1)},\ldots,H^{(2^j)}),
        (H^{(2^{j+1})},H^{(2^{j+1}-1)},\ldots,H^{(2^j)})
    \]
    are conditionally IID given $H^{(2^j)}$.
    By the tower rule of conditional expectations and another application of Jensen's inequality,
    \[
        Q_{j+1}
        = \bbE[\bbP[H_N^{(i)}\in S\quad\forall\, 0\leq i\leq 2^{j+1} | H_N^{(2^j)}]]
        = \bbE[Q_j(H_N^{(2^j)})^2]
        \ge \bbE[Q_j(H_N^{(2^j)})]^2
        = Q_j^2.
    \]
    This completes the induction.
\end{proof}

\begin{remark}
\label{rmk:positive-correlation-with-rounding}
    Lemma~\ref{lem:positive-correlation} extends to the setting in which there is an auxiliary independent random variable $\omega \in \Omega_N$ which is shared across the correlated Hamiltonians, and the success event is a set $S\subseteq \sH_N\times\Omega_N$.
    Namely, define
    \begin{align*}
        q &\equiv \bbP[(H_N^{(0)},\omega) \in S], &
        Q_j &\equiv \bbP[(H_N^{(i)},\omega) \in S \quad \forall\, 0\le i\le 2^j], \\
        q(\omega) &\equiv \bbP[(H_N^{(0)},\omega) \in S | \omega], &
        Q_j(\omega) &\equiv \bbP[(H_N^{(i)},\omega) \in S \quad \forall\, 0\le i\le 2^j | \omega].
    \end{align*}
    Then Lemma~\ref{lem:positive-correlation} implies that for all $\omega$, $Q_j(\omega) \ge q(\omega)^{2^{j+1}}$.
    By Jensen's inequality,
    \[
        Q_j = \bbE[Q_j(\omega)]
        \ge \bbE[q(\omega)^{2^{j+1}}]
        \ge \bbE[q(\omega)]^{2^{j+1}}
        = q^{2^{j+1}}.
    \]
    Using this slightly generalized Lemma~\ref{lem:positive-correlation}, the remainder of this section goes through without modification when $\cA$ depends also on independent auxiliary randomness $(\omega,\vec U)$.
\end{remark}

\subsection{Proof of Theorem~\ref{thm:SK-hardness}}

In the below proof, we always set the parameters $\delta,\tau,L,\eps$ in the order \eqref{eq:parameter-ordering-SK}, i.e. $\delta$ will be taken sufficiently small in terms of $\gamma$, and so on.
Note that for fixed $\gamma$, the value $r_2(\gamma,\delta)$ defined in Lemma~\ref{lem:perturbed-gapped-state-separation} tends to $0$ as $\delta \to 0$, so by taking $\delta$ small enough we can ensure $r_2(\gamma,\delta) \le c_{\ref{lem:perturbed-gapped-state-separation}}(\gamma)/2$ and $r_2(\gamma,\delta) \le r_1(\gamma,p=1/2)$.

Since $r_1(\gamma,p)$ is decreasing in $p$, there exists a unique $p_c = p_c(\gamma,\delta) \in [1/2,1]$ such that $r_1(\gamma,p_c) = r_2(\gamma,\delta)$.
Further, let $k = \lfloor \log p_c / \log(1-\eps) \rfloor$, so that $k$ is the largest integer such that $(1-\eps)^k \ge p_c$.
It is clear that
\begin{align*}
    1-p_c &\asymp \fr{\gamma^2}{\log(1/\delta)}, &
    k &\asymp \fr{\gamma^2}{\eps \log(1/\delta)}.
\end{align*}
Finally let $K = \lceil 1/\eps \rceil$.
Consider an ensemble of Hamiltonians $H_N^{(0)},\ldots,H_N^{(K)}$ where $H_N^{(i)},H_N^{(j)}$ are $(1-\eps)^{|i-j|}$-correlated, and let $(\omega,\vec U)$ be a sample of the auxiliary randomness of $\cA$.
For $0\le i\le K$, let $\bsig^{(i)} = \cA(H_N^{(i)},\omega,\vec U)$, and note that the $(\omega,\vec U)$ are shared across all $H^{(i)}$.
We now define several events:
\begin{align}
    \label{eq:ssolve-v1}
    \Ssolve &= \lt\{
        \bsig^{(i)} \in S(\gamma,\delta; H_N^{(i)}) \qquad \forall\, 0\le i\le K
    \rt\}, \\
    \Sstable &= \lt\{
        \|\bsig^{(i+1)} - \bsig^{(i)}\|/\sqrt{N}
        \leq
        L\eps^{1/4} + \tau,\quad\forall\, 0\leq 1\leq K-1
    \rt\}, \\
    \Schaos &= \lt\{
        \text{there does not exist $\bsig'\in S(\gamma,\delta; H_N^{(K)})$ such that $\|\bsig'-\bsig^{(0)}\| \le \sqrt{c_{\ref{lem:SK-chaos}} N}$}
    \rt\} \\
    \label{eq:sogp-v1}
    \Sogp &= \lt\{
        \begin{array}{l}
            \text{for all $0\le i\le j\le K$ with $|i-j| \le k$, there does not exist} \\
            \text{$\bsig' \in S(\gamma,\delta; H_N^{(j)})$ such that $\|\bsig'-\bsig^{(i)}\| / \sqrt{N} \in [r_2(\gamma,\delta),c_{\ref{lem:perturbed-gapped-state-separation}}(\gamma)]$}
        \end{array}
    \rt\}
\end{align}

\begin{lemma}
    \label{lem:no-intersection}
    We have $\Ssolve \cap \Sstable \cap \Schaos \cap \Sogp = \emptyset$.
\end{lemma}
\begin{proof}
    Suppose $\Ssolve \cap \Sstable \cap \Sogp$ holds.
    We will show this contradicts $\Schaos$.

    First, consider any $0\le i\le j\le K$ with $|i-j| \le k$.
    We will show that
    \begin{equation}
        \label{eq:low-movement-in-k}
        \|\bsig^{(i)} - \bsig^{(j)}\| / \sqrt{N} \le r_2(\gamma,\delta).
    \end{equation}
    By $\Ssolve$ and $\Sogp$, for any $\ell$ such that $i\le \ell \le j$,
    \[
        \|\bsig^{(i)} - \bsig^{(\ell)}\| / \sqrt{N}
        \not\in [c_{\ref{lem:perturbed-gapped-state-separation}}(\gamma) / 2 , c_{\ref{lem:perturbed-gapped-state-separation}}(\gamma)]
        \subseteq
        [r_2(\gamma,\delta),c_{\ref{lem:perturbed-gapped-state-separation}}(\gamma)],
    \]
    where we recall that we set $r_2(\gamma,\delta)\le c_{\ref{lem:perturbed-gapped-state-separation}}(\gamma)/2$.
    We set $\tau,L,\eps=L^{-5}$ such that $L\eps^{1/4} + \tau \le c_{\ref{lem:perturbed-gapped-state-separation}}(\gamma)/2$.
    Then $\|\bsig^{(i)} - \bsig^{(\ell)}\| / \sqrt{N}$ and $\|\bsig^{(i)} - \bsig^{(\ell+1)}\| / \sqrt{N}$ can never be on opposite sides of $[c_{\ref{lem:perturbed-gapped-state-separation}}(\gamma)/2,c_{\ref{lem:perturbed-gapped-state-separation}}(\gamma)]$, proving \eqref{eq:low-movement-in-k}.
    Let $T = \lceil K/k\rceil \asymp \gamma^2 / \log(1/\delta)$.
    Then
    \[
        \|\bsig^{(0)} - \bsig^{(K)}\|
        \le \sum_{t=0}^{T-2} \|\bsig^{(tk)} - \bsig^{((t+1)k)}\|
        + \|\bsig^{((T-1)k)} - \bsig^{(K)}\|
        \le Tr_2(\gamma,\delta) \sqrt{N}
        \asymp \fr{\gamma^{3/2} \delta^{1/4}}{\log(1/\delta)^{3/4}} \sqrt{N}.
    \]
    For sufficiently small $\delta$ (depending on $\gamma$), this is less than $c_{\ref{lem:SK-chaos}} \sqrt{N}$.
    Then $\bsig^{(K)} \in S(\gamma,\delta; H_N^{(K)})$ and $\|\bsig{(K)}-\bsig^{(0)}\| \le \sqrt{c_{\ref{lem:SK-chaos}} N}$, contradicting $\Schaos$.
\end{proof}

\begin{proof}[Proof of Theorem~\ref{thm:SK-hardness}]
    Let $q = \bbP(\cA(H_N,\omega,\vec U) \in S(\gamma,\delta; H_N))$ be the probability that $\cA$ succeeds on one Hamiltonian.
    By Lemma~\ref{lem:positive-correlation} (in combination with Remark~\ref{rmk:positive-correlation-with-rounding}), we have $\bbP[\Ssolve] \ge q^{4K}$.

    Recall $\punstable$ from Definition~\ref{def:approx-lipschitz}.
    By a union bound, $\bbP[\Sstable] \ge 1-K\punstable$.

    The correlation between $H_N^{(0)}$ and $H_N^{(K)}$ is $(1-\eps)^K = (1-\eps)^{\lceil 1/\eps \rceil} = e^{-1} + O(\eps) \le 1/2$.
    Thus, Lemma~\ref{lem:SK-chaos} implies $\bbP[\Schaos] \ge 1-e^{-cN}$.

    For any $|i-j| \le k$, the correlation between $H_N^{(i)}$ and $H_N^{(j)}$ is $p_{i,j} \equiv (1-\eps)^{|i-j|} \ge (1-\eps)^k \ge p_c$.
    Thus $r_1(\gamma,p_{i,j}) \le r_1(\gamma,p_c) = r_2(\gamma,\delta)$.
    By Lemma~\ref{lem:perturbed-gapped-state-separation}, the event in $\Sogp$ for this particular $(i,j)$ holds with probability $1-e^{-cN}$.
    By a union bound over these $(i,j)$, $\bbP[\Sogp] \ge 1-e^{-cN}$.
    However, by Lemma~\ref{lem:no-intersection},
    \[
        \bbP[\Ssolve] + \bbP[\Sstable] + \bbP[\Schaos] + \bbP[\Sogp] \le 3.
    \]
    It follows that
    \[
    q
    \le
    (e^{-cN} + K\punstable)^{1/4K}
    \leq
    e^{-cN/4K}+K^{1/4K}\punstable^{1/4K}
    \leq
    e^{-cN/4K}+2\punstable^{1/4K}
    .
    \]
    This implies the result after adjusting $c$.
\end{proof}

\begin{remark}
\label{rem:BOGP}
As previously alluded to, the above argument can be phrased using a version of the branching overlap gap property.
This gives a ``global landscape obstruction'' that does not directly rely on the fact that $\cA(H_N)$ does not depend on the other Hamiltonians in the ensemble.
Namely, one can fix another large constant $W$ and consider not $K+1$ interpolated Hamiltonians but $1+W+\dots+W^K$ of them, indexed by the rooted tree $\bbT\equiv [W]^0 \cup [W]^1 \cup\dots\cup [W]^K$, with root Hamiltonian denoted $H_N^{\emptyset}$.
For each $1\leq k\leq K$ and $v\in [W]^k$, given $H_N^{(v)}$ we let $H_N^{(v1)},\dots,H_N^{(vK)}$ be conditionally independent Hamiltonians which are $1-\eps$-correlated with $H_N^{(v)}$.
Using Jensen's inequality as in \cite[Proof of Proposition 3.6(a)]{huang2021tight}, one can extend Lemma~\ref{lem:positive-correlation} to this ensemble, showing that if $\cA$ finds a gapped state with probability $q$, then it finds a gapped state on all these Hamiltonians simultaneously with probability at least $q^{4K|\bbT|}$.
From this it suffices to show that with exponentially good probability over the disorder, no ensemble of gapped states $(\bsig^{(v)})_{v\in\bbT}$ is such that the stability estimate of Definition~\ref{def:approx-lipschitz} applies to every pair $(v,v')$ such that $v$ is an ancestor of $v'$ in $\bbT$.

In the proof of Theorem~\ref{thm:SK-hardness}, we used both Lemma~\ref{lem:SK-chaos} and Lemma~\ref{lem:perturbed-gapped-state-separation} which crucially relied on the fact that $\cA(H_N^{(k)})$ depends only on $H_N^{(k)}$, and held with probability $1-e^{-cN}$.
We take $W$ large depending on this $c$.
Let us say $w\in W$ is \emph{nice} for some $\bsig\in S(\gamma,\delta;H_N^{\emptyset})$ if Lemma~\ref{lem:SK-chaos} holds for $H_N^{\emptyset}$ and $H_N^{(v)}$ for \emph{every} leaf $v\in [W]^K$ of $\bbT$ with first digit $w$.
Then for each $\bsig\in S(\gamma,\delta;H_N^{\emptyset})$, the probability that fewer than $2W/3$ digits $w$ are nice for $\bsig$ is at most
\[
2^{N+W} e^{-cWN/3}\leq 2^{-N}.
\]
Therefore with high probability over the ensemble, all $\bsig\in S(\gamma,\delta;H_N^{\emptyset})$ have at least $2W/3$ nice digits $w$.

Next let us say a correlated pair $(H_N,H_N')$ is \emph{good} with respect to $\bsig\in\bSig_N$ if the conclusion of Lemma~\ref{lem:perturbed-gapped-state-separation} holds for it in both directions, with $\bsig$ replacing $\cA(H_N)$ in one case and $\cA(H_N')$ in the other.
Because $W$ is large depending on $c$, we similarly find that the following holds with probability at least $1-e^{-N}$: for every $v\in [W]^0 \cup \dots\cup [W]^{K-1}$ and $\bsig\in\bSig_N$, there exist at least $2W/3$ digits $i\in [W]$ such that $(H_N^{(v)},H_N^{(u)})$ is a good pair with respect to $\bsig$, for all $u$ with initial substring $vi$ (i.e.\ for $H_N^{(vi)}$ and all its descendants).
This implies that for any given sequence $(\bsig^{\emptyset},\dots,\bsig^{(K)})\in \bSig_N^{K+1}$, there exists a root-to-leaf path in $\bbT$ to some $v_*\in [W]^K$ such that all pairs of Hamiltonians along the path are good relative to their corresponding $\bsig^{(i)}$, and such that Lemma~\ref{lem:SK-chaos} holds for $(H_N^{\emptyset},H_N^{(v_*)},\bsig^{(0)})$ and $(H_N^{(v_*)},H_N^{\emptyset},\bsig^{(K)})$.
Our preceding proof of Theorem~\ref{thm:SK-hardness} can then be applied along said path to obtain a contradiction.

We finally mention that despite the close similarity, the correlation structure in the branching ensemble of Hamiltonians described above differs from the one used in \cite{huang2021tight,huang2023algorithmic}, in the same way that an Ornstein--Uhlenbeck process differs from a Brownian motion.
\end{remark}

\section{Extension to Strong Low Degree Hardness}
\label{sec:low-degree}

In this section we prove Theorem~\ref{thm:SK-hardness-LDP}\ref{it:low-degree-strong-conseq}, and explain how to improve several prior applications of the ensemble OGP to strong low degree hardness.
The central tool is an upgrade of Lemma~\ref{lem:positive-correlation} which simultaneously handles success and stability events, and thus avoids any need for $K$-fold union bounds along the ensemble.
(By contrast, we used the simple estimate $\bbP[\Sstable] \ge 1-K\punstable$ in proving Theorem~\ref{thm:SK-hardness}.)
Like before, let $H_N^{(0)},\ldots,H_N^{(K)}$ be an ensemble of Hamiltonians where $H_N^{(i)},H_N^{(j)}$ are $(1-\eps)^{|i-j|}$-correlated, and $(\omega,\vec U)$ be the auxiliary randomness of $\cA$.
Let
\begin{align*}
    \Ssolve(i)
    &\equiv
    \{\cA(H_N^{(i)},\omega,\vec U) \in S(\gamma,\delta;H_N^{(i)})\}, \\
    \Sstab(i)
    &\equiv
    \{\|\cA(H_N^{(i)},\omega,\vec U)-\cA(H_N^{(i+1)},\omega,\vec U)\|/\sqrt{N}
    \leq
    L\eps^{1/4}+\tau\}.
\end{align*}
Then define the ``success and stability'' event
\[
    \Sall(K) = \Big(\bigcap_{i=0}^K \Ssolve(i)\Big) \cap \Big(\bigcap_{i=0}^{K-1} \Sstab(i)\Big).
\]
The following elementary correlation inequality is the crux of our improved ensemble OGP.
\begin{lemma}
    \label{lem:grand-correlation}
    Let $\psolve=\bbP[\Ssolve(0)]$ and $\punstable = 1 - \bbP[\Sstab(0)]$.
    Then, for $j = \lfloor \log_2 K \rfloor$,
    \[
        \bbP[\Sall(K)]
        \geq (\psolve^2-\punstable)_+^{2^j}
        \geq (\psolve^2-\punstable)_+^{2K}.
    \]
\end{lemma}
\begin{proof}
    We first prove the statement for a deterministic function $\cA(H_N,\omega,\vec U) = \cA(H_N)$.
    The auxiliary randomness $(\omega,\vec U)$ can be addressed similarly to Remark~\ref{rmk:positive-correlation-with-rounding}.
    We will again induct on $j$.
    For the base case $j=0$ (i.e. $K=1$), we recall from the base case in Lemma~\ref{lem:positive-correlation} that
    \[
        \bbP(\Ssolve(0) \cap \Ssolve(1)) \ge \psolve^2.
    \]
    It follows that
    \[
        \bbP[\Sall(1)] \ge \psolve^2 - \punstable.
    \]
    We now proceed as in the proof of Lemma~\ref{lem:positive-correlation}.
    Let
    \begin{align*}
        Q_j &\equiv \bbP[\Ssolve(i)\,\,\forall\, 0\leq i\leq 2^j\,\text{and}\,\Sstab(i)\,\,\forall 0\le i\le 2^j-1], \\
        Q_j(H_N^{(0)}) &\equiv \bbP[\Ssolve(i)\,\,\forall\, 0\leq i\leq 2^j\,\text{and}\,\Sstab(i)\,\,\forall 0\le i\le 2^j-1 | H_N^{(0)} ].
    \end{align*}
    Recall that the sequences
    \[
        (H^{(0)},H^{(1)},\ldots,H^{(2^j)}),
        \quad
        (H^{(2^{j+1})},H^{(2^{j+1}-1)},\ldots,H^{(2^j)})
    \]
    are conditionally IID given $H_N^{(2^j)}$.
    It follows identically to before that
    \[
        Q_{j+1}
        = \bbE[\bbP[H_N^{(i)}\in S\quad\forall\, 0\leq i\leq 2^{j+1} | H_N^{(2^j)}]]
        = \bbE[Q_j(H_N^{(2^j)})^2]
        \ge \bbE[Q_j(H_N^{(2^j)})]^2
        = Q_j^2,
    \]
    which completes the induction.
    Finally, we address the possibility of auxiliary randomness $(\omega,\vec U)$, which we will rename to $\omega$ for this proof.
    Define
    \begin{align*}
        \psolve(\omega) &= \bbP[\Ssolve(0) | \omega], &
        \punstable(\omega) = 1 - \bbP[\Sstab(0) | \omega].
    \end{align*}
    The above argument shows that, conditional on any $\omega$,
    \[
        \bbP[\Sall(K) | \omega] \ge (\psolve(\omega)^2 - \punstable(\omega))_+^{2K}.
    \]
    By applying Jensen's inequality twice, to the functions $x\mapsto x_+^{2K}$ and $x\mapsto x^2$, we obtain:
    \begin{align*}
        \bbP[\Sall(K)]
        &= \bbE[\bbP[\Sall(K) | \omega]]
        \ge \bbE[(\psolve(\omega)^2 - \punstable(\omega))_+^{2K}] \\
        &\ge (\bbE[\psolve(\omega)^2] - \bbE[\punstable(\omega)])_+^{2K}
        \ge (\bbE[\psolve(\omega)]^2 - \bbE[\punstable(\omega)])_+^{2K} \\
        &= (\psolve^2 - \punstable)_+^{2K}. \qedhere
    \end{align*}
\end{proof}

\subsection{Strong Low Degree Hardness for Gapped States in the SK Model}
\label{subsec:low-degree-main-proof}

Here we prove strong low degree hardness of finding gapped states up to degree $D\leq o(N)$.
We first consider $D\leq \frac{\log N}{11}$ and show the success probability of such algorithms is at most stretched exponential.

\begin{proof}[Proof of Theorem~\ref{thm:SK-hardness-LDP}\ref{it:low-degree-strong-conseq}]
    We now recall the estimate \eqref{eq:stability-of-rounded-low-degree-hypercontractivity} from Proposition~\ref{prop:low-degree-stable}, which ensures that for $L\geq 2 \cdot 3^D$, the function $\cA=\cA(H_N,\omega,\vec U)$ is stable with parameters
    \[
        (L' \equiv LC^{1/4}D^{1/4},\tau,e^{-\Omega(DL^{4/D})}+e^{-\Omega(\tau^4 N)}).
    \]
    We will maintain the same (dimension-free) values of $(\gamma,\delta)$, and further set $\tau = c_{\ref{lem:perturbed-gapped-state-separation}}(\gamma)/4$.
    The remaining parameters will grow with $N$, as follows:
    \begin{align*}
        L &= N^{0.1} \ge 2 \cdot 3^D, &
        \eps &= N^{-0.5} = L^{-5}, &
        K &= \lceil 1/\eps \rceil.
    \end{align*}
    (The first inequality uses that $D \le \fr{1}{11} \log N$.)
    As before, let $k$ be the largest number such that $(1-\eps)^k \ge p_c(\gamma,\delta)$, so that $k \asymp \fr{\gamma^2}{\eps \log (1/\delta)}$ as above.
    We define the events $\Ssolve,\Sstab,\Schaos,\Sogp$ identically as in \eqref{eq:ssolve-v1} through \eqref{eq:sogp-v1}, but with these new parameters, and with $L'$ in place of $L$ in $\Sstab$.

    In the above proof of Lemma~\ref{lem:no-intersection}, the estimate $L'\eps^{1/4} + \tau \le c_{\ref{lem:perturbed-gapped-state-separation}}(\gamma)/2$ (with $L'$ replacing $L$ from the proof above) remains true under the new parameters.
    Moreover, although both $K$ and $k$ now scale with $N$ (through $\eps$) we nonetheless have $K/k \asymp \gamma^2 / \log(1/\delta)$.
    The remaining steps of the proof hold verbatim, so Lemma~\ref{lem:no-intersection} continues to hold.

    We also have $\bbP(\Schaos), \bbP(\Sogp) = 1-e^{-cN}$ identically to the proof of Theorem~\ref{thm:SK-hardness}.
    (The bound on $\bbP(\Sogp)$ now involves a union bound over polynomially many choices of $(i,j)$.)
    If we let $\psolve = \bbP(\cA(H_N,\omega,\vec U) \in S(\gamma,\delta;H_N))$ and $\punstable = e^{-\Omega(DL^{4/D})}+e^{-\Omega(\tau^4 N)}$, then Lemma~\ref{lem:grand-correlation} gives
    \[
        \bbP(\Ssolve \cap \Sstab)
        \ge (\psolve^2 - \punstable)_+^{2K}.
    \]
    Combined with Lemma~\ref{lem:no-intersection} this implies
    \begin{equation}
        \label{eq:strong-low-deg-hardness}
        (\psolve^2 - \punstable)_+^{2K} \le e^{-cN}, \qquad
        \implies \qquad \psolve \le \punstable^{1/2} + e^{-cN/4K}.
    \end{equation}
    The right-hand side is bounded by $e^{-\Omega(DN^{1/3D})}$.
\end{proof}

\begin{proof}[Proof of Theorem~\ref{thm:SK-hardness-LDP}\ref{it:low-degree-ultimate-conseq}]
    We proceed similarly but with slightly different parameters, using \eqref{eq:stability-of-rounded-low-degree-l2} from Proposition~\ref{prop:low-degree-stable} instead of \eqref{eq:stability-of-rounded-low-degree-hypercontractivity}.
    Provided $CD\eps N^2 = \omega_N(1)$, we have
    \[
        12 \lt(CD\eps N^2 + \sqrt{CD\eps} N\rt) \le 16 CD\eps N^2,
    \]
    and therefore \eqref{eq:stability-of-rounded-low-degree-l2} ensures that $\cA$ is stable with parameters
    \[
        (L' \equiv 2LC^{1/4}D^{1/4},0,L^{-4})
    \]
    for any $L\geq 1$.
    We maintain the same dimension-free $(\gamma,\delta)$ as above and now set $\tau = 0$.
    For the remaining parameters, we set
    \begin{align*}
        \eps &= \fr{\log (N/D)}{N} = \omega(1/N), &
        L &= \fr{c_{\ref{lem:perturbed-gapped-state-separation}}(\gamma)}{4C^{1/4}D^{1/4}\eps^{1/4}}, &
        K &= \lceil 1/\eps \rceil,
    \end{align*}
    which clearly satisfy the condition $CD\eps N^2 = \omega_N(1)$ required above.
    As before, let $k$ be the largest integer such that $(1-\eps)^k \ge p_c(\gamma,\delta)$.

    Define $\Ssolve,\Sstab,\Schaos,\Sogp$ identically to the above proof of part \ref{it:low-degree-strong-conseq}.
    Identically to that proof, Lemma~\ref{lem:no-intersection} continues to hold and $\bbP(\Schaos), \bbP(\Sogp) = 1-e^{-cN}$.
    Thus \eqref{eq:strong-low-deg-hardness} remains true.
    Note that
    \[
        \punstable \le L^{-4} \asymp D\eps = D \log (N/D) / N \le \sqrt{D/N},
    \]
    while
    \[
        e^{-cN/4K} = e^{-c' N\eps} = (D/N)^{c'}.
    \]
    Thus,
    \[
        \psolve
        \le \punstable^{1/2} + e^{-cN/4K}
        \le (D/N)^{c'} = o_N(1). \qedhere
    \]
\end{proof}

\subsection{Strong Low Degree Hardness for Other Mean-Field Models}
\label{subsec:prior-ogp-improve}

In this subsection and the next, we improve several existing results proved via the ensemble overlap gap property to strong low degree hardness.
We emphasize that in all of the cases below, we make essentially no changes to the \emph{landscape} properties already established, but only modify the way they are used to deduce algorithmic hardness.

For concreteness, we continue to use the randomized rounding scheme defined earlier to convert the outputs of low degree polynomials into $\bSig_N$-valued functions.
This is slightly different from what is done in the papers below, but does not affect the results.
Additionally, several of the works we adapt prove hardness for solving approximate versions of their respective problems, i.e. allowing a small constant fraction of the problem's constraints to be violated; this is analogous to the error tolerance parameter $\delta$ we have used and makes the obstructions stronger.
Our strengthenings also hold for these approximate results, but we suppress these for the sake of simplicity.

\subsubsection{Mean Field Spin Glasses}

For our first OGP improvement, we recall the pure $k$-spin Hamiltonian which is a special case of \eqref{eq:def-hamiltonian}:
\[
    H_N(\bsig)
    =
    N^{-\frac{k-1}{2}}
    \la \bG_N^{(k)}, \bsig^{\otimes k} \ra
    =
    N^{-\frac{k-1}{2}}
    \sum_{i_1,\dots,i_k=1}^N g_{i_1,\dots,i_k}\sigma_{i_1}\dots\sigma_{i_k}
    .
\]
The asymptotic maximum value is given by the Parisi formula \cite{parisi1979infinite,talagrand2006parisi,panchenko2013parisi}, and is denoted
\[
    \GS(k)
    =
    \lim_{N\to\infty}
    \bbE\lt[\max_{\bsig\in\bSig_N}H_N(\bsig)/N\rt].
\]

The next result shows that degree $o(N)$ polynomials have $o_N(1)$ probability to attain this asymptotic maximum value for $k\geq 4$ which are even.
(The result can be extended to all even mixed $p$-spin models exhibiting an overlap gap at zero temperature, by combining the below with the discussion from \cite[Section 6]{sellke2021optimizing}.)

\begin{corollary}{(Based on \cite[Theorem 2.4]{gamarnik2020optimization})}
\label{cor:GJW}
    Consider an Ising pure $k$-spin model with $k\geq 4$ even.
    There exists $\delta=\delta(k)>0$ and $c'>0$ such that the following holds.
    If $\cA^{\circ}$ is a
    degree $D\leq o(N)$ polynomial with $\bbE^{H_N}[\|\cA^{\circ}(H_N,\omega)\|^2]\le CN$, then for $\cA$ as in \eqref{eq:round-stable-alg},
    \[
    \bbP[H_N(\cA(H_N,\omega,\vec U))/N\geq \GS(k)-\delta]
    \leq
    (D/N)^{c'}
    \leq
    o_N(1).
    \]
\end{corollary}
The main landscape obstruction we will use is the following.
\begin{lemma}[{\cite[Theorem 2.10]{gamarnik2020optimization}}]
    \label{lem:kspin-ogp}
    There exist $N$-independent parameters $(\delta,\beta,\eta)$ such that uniformly in $p\in [0,1]$, the following holds.
    If $(H_N,H_{N,p})$ are $p$-correlated $k$-spin Hamiltonians, with probability $1-e^{-cN}$ there do not exist $\bsig,\bsig' \in \bSig_N$ such that
    \[
        H_N(\bsig),H_{N,p}(\bsig')\geq (\GS(k)-\delta)N,
    \]
    and $\|\bsig-\bsig'\|/\sqrt{N}\in [\beta-\eta,\beta]$.
\end{lemma}
\noindent (Note that our parameters are slightly different from theirs, as we work with Euclidean distances rather than overlaps.)
We consider an ensemble of ($k$-spin) Hamiltonians $H_N^{(0)},H_N^{(1)},\dots,H_N^{(K)}$ such that $H_N^{(i)},H_N^{(j)}$ are $(1-\eps)^{|i-j|}$-correlated.
Let $b$ be a small constant depending on $k,\delta,\beta,\eta$, which we will set below.
Let
\begin{align}
    \label{eq:eps-L-K-setting}
    \eps &= \frac{\log (N/D)}{N}, &
    L &= \fr{b}{\eps^{1/4}}, &
    K &= \lt\lceil \fr{1}{b\eps} \rt\rceil.
\end{align}
Let $\bsig^{(i)} = \cA(H_N^{(i)},\omega,\vec U)$, and define the events:
\begin{align*}
    \Ssolve &= \lt\{
        H_N^{(i)}(\bsig^{(i)}) \ge (\GS(k) - \delta)N
        \qquad \forall 0\le i\le K
    \rt\}, \\
    \Sstab &= \lt\{
        \|\bsig^{(i)} - \bsig^{(i+1)}\|/\sqrt{N} \leq L\eps^{1/4}
        \qquad \forall 0\le i\le K-1
    \rt\}, \\
    \Schaos &= \lt\{
        \begin{array}{l}
        \text{There do not exist $\bsig,\bsig' \in \bSig_N$ such that} \\
        \text{$H_N^{(0)}(\bsig), H_N^{(K)}(\bsig') \ge N(\GS(k) - \delta)$ and $\|\bsig-\bsig'\|/\sqrt{N} \le \beta - \eta$}
        \end{array}
    \rt\}, \\
    \Sogp &= \lt\{
        \text{For all $0\le i,j\le K$, the conclusion of Lemma~\ref{lem:kspin-ogp} holds for $H_N^{(i)}, H_N^{(j)}$}
    \rt\}.
\end{align*}
\begin{lemma}
    \label{lem:pspin-chaos-ogp-bounds}
    We have $\bbP(\Schaos), \bbP(\Sogp) = 1-e^{-cN}$.
\end{lemma}
\begin{proof}
    The bound on $\bbP(\Sogp)$ follows from union bounding Lemma~\ref{lem:kspin-ogp} over all $i,j$.
    For $\bbP(\Schaos)$, note that $(H_N^{(0)},H_N^{(K)})$ are $p$-correlated for
    \[
        p = (1-\eps)^{K} \le e^{-1/b}.
    \]
    We write $H_N^{(K)} = p H_N^{(0)} + \sqrt{1-p^2} \wt H_N$, where $\wt H_N$ is independent of $H_N^{(0)}$.
    By \cite[Theorem 2.10, Item 2]{gamarnik2020optimization} (applied to $H_N^{(0)}$ and $\wt H_N$), with probability $1-e^{-cN}$, there do not exist $\bsig,\bsig' \in \bSig_N$ such that $H_N^{(0)}(\bsig), \wt H_N(\bsig') \ge N(\GS(k) - \delta)$ and $\|\bsig-\bsig'\|/\sqrt{N} \le \beta - \eta$.
    Furthermore by Proposition~\ref{prop:gradients-bounded} with $k=0$, we have $\sup_{\bsig \in \bSig_N} H_N^{(0)}(\bsig) / N \le C$ with probability $1-e^{-cN}$, for some $C = O(1)$.
    Setting $b$ large enough that $Cp \le \delta/2$ implies the result (after adjusting $\delta$).
\end{proof}

\begin{proof}[Proof of Corollary~\ref{cor:GJW}]
    Like above, let $\psolve = \bbP[H_N^{(0)}(\bsig^{(0)})/N\ge \GS(k)-\delta]$, and $\punstable = \bbP[\|\bsig^{(0)} - \bsig^{(1)}\| / \sqrt{N} > L\eps^{1/4}]$.
    Note that $L\eps^{1/4} = b$, and we can set $b$ small enough that $b\le \eta$.
    Then, as argued in \cite[Proof of Theorem 2.4]{gamarnik2020optimization}, we have $\Ssolve \cap \Sstab \cap \Schaos \cap \Sogp = \emptyset$.
    Lemma~\ref{lem:grand-correlation} implies
    \[
        \bbP(\Ssolve \cap \Sstab) \ge (\psolve^2 - \punstable)^{2K}.
    \]
    Combining with Lemma~\ref{lem:pspin-chaos-ogp-bounds} shows
    \[
        (\psolve^2 - \punstable)^{2K} \le e^{-cN} \qquad \implies \qquad
        \psolve \le \punstable^{1/2} + e^{-cN/4K}.
    \]
    These two terms are bounded identically to the proof of Theorem~\ref{thm:SK-hardness-LDP}\ref{it:low-degree-ultimate-conseq}.
    For the choice of $\eps$ above, the right-hand side of \eqref{eq:stability-of-rounded-low-degree-l2} in Proposition~\ref{prop:low-degree-stable} is bounded by $16 CD\eps N^2$.
    Applying this estimate,
    \begin{equation}
        \label{eq:pspin-punstable}
        \punstable \le \fr{16CD\eps}{b^4}
        \asymp \fr{D\log(N/D)}{N}
        \le \sqrt{D/N}.
    \end{equation}
    Moreover, (for dimension-free but varying $c$),
    \begin{equation}
        \label{eq:pspin-exponential}
        e^{-cN/4K} = e^{-cN\eps} = e^{-c\log (N/D)} = (D/N)^c.
    \end{equation}
    It follows that $\psolve \le (D/N)^{c'} = o_N(1)$.
\end{proof}

\subsubsection{Ising perceptron models}

Next we turn to the Ising perceptron.
Here we are given $M$ IID Gaussian vectors $\bg_1,\dots,\bg_M \sim \cN(0,I_N)$ with $\lim_{N\to\infty} M/N=\alpha\in(0,\infty)$, and a fixed $I\subseteq\bbR$.
For convenience we let $\bG$ be the $M\times N$ matrix with $\bg_i$ as the $i$-th row.
We say $\bsig\in\bSig_N$ is a solution, denoted $\bsig\in S(\bG;I)$, if
\[
    \la \bg_i,\bsig\ra/\sqrt{N} \in I,\quad\forall\, i\in [M].
\]
We obtain strong low degree hardness for two important cases of the random perceptron, based on \cite{gamarnik2022algorithms,li2024discrepancy}: the symmetric Ising perceptron, with $I = [-\kappa,\kappa]$ for small $\kappa$, and the negative Ising perceptron, with $I = [\kappa,\infty)$ for large negative $\kappa$.

Here the adaptation is somewhat more challenging; these results rely on the multi-ensemble overlap property and finiteness of multi-color Ramsey-numbers, via a technique introduced by \cite{gamarnik2021partitioning}.
Nevertheless we show low degree hardness up to the near-optimal level $D\leq o(N/\log N)$.
The logarithmic factor comes from the slightly faster than $\exp(R)$ growth in the following famous bound.

\begin{proposition}{(\cite{erdos1935combinatorial})}
\label{prop:ramsey-bound}
Let $K,m\in\bbZ_+$. Then given any $K$-coloring of the complete graph on $K^{Km}$ vertices, there exists a monochromatic $m$-clique.
\end{proposition}

\begin{corollary}{(Based on \cite[Theorem 3.2]{gamarnik2022algorithms})}
\label{cor:GKPX}
    There exists sufficiently small $\kappa_0 > 0$ such that for $\kappa\in (0,\kappa_0)$ and $\alpha\geq 10\kappa^2\log(1/\kappa)$, if $\cA^\circ$ is any degree $D\leq o(N/\log N)$ polynomial with $\bbE[\|\cA^{\circ}(H_N,\omega)\|^2]\le CN$, then $\cA$ as in \eqref{eq:round-stable-alg} has probability $o_N(1)$ to solve the Ising perceptron with $I=[-\kappa,\kappa]$.
\end{corollary}

The proof is based on the following landscape obstruction.
Let $I = [-\kappa,\kappa]$.
\begin{lemma}[{\cite[Theorem 2.4]{gamarnik2022algorithms}}]
    \label{lem:sbp-ogp}
    For any $\kappa,\alpha$ as in Theorem~\ref{cor:GKPX}, there exist $0 < \eta < \beta < 1$ and $m\in \bbZ_+$ such that the following holds for all $p_1,\ldots,p_m \in [0,1]$.
    Let $\wt\bG^0,\wt\bG^1,\ldots,\wt\bG^m$ be i.i.d. copies of $\bG$, and for $1\le i\le m$ let
    \[
        \bG^i = p_i \wt\bG^0 + \sqrt{1-p_i^2} \wt \bG^i.
    \]
    With probability $1-e^{-cN}$, there does not exist $(\bx^{(1)},\dots,\bx^{(m)}) \in \bSig_N^m$ such that:
    \begin{itemize}
        \item $\bx^{(i)}\in S(\bG^i;I))$ for all $1\le i\le m$;
        \item $\|\bx^{(i)}-\bx^{(j)}\|/\sqrt{N}\in [\beta-\eta,\beta]$ for all $1\le i,j\le m$.
    \end{itemize}
\end{lemma}
We will also require the following ``chaos'' obstruction.
\begin{lemma}[{Adaptation of \cite[Lemma 6.12]{gamarnik2022algorithms}}]
    \label{lem:sbp-chaos}
    For any $\kappa,\alpha$ as in Theorem~\ref{cor:GKPX}, there exists sufficiently small constant $\rho > 0$ such that the following holds.
    Let $\eta,\beta,m$ be as in Lemma~\ref{lem:sbp-ogp}.
    For all $p_1,\ldots,p_m \in [0,\rho]$, with probability $1-e^{-cN}$, there does not exist $(\bx^{(1)},\dots,\bx^{(m)}) \in \bSig_N^m$ such that:
    \begin{itemize}
        \item $\bx^{(i)}\in S(\bG^i;I))$ for all $1\le i\le m$;
        \item $\|\bx^{(i)}-\bx^{(j)}\|/\sqrt{N}\le \beta-\eta$ for all $1\le i,j\le m$.
    \end{itemize}
\end{lemma}
\begin{proof}
    This is shown in \cite{gamarnik2022algorithms} for $\rho = 0$, through a first moment computation.
    The effect on this first moment computation of taking $\rho$ positive is easily seen to affect their estimates by a factor of $e^{\delta(\rho) N}$, where $\delta(\rho) \to 0$ as $\rho \to 0$.
    The result follows by taking $\rho$ small enough that $\delta(\rho) \le c/2$ and adjusting $c$.
\end{proof}
Let $\psolve = \bbP[\cA(\bG,\omega,\vec U) \in S(\bG;I)]$.
We suppose for sake of contradiction that $\psolve = \Omega(1)$ (possibly along a subsequence of values $N$).

Let $\eps,L,K$ be set as in \eqref{eq:eps-L-K-setting} above, where $b$ is a sufficiently small constant depending on $\kappa,\alpha,\eta,\beta,m$.
We will set $W$ as in \eqref{eq:def-W} below.
We construct a jointly Gaussian family of $M\times N$ disorder matrices $\bG^{(0)}$ and $\bG^{(k,w)}$ for $0\leq k\leq K$ and $1\leq w\leq W$, each with the same marginal distribution as $\bG$.
To construct this family, we first sample $\bG^{(0)}=\bG^{(0,1)}=\dots=\bG^{(0,W)}$.
Then for each $1\leq w\leq W$, we let
\begin{equation}
    \label{eq:multi-OGP-interpolation}
    \Big(
        \bG^{(0,w)},
        \bG^{(1,w)},
        \ldots,
        \bG^{(K,w)}
    \Big)
\end{equation}
be a Markovian sequence as before where $(\bG^{(k,w)},\bG^{(k',w)})$ are $(1-\eps)^{|k-k'|}$-correlated, and where the sequences \eqref{eq:multi-OGP-interpolation} over different $w$ are conditionally independent given $\bG^{(0)}$.
Let $\bsig^{(k,w)} = \cA(\bG^{(k,w)},\omega,\vec U)$.
For $1\le w\le W$, define
\begin{align*}
    \Ssolve(w) &= \lt\{
        \bsig^{(k,w)} \in S(\bG^{(k,w)};I) \qquad \forall 0\le k\le K
    \rt\}, \\
    \Sstab(w) &= \lt\{
        \|\bsig^{(k,w)} - \bsig^{(k+1,w)}\| / \sqrt{N} \le L\eps^{1/4} \qquad
        \forall 0\le k\le K-1
    \rt\}, \\
    \Sall(w) &= \Ssolve(w) \cap \Sstab(w).
\end{align*}
Finally, let
\begin{align*}
    U &= \{w \le W : \Sall(w)\,\textup{holds}\}, &
    \Sramsey &= \lt\{
        |U| \ge (K+1)^{(K+1)m}
    \rt\}.
\end{align*}
\begin{lemma}
    \label{lem:ramsey-poisson}
    Let
    \begin{equation}
        \label{eq:def-W}
        W = 3\lceil (K+1)^{(K+1)m} / \bbP[\Sall(1)] \rceil.
    \end{equation}
    Then $W \le e^{o(N)}$, and $\bbP(\Sramsey) \ge e^{-o(N)}$.
\end{lemma}
\begin{proof}
    Let $\punstable = \bbP[\|\bsig^{(k,w)} - \bsig^{(k+1,w)}\| / \sqrt{N} > L\eps^{1/4}]$.
    Identically to \eqref{eq:pspin-punstable}, we have $\punstable = o_N(1)$.
    By Lemma~\ref{lem:grand-correlation},
    \[
        \bbP[\Sall(1)] \ge (\psolve^2 - \punstable)^{2K} \ge e^{-o(N)},
    \]
    where the last inequality follows from $\psolve = \Omega(1)$ and $K = \lceil 1/b\eps \rceil = o(N / \log N)$.
    Moreover
    \[
        (K+1)^{(K+1)m} \le N^{o(N / \log N)} = e^{o(N)}.
    \]
    This implies $W \le e^{o(N)}$.
    Next, let $p = \bbP[\Sall(1)]$.
    Note that
    \[
        \bbE[\bbP[\Sall(1) | \bG^{(0)}]] = p,
    \]
    and thus
    \[
        \bbP[\bbP[\Sall(1) | \bG^{(0)}] \ge p/2] \ge p/2.
    \]
    Conditional on a realization of $\bG^{(0)}$ such that $\bbP[\Sall(1) | \bG^{(0)}] \ge p/2$, $|U|$ stochastically dominates a $\Bin(W,p/2)$ random variable.
    By a standard Chernoff bound, $U \ge (K+1)^{(K+1)m}$ with high conditional probability.
    It follows that $\bbP(\Sramsey) \ge (1-o(1)) p/2 \ge e^{-o(N)}$.
\end{proof}
We further define the events:
\begin{align*}
    \Schaos &= \lt\{
        \begin{array}{l}
            \text{For any distinct $1\le w_1,\ldots,w_m \le W$, the conclusion of} \\
            \text{Lemma~\ref{lem:sbp-chaos} holds for $\bG^{(K,w_1)},\ldots,\bG^{(K,w_m)}$}
        \end{array}
    \rt\}, \\
    \Sogp &= \lt\{
        \begin{array}{l}
            \text{For any distinct $0\le k_1,\ldots,k_m \le K$, $1\le w_1,\ldots,w_m \le W$,} \\
            \text{the conclusion of Lemma~\ref{lem:sbp-ogp} holds for $\bG^{(k_1,w_1)},\ldots,\bG^{(k_m,w_m)}$}
        \end{array}
    \rt\}.
\end{align*}
\begin{lemma}
    \label{lem:sbp-chaos-ogp-bounds}
    We have $\bbP(\Schaos),\bbP(\Sogp) \ge 1-e^{-cN}$.
\end{lemma}
\begin{proof}
    Let $p = (1-\eps)^K \le e^{-1/b}$.
    We can set $b$ small enough that $p\le \rho$ for $\rho$ defined in Lemma~\ref{lem:sbp-chaos}, so that this lemma applies.
    Since $W \le e^{o(N)}$ by Lemma~\ref{lem:ramsey-poisson}, taking a union bound of Lemma~\ref{lem:sbp-chaos} over the $W^m \le e^{o(N)}$ choices of $(w_1,\ldots,w_m)$ implies $\bbP(\Schaos) \ge 1-e^{-cN}$.
    Similarly, taking a union bound of Lemma~\ref{lem:sbp-ogp} over the $(K+1)^mW^m \le e^{o(N)}$ choices of $(k_1,\ldots,k_m,w_1,\ldots,w_m)$ implies $\bbP(\Sogp) \ge 1-e^{-cN}$.
\end{proof}

\begin{lemma}
    \label{lem:sbp-no-intersection}
    We have $\Sramsey \cap \Schaos \cap \Sogp = \emptyset$.
\end{lemma}
\begin{proof}
    Suppose the event $\Sramsey$ holds.
    We color the edges $(u,u')$ of the complete graph on vertex set $U$ using the colors $\{0,\ldots,K\}$, as follows.
    If there exists $1\leq k\leq K$ such that
    \[
        \|\bsig^{(k,u)} - \bsig^{(k,u')}\|/\sqrt{N}
        \in [\beta-\eta,\beta]
    \]
    then we color $(u,u')$ by the minimal such $k$.
    Otherwise we color $(u,u')$ by $0$.
    Set $b$ small enough that $L\eps^{1/4} = b \le \eta/2$.
    Since $\Sstab(u),\Sstab(u')$ hold (by definition of $u,u' \in U$), the latter case implies
    \[
        \|\bsig^{(k,u)} - \bsig^{(k,u')}\|/\sqrt{N}
        < \beta - \eta
    \]
    for all $0\le k\le K$, and in particular for $k=K$.
    This coloring uses $K+1$ colors, so by Proposition~\ref{prop:ramsey-bound} there exists a monochromatic $m$-clique.
    However, if the color $0$ contains an $m$-clique, then $\Schaos$ does not hold, and if the color $k$ contains an $m$-clique for any $1\le k\le K$, then $\Sogp$ does not hold.
\end{proof}
\begin{proof}[Proof of Corollary~\ref{cor:GKPX}]
    Lemmas~\ref{lem:ramsey-poisson}, \ref{lem:sbp-chaos-ogp-bounds}, and \ref{lem:sbp-no-intersection} yield the desired contradiction.
\end{proof}

Exactly the same argument also yields strong low degree hardness for the asymmetric Ising perceptron considered in \cite{li2024discrepancy}, where $I=[\kappa,\infty)$ for $\kappa\ll 0$.
Indeed the implementation of the ensemble multi-OGP in \cite{li2024discrepancy} is identical to \cite{gamarnik2022algorithms} modulo the precise choice of $I$ and $(\beta,\eta)$ (and the moment computations used to establish the multi-OGP and chaos properties).
We thus obtain the following.

\begin{corollary}{(Based on \cite[Theorem 4.3]{li2024discrepancy})}
\label{cor:LSZ}
     There is an absolute constant $C_{\ref{cor:LSZ}}$ such that for $\kappa\leq -1$ and $\alpha\geq \frac{C_{\ref{cor:LSZ}} e^{\kappa^2/2} \log^2 |\kappa|}{\kappa}$, if $\cA^\circ$ is any degree $D\leq o(N/\log N)$ polynomial with $\bbE[\|\cA^{\circ}(H_N,\omega)\|^2]\le CN$, then $\cA$ as in \eqref{eq:round-stable-alg} has probability $o_N(1)$ to solve the Ising perceptron with $I=[\kappa,\infty)$.
\end{corollary}

\subsection{Sparse Constraint Satisfaction Problems}
\label{subsec:CSP}

In this subsection we consider two further adaptations of the ensemble OGP, where the disorder consists of sparse Bernoulli variables rather than Gaussians.
This requires a generalization of the $L^2$-stability estimate \eqref{eq:low-degree-l2-stable} for low degree polynomials, but does not otherwise add any difficulties.
Given a finite set $\Gamma$ equipped with a probability measure $\mu$, we say a pair $(y,\hat y)\in \Gamma^2$ is $p$-correlated for $p\in [0,1]$ if both have marginal law $\mu$ and $\hat y$ is given by resampling from $\mu$ with probability $1-p$.
That is, given $y$, the conditional law of $\hat y$ is given by $y$ with probability $p$ and an independent copy with probability $1-p$.
We say $\by,\hat\by\in\Gamma^d$ are $p$-correlated if their $i$-th entries are $p$-correlated for each $1\leq i\leq d$, and are otherwise independent.
The space of degree $D$ functions on $\Gamma^d$ is defined to be the span of all functions of the form
\begin{equation}
\label{eq:general-polynomial-def}
F(y_1,\dots,y_d)
=
\prod_{i\in I} f_i(y_i)
\end{equation}
where $|I|\leq D$ and $f_i:\Gamma\to\bbR$ are arbitrary.
We say $\cA^{\circ}:\Gamma^d\to\bSig_N$ is a degree $D$ function if each of its $N$ output coordinates is degree $D$.
As before we take $\cA$ to be a randomized rounding of $\cA^{\circ}$.

We also note in passing that a degree $D$ function is a slightly more general notion than a degree $D$ polynomial, because it only counts the number of \emph{distinct} input variables which participate in each term.

\begin{remark}
    By working with resampling rather than correlated Gaussian processes as in the proof below, Theorem~\ref{thm:SK-hardness}\ref{it:low-degree-ultimate-conseq} also extends to hardness for degree $o(N)$ functions.
\end{remark}


\begin{proposition}
\label{prop:sparse-bernoulli-fourier}
    Let $\cA^{\circ}:\Gamma^d\to\bSig_N$ be a degree $D$ function such that $\bbE^{\by\sim\mu_q}[\|\cA^{\circ}(\by)\|^2]\leq CN$, and let $\cA$ be defined in terms of $\cA^\circ$ as in \eqref{eq:round-stable-alg}.
    Then if $(\by,\hat{\by})\in (\Gamma^d)^2$ is a $(1-\eps)$-correlated pair of $q$-biased bit strings:
    \begin{align*}
    \bbE[\|\cA(\by,\omega,\vec U)-\cA(\hat{\by},\omega,\vec U)\|^4]
    \leq
    12 \lt(CD\eps N^2 + \sqrt{CD\eps} N\rt).
    \end{align*}
\end{proposition}

\begin{proof}
    We recall from e.g. \cite[Chapter 8]{donnell2014analysis} that $L^2(\Gamma;\mu)$ has a Fourier basis, and that degree $D$ polynomials are exactly those functions with only degree $\leq D$ terms in this basis.
    (This basis for $L^2(\Gamma;\mu)$ is also known as the Hoeffding or functional ANOVA decomposition.)
    From \cite[Proposition 8.28]{donnell2014analysis}, we have that analogously to \eqref{eq:hermite-correlation}, for any degree $k$ Fourier basis function $\phi$:
    \[
    \bbE[\phi(\by)\phi(\hat{\by})]
    =
    (1-\eps)^k.
    \]
    Then the proof is exactly identical to that of \eqref{eq:low-degree-l2-stable}, \eqref{eq:stability-of-rounded-low-degree-l2}.
\end{proof}

We also have an exact analog of Lemma~\ref{lem:grand-correlation}.
Fix an arbitrary
finite
set $\Gamma$ and let $\mu$ be a probability measure supported on $\Gamma$.
Let $\by^{(0)},\dots,\by^{(K)}\in\Gamma^d$ be a Markovian sequence of strings with marginal law $\mu^{\otimes d}$ such that each adjacent pair is $p$-correlated.
As before, let $S=S(\by)\subseteq \Gamma^d$ be an arbitrary set-valued function.
Given some $(L,\eps,\tau)$, define
\begin{align*}
    \Ssolve(i)
    &\equiv
    \{\cA(\by^{(i)},\omega,\vec U)\in S(\by^{(i)})\}, \\
    \Sstab(i)
    &\equiv
    \{\|\cA(\by^{(i)},\omega,\vec U)-\cA(\by^{(i+1)},\omega,\vec U)\|/\sqrt{N}
    \leq
    L\eps^{1/4}+\tau\}
\end{align*}
with probabilities $\psolve = \bbP(\Ssolve(0))$, $\punstable = 1 - \bbP(\Sstab(0))$.
Define the ``success and stability'' event
\[
    \Sall(K) = \Big(\bigcap_{i=0}^K \Ssolve(i)\Big) \cap \Big(\bigcap_{i=0}^{K-1} \Sstab(i)\Big).
\]

\begin{lemma}
\label{lem:bernoulli-correlation}
With the definitions above,
\[
\bbP[\Sall(K)]
\geq
(\psolve^2-\punstable)_+^{2K}.
\]
\end{lemma}

\begin{proof}
    Like in the proof of Lemma~\ref{lem:grand-correlation}, we first consider the case where $\cA(\by,\omega,\vec U) = \cA(\by)$ does not use the auxiliary randomness $(\omega,\vec U)$, and argue by induction in powers of two.

    The pair of strings $(\by^{(0)},\by^{(1)})$ can be generated as follows.
    Sample $J \subseteq \{1,\ldots,d\}$ by including each element independently with probability $p$.
    For each $j\in J$, we sample $(y^{(0)}_j, y^{(1)}_j)$ to be equal, and otherwise sample them to be independent.

    Further, let $\by^{(0)}_J = (y^{(0)}_j : j\in J)$.
    Then $(\by^{(0)},\by^{(1)})$ are conditionally IID given $\by^{(0)}_J$.
    So, by Jensen's inequality,
    \[
        \bbP(\Ssolve(0) \cap \Ssolve(1))
        = \bbE(\bbP(\Ssolve(0) | \by^{(0)}_J)^2)
        \ge \bbE(\bbP(\Ssolve(0) | \by^{(0)}_J))^2
        = \psolve^2.
    \]
    It follows that $\bbP[\Sall(1)] \ge \psolve^2 - \punstable$, establishing the base case.
    The inductive step is identical to Lemma~\ref{lem:grand-correlation}, and the auxiliary randomness $(\omega,\vec U)$ is also addressed identically to that lemma.
\end{proof}

\subsubsection{Maximum independent set}

Next we turn to the maximum independent set problem, following the approach of \cite{wein2020independent}.
Here the random input is a sparse Erd\H{o}s--R\'enyi graph $\by\sim G(N,d/N)$, with $d$ large but constant while $N\to\infty$.
We consider deterministic algorithms $\cA^{\circ}:\bbR^{\binom{N}{2}}\to \bbR^N$ given by a degree $D$ function of $\by$ (identified with its adjacency matrix).
Let $\cA$ be the randomized rounding given by \eqref{eq:round-stable-alg}; we identify its output with the set $\{i\in [N]~:~\cA(\by,\omega,\vec U)_i=1\}$.

It is known that in the double limit $N\to\infty$ followed by $d\to\infty$, the maximum independent set has size $(1+o_d(1)) \cdot \fr{2 \log d}{d} N$ \cite{frieze1990independence}.
A simple greedy algorithm finds a roughly half-optimal independent set of size $(1+o_d(1)) \cdot \fr{\log d}{d} N$ \cite{karp1976probabilistic}, but there is now substantial evidence that algorithms cannot do better than this: \cite{rahman2017independent} showed that local algorithms cannot find larger independent sets with even small probability, while \cite{wein2020independent} showed that low degree polynomials cannot find such sets with probability too close to $1$.
We upgrade this latter result to strong low degree hardness.
Below, $\iota$ plays the role of $\eps$ in \cite{wein2020independent}.

\begin{corollary}{(Based on \cite[Theorem 1.3]{wein2020independent})}
    \label{cor:wein}
    For any $\iota>0$ there exists sufficiently large $d_*(\iota)$ such that for all $d\geq d_*(\iota)$, the following holds.
    If $\cA^\circ$ is a degree $D = o(N)$ function satisfying $\bbE[\|\cA^\circ(\by,\omega)\|^2] \le CN$, then for $\cA$ as in \eqref{eq:round-stable-alg},
    \begin{equation}
    	\label{eq:wein-conclusion}
    	\bbP\lt[\cA(\by,\omega,\vec U)\,\text{is an independent set of}\,\by\,\text{of size at least}\,\fr{(1+\iota) \log d}{d}\cdot N\rt] = o_N(1).
    \end{equation}
\end{corollary}

In the below proof, given $\iota>0$ sufficiently small (without loss of generality), we will choose $d$ sufficiently large depending on $\iota$ and $b$ sufficiently small depending on $(\iota,d)$.
We then set parameters:
\begin{align*}
    \eps &= \fr{\log (N/D)}{N}, &
    L &= \fr{b}{\eps^{1/4}}, &
    K &= \lt\lceil \fr{1}{b\eps} \rt\rceil, &
    m &= 1 + \lt\lceil \fr{5}{\iota^2} \rt\rceil, &
    b_0 &= bm.
\end{align*}
Generate a Markovian sequence $\by^{(0)},\ldots,\by^{(K)}$ where $\by^{(t+1)}$ is obtained by resampling each of the $\binom{N}{2}$ IID entries in $\by^{(t)}$ with probability $\eps$, so that $\by^{(t)},\by^{(t')}$ are $(1-\eps)^{|t-t'|}$-correlated.
The proof is based on the following landscape obstruction.
\begin{lemma}[{Adaptation of \cite[Proposition 2.3]{wein2020independent}}]
    \label{lem:indset-ogp}
    With probability $1-e^{-cN}$, there do not exist $0\leq t_1\leq \dots\leq t_m\leq K$ and subsets $S_1,\dots,S_m\subseteq [N]$ satisfying the following:
    \begin{enumerate}[label=(\roman*)]
        \item $S_j$ is an independent set in $\by^{(t_j)}$ for each $1\leq j\leq m$.
        \item $|S_j|\geq \frac{(1+\iota)N\log d}{d}$ for each $1\leq j\leq m$.
        \item $|S_j\backslash (\cup_{\ell<j} S_{\ell})|\in \lt[\frac{\iota N\log d}{4d},\frac{\iota N\log d}{2d}\rt]$ for all $2\leq j\leq m$.
    \end{enumerate}
\end{lemma}
\begin{proof}
    This is proved by bounding the first moment of the number of such $(t_1,\ldots,t_m,S_1,\ldots,S_m)$.
    While \cite{wein2020independent} uses a slightly different correlated ensemble of $\by^{(t)}$, in both cases the first moment contribution of any $(t_1,\ldots,t_m)$ is upper bounded by the case $t_1=\cdots=t_m$, and in this case the event fails to hold with probability $e^{-cN}$.
	The result follows by a union bound over the $N^{O(1)}$ sequences of $t_j$.
\end{proof}
Next, we need a ``chaos'' result.
Let $S^{(t)}$ be the output of $\cA(\by^{(t)},\omega,\vec U)$, intepreted as a subset of $[N]$.
\begin{lemma}[{Adaptation of \cite[Lemma 2.4]{wein2020independent}}]
    \label{lem:indset-chaos}
    With probability $1-e^{-cN}$, there does not exist $j\le m$ and a sequence $0\leq t_1\leq \dots\leq t_j\leq K$ with $t_j\geq t_{j-1}+\frac{1}{b_0\eps}$ and set $S\subseteq [N]$ satisfying the following:
    \begin{enumerate}[label=(\roman*)]
        \item $S$ is an independent set in $\by^{(t_j)}$.
        \item $|S|\geq \frac{(1+\iota)N\log d}{d}$.
        \item \label{it:size-of-indsets-so-far} $|\cup_{\ell<j} S^{(t_\ell)}| \ge \frac{N\log d}{\iota^3 d}$.
        \item $|S \cap (\cup_{\ell<j} S^{(t_\ell)})| \ge \frac{\iota \log d}{d} \cdot N$.
    \end{enumerate}
\end{lemma}
Note that in this event, only $S$ is allowed to vary arbitrarily, while $S^{(t_1)},\ldots,S^{(t_{j-1})}$ must depend only on their respective input $\by^{(t_\ell)}$ and in particular not on $\by^{(t_j)}$.
\begin{proof}
    If $(\by^{(t)})_{t\leq t_{j-1}}$ were independent of $\by^{(t_j)}$, this follows directly from \cite[Lemma 2.4]{wein2020independent}, where $1/\iota^3$ plays the role of $a$ therein.
    Here the correlation between $\by^{(t_{j-1})}$ and $\by^{(t_j)}$ is at most $(1-\eps)^{1/b_0\eps} \le e^{-1/b_0} = e^{-1/bm}$ (and $(\by^{(t)})_{t\leq t_{j-1}}$ is conditionally independent of $\by^{(t_j)}$ given $\by^{(t_{j-1})}$).
    For sufficiently small $b$, the lemma can be proved by the same first moment computation.
\end{proof}
Finally, we need the following bound on the largest independent set.
\begin{lemma}[{\cite[Lemma 2.2]{wein2020independent}}]
    \label{lem:indset-max}
    With probability $1-e^{-cN}$ there is no independent set in any $\by^{(k)}$ of size larger than $\frac{(2+\iota)N\log d}{d}$.
\end{lemma}
Similarly to above, define
\begin{align*}
    \Ssolve &= \{\text{$S^{(t)}$ is an independent set of $\by^{(t)}$ of size $\ge (1+\iota)N \log d / d$, for all $0\le t\le K$}\}, \\
    \Sstab &= \{\text{$|S^{(t)} \Delta S^{(t+1)}| / N \le L^2\eps^{1/2}/4$ for all $0\le t\le K-1$}\},
\end{align*}
where $S^{(t)} \Delta S^{(t+1)}$ denotes the symmetric difference between sets $S^{(t)}, S^{(t+1)}$.
Let $\Sogp$, $\Schaos$, $\Smax$ be the events in Lemmas~\ref{lem:indset-ogp}, \ref{lem:indset-chaos}, and \ref{lem:indset-max}, which hold with probability $1-e^{-cN}$.
\begin{lemma}[{Adaptation of \cite[Proof of Theorem 1.3]{wein2020independent}}]
    \label{lem:indset-no-intersection}
    $\Ssolve \cap \Sstab \cap \Sogp \cap \Schaos \cap \Smax = \emptyset$.
\end{lemma}
\begin{proof}
    Suppose $\Ssolve, \Sstab, \Schaos, \Smax$ hold; we will contradict $\Sogp$.
    Let $t_1 = 0$ and recursively define $t_j$ for $2\leq j\leq m$ to be the smallest $t\geq t_{j-1}$ such that
    \begin{equation}
        \label{eq:ind-set-iterative-moats}
        |S^{(t)} \backslash (\cup_{\ell<j} S^{(t_\ell)})|\in \lt[\frac{\iota N\log d}{4d},\frac{\iota N\log d}{2d}\rt].
    \end{equation}
    Note that $L^2\eps^{1/2} / 4 = b^2/4$.
    We set $b$ small enough that this is smaller than $\fr{\iota\log d}{4d}$.
    Then, $\Sstab$ ensures that as $t$ advances, the left-hand side of \eqref{eq:ind-set-iterative-moats} never jumps over the right-hand interval.
    However, $\Schaos$ implies that $t_j \le t_{k-1} + \fr{1}{b_0\eps}$ is successfully chosen.
    (Note that \ref{it:size-of-indsets-so-far} in Lemma~\ref{lem:indset-chaos} holds since the rule \eqref{eq:ind-set-iterative-moats} was used to construct previous independent sets, while $\Smax$ upper bounds the size of $S^{(t_1)}$.)
    Therefore, $t_m \le \fr{m}{b_0 \eps} = \fr{1}{b\eps} \le K$, so this procedure finds $(t_1,\ldots,t_m)$ before exhausting the times $\{0,\ldots,K\}$.
    Then, the existence of $(t_1,\ldots,t_m)$ and $(S^{(t_1)},\ldots,S^{(t_m)})$ contradicts $\Sogp$.
\end{proof}

\begin{proof}[Proof of Corollary~\ref{cor:wein}]
    Let $\psolve$ denote the left-hand side of \eqref{eq:wein-conclusion} and
    \[
        \punstable = \bbP\lt[|S^{(1)} \Delta S^{(2)}| / N \ge L^2\eps^{1/2}/4\rt]
        = \bbP\lt[\|\cA(\by^{(1)},\omega,\vec U) - \cA(\by^{(2)},\omega,\vec U)\| / \sqrt{N} \ge L\eps^{1/4} \rt].
    \]
    By the above discussion, $\bbP(\Sogp \cap \Schaos \cap \Smax) = 1-e^{-cN}$.
    Lemmas~\ref{lem:bernoulli-correlation} and \ref{lem:indset-no-intersection} then imply
    \[
        (\psolve^2-\punstable)_+^{2K} \le \bbP(\Ssolve \cap \Sstab) \le e^{-cN} \qquad \implies \qquad
        \psolve \le \punstable^{1/2} + e^{-cN/4K}.
    \]
    We can bound $\punstable = o_N(1)$ similarly to \eqref{eq:pspin-punstable}, but using Proposition~\ref{prop:sparse-bernoulli-fourier} in place of Proposition~\ref{prop:low-degree-stable}
    The term $e^{-cN/4K}$ is bounded by $o_N(1)$ identically to \eqref{eq:pspin-exponential}.
\end{proof}

\subsubsection{Random $k$-SAT}

Next we consider random $k$-SAT.
Here the problem instance consists of $N$ Boolean variables $x_1,\ldots,x_N \in \{\ttT,\ttF\}$ and $M$ IID clauses, where $\lim_{N\to\infty} M/N = \alpha$.
Each clause is a boolean disjunction (``or'') of $k$ IID literals, sampled uniformly at random from $\Gamma = \{x_1,\ldots,x_N,\bar x_1,\ldots,\bar x_N\}$, where $\bar x_i$ denotes the negation of $x_i$ --- for example, when $k=3$ a possible clause is $x_1 \lor \bar x_3 \lor x_7$.
Thus a random $k$-SAT instance can be identified with a sample $\by \sim \mu^{\otimes kM}$, where $\mu = \mathsf{unif}(\Gamma)$.
The goal is to find an assignment $\bx = (x_1,\ldots,x_N) \in \{\ttT,\ttF\}^N$ such that all $M$ clauses evaluate to true.
We will identify the output space $\bSig_N = \{-1,1\}^N$ of $\cA$ with $\{\ttT,\ttF\}^N$ by identifying $+1$ with $\ttT$ and $-1$ with $\ttF$.

We consider the setting where $k$ is large but fixed and $N\to\infty$.
It is known that satisfying assignments exist precisely for $\alpha$ up to an explicit threshold $\alpha_\sat(k)$ \cite{mezard2002analytic,ding2022proof}, which has asymptotic $2^k \log 2 - \fr12 (1 + \log 2) + o_k(1)$, while the best algorithm known for this problem \cite{coja2010better} finds a satisfying assignment up to clause density $(1-o_k(1)) 2^k \log k / k$.
Bresler and the first author \cite{bresler2021ksat} showed that for an explicit $\kappa^* \approx 4.911$, random $k$-SAT is hard for low degree polynomials above clause density $(1+o_k(1)) \kappa^* 2^k \log k / k$, in the sense that such algorithms cannot succeed with probability too close to $1$ (and that the more restricted class of sequential local algorithms cannot succeed with even small probability).
We improve this result to strong low degree hardness.

\begin{corollary}{(Based on \cite[Theorem 2.6]{bresler2021ksat})}
\label{cor:BH}
	For any $\kappa > \kappa^* \approx 4.911$, there exists sufficiently large $k_*(\kappa)$ such that for all $k \ge k_*(\kappa)$ and $\alpha = \kappa 2^k \log k / k$, the following holds.
	If $\cA^\circ$ is a degree $D = o(N)$ function satisfying $\bbE[\|\cA^\circ(\by,\omega)\|^2] \le CN$, then for $\cA$ as in \eqref{eq:round-stable-alg},
	\begin{equation}
		\label{eq:BH-conclusion}
		\bbP\lt[\cA(\by,\omega,\vec U)\,\text{is a satisfying assignment of}\,\by\rt] = o_N(1).
	\end{equation}
\end{corollary}
In the below proof, given $\kappa > \kappa^*$ we will choose $k$ sufficiently large depending on $\kappa$ and $b$ sufficiently small depending on $(\kappa,k)$.
We set parameters:
\begin{align*}
	\eps &= \fr{\log (N/D)}{N}, &
	L &= \fr{b}{\eps^{1/4}}, &
	K &= \lt\lceil \fr{1}{b\eps} \rt\rceil, &
	b_0 &= bk.
\end{align*}
We generate a Markovian sequence $\by^{(0)},\ldots,\by^{(K)}$, where each $\by^{(t)}$ is marginally a sample from $\mu^{\otimes kM}$ and where $\by^{(t+1)}$ is obtained by resampling each of the $kM$ IID entries in $\by^{(t)}$ with probability $\eps$.
Let $h$ be the \emph{conditional overlap entropy} defined in \cite[Section 4.3]{bresler2021ksat} (and denoted $H(\cdots)$ therein).
The only property of it we will use Lemma~\ref{lem:conditional-overlap-entropy} below, which establishes a quantitative continuity estimate.
We use the following landscape obstruction.
\begin{lemma}[{Adaptation of \cite[Proposition 4.7(c)]{bresler2021ksat}}]
	\label{lem:ksat-ogp}
	There exist constants $(\beta,\eta)$ depending on $k,\kappa$ such that the following holds with probability $1-e^{-cN}$.
    There does not exist a sequence $0\le t_0 \le \cdots \le t_k \le K$ and assignments $\bx^0,\ldots,\bx^k \in \bSig_N$ such that:
    \begin{enumerate}[label=(\roman*)]
        \item For all $0\le \ell \le k$, $\bx^\ell$ satisfies the $k$-SAT formula $\by^{(t_\ell)}$.
        \item For all $1\le \ell \le k$, we have $h(\bx^\ell | \bx^0, \ldots, \bx^{\ell-1}) \in [\beta-\eta,\beta]$.
    \end{enumerate}
\end{lemma}
\begin{proof}
	This is proved by bounding the first moment of the number of such $(t_0,\ldots,t_k,\bx^0,\ldots,\bx^k)$.
	While \cite{bresler2021ksat} uses a slightly different correlated ensemble, the first moment contribution of any $(t_0,\ldots,t_k)$ is upper bounded by the case $t_0=\cdots=t_k$, and in this case the event fails to hold with probability $e^{-cN}$.
	The result follows by a union bound over the $N^{O(1)}$ sequences of $t_j$.
	We note that $[\beta-\eta,\beta]$ is denoted in \cite{bresler2021ksat} as $[\beta_+ \fr{\log k}{k}, \beta_+ \fr{\log k}{k}]$.
\end{proof}
We will need the following ``chaos" result.
Let $\bx^{(t)} = \cA(\by^{(t)},\omega,\vec U)$.
\begin{lemma}[{Adaptation of \cite[Proposition 4.7(b)]{bresler2021ksat}}]
	\label{lem:ksat-chaos}
	With probability $1-e^{-cN}$, there does not exist $j\le k$ and a sequence $0\le t_0 \le \cdots \le t_j \le K$, with $t_j \ge t_{j-1} + \fr{1}{b_0\eps}$ and assignment $\bx \in \Sigma_N$ satisfying the following:
	\begin{enumerate}[label=(\roman*)]
		\item $\bx$ satisfies $\by^{(t_j)}$.
		\item $h(\bx | \bx^{(t_0)},\ldots,\bx^{(t_{j-1})}) \le \beta$.
	\end{enumerate}
\end{lemma}
Note that we only allow $\bx$ to vary arbitrarily here, while $\bx^{(t_0)},\ldots,\bx^{(t_{j-1})}$ must depend only on their respective inputs $\by^{(t_\ell)}$, and not on $\by^{(t_j)}$.
\begin{proof}
	If $(\by^{(t)})_{t\le t_{j-1}}$ were independent of $\by^{(t_j)}$, this follows directly from \cite[Proposition 4.7(b)]{bresler2021ksat}.
	Here the correlation between $\by^{(t_{j-1})}$ and $\by^{(t_j)}$ is at most $(1-\eps)^{1/b_0\eps} \le e^{-1/b_0} = e^{-1/bk}$ (and $(\by^{(t)})_{t\leq t_{j-1}}$ is conditionally independent of $\by^{(t_j)}$ given $\by^{(t_{j-1})}$).
    For sufficiently small $b$, the lemma can be proved by the same first moment computation.
\end{proof}
Finally, we require the following continuity estimate on the conditional overlap entropy $h$.
For $\bx,\bx'\in \bSig_N$, let $\Delta(\bx,\bx') = \|\bx-\bx'\|^2/4 \in [0,N]$ denote the Hamming distance of $\bx,\bx'$.
\begin{lemma}[{\cite[Lemma 4.8]{bresler2021ksat}}]
	\label{lem:conditional-overlap-entropy}
	For any $\ell \in \bbN$ and $\bx,\bx',\by^0,\ldots,\by^{\ell-1} \in \bSig_N$ such that $\Delta(\bx,\bx') \le \fr{N}{2}$,
	\[
		\lt|
			h(\bx | \by^0,\ldots,\by^{\ell-1})
			- h(\bx' | \by^0,\ldots,\by^{\ell-1})
		\rt| \le h(\Delta(\bx,\bx')/N).
	\]
	The $h$ on the right-hand side denotes the binary entropy function $h(x) = - x \log x - (1-x) \log(1-x)$.
\end{lemma}
Similarly to above, we define ``success'' and ``stability'' events:
\begin{align*}
	\Ssolve &= \lt\{
		\text{$\bx^{(t)}$ satisfies $\by^{(t)}$ for all $0\le t\le K$}
	\rt\}, \\
	\Sstab &= \lt\{
		\text{$\Delta(\bx^{(t)},\bx^{(t+1)}) / N \le L^2 \eps^{1/2} / 4$ for all $0\le t\le K-1$}
	\rt\}.
\end{align*}
Let $\Sogp,\Schaos$ be the events in Lemmas~\ref{lem:ksat-ogp} and \ref{lem:ksat-chaos}, which hold with probability $1-e^{-cN}$.
\begin{lemma}[{Adaptation of \cite[Proof of Theorem 2.6]{bresler2021ksat}}]
    \label{lem:ksat-no-intersection}
    We have $\Ssolve \cap \Sstab \cap \Sogp \cap \Schaos = \emptyset$.
\end{lemma}
\begin{proof}
    Suppose $\Ssolve, \Sstab, \Schaos$ hold; we will contradict $\Sogp$.
    Let $t_0 = 0$ and recursively define $t_j$ for $1\leq j\leq k$ to be the smallest $t\geq t_{j-1}$ such that
    \begin{equation}
        \label{eq:ksat-iterative-moats}
        h(\bx^{(t)} | \bx^{(t_0)},\ldots,\bx^{(t_{j-1})}) \in [\beta-\eta,\beta].
    \end{equation}
    Note that $L^2\eps^{1/2} / 4 = b^2/4$.
    We set $b$ small enough that $b^2/4 \le 1/2$ and for $h$ the binary entropy function, $h(b^2/4) \le \eta$.
    Then, $\Sstab$ and Lemma~\ref{lem:conditional-overlap-entropy} ensure that as $t$ advances, the left-hand side of \eqref{eq:ksat-iterative-moats} never jumps over the right-hand interval.
    However, $\Schaos$ implies that $t_j \le t_{k-1} + \fr{1}{b_0\eps}$ is successfully chosen.
    Therefore, $t_k \le \fr{k}{b_0 \eps} = \fr{1}{b\eps} \le K$, so this procedure finds $(t_0,\ldots,t_k)$ before exhausting the times $\{0,\ldots,K\}$.
    Then, the existence of $(t_0,\ldots,t_k)$ and $(\bx^{(t_0)},\ldots,\bx^{(t_k)})$ contradicts $\Sogp$.
\end{proof}
\begin{proof}[Proof of Corollary~\ref{cor:BH}]
    Let $\psolve$ denote the left-hand side of \eqref{eq:BH-conclusion} and
    \[
        \punstable =   \bbP\lt[\Delta(\bx^{(0)},\bx^{(1)}) / N \ge L^2\eps^{1/2}/4\rt]
        = \bbP\lt[\|\cA(\by^{(0)},\omega,\vec U) - \cA(\by^{(1)},\omega,\vec U)\| / \sqrt{N} \ge L\eps^{1/4} \rt].
    \]
    By the above discussion, $\bbP(\Sogp \cap \Schaos) = 1-e^{-cN}$.
    Lemmas~\ref{lem:bernoulli-correlation} and \ref{lem:indset-no-intersection} then imply
    \[
        (\psolve^2-\punstable)_+^{2K} \le \bbP(\Ssolve \cap \Sstab) \le e^{-cN} \qquad \implies \qquad
        \psolve \le \punstable^{1/2} + e^{-cN/4K}.
    \]
    These two terms are bounded similarly to \eqref{eq:pspin-punstable} (using Proposition~\ref{prop:sparse-bernoulli-fourier} in place of Proposition~\ref{prop:low-degree-stable}) and \eqref{eq:pspin-exponential}.
\end{proof}

\subsection{Maximum Independent Set in $G(N,1/2)$}
\label{subsec:dense-max-ind-set}

Finally, we turn to the maximum independent set problem on $G(N,1/2)$.
For this problem, the largest independent set has size $(2+o_N(1)) \log_2 N$, while a greedy algorithm finds an independent set of size $(1+o_N(1)) \log_2 N$ \cite{karp1976probabilistic}.
The latter value is conjectured to be the limit of all efficient algorithms, and \cite[Theorem 1.9]{wein2020independent} showed that surpassing this value is hard for degree $D\leq o(\log^2 N)$ polynomials, in that such algorithms fail with probability at least $e^{-\Omega(D)}N^{-2}$.
We upgrade this result to strong low degree hardness.
We note that the degree $\log^2 N$ arises because the main landscape obstruction occurs with probability $1-e^{-\Theta(\log^2 N)}$, which is consistent with the discussion in Subsection~\ref{subsec:strong-hardness-discussion}.
This range of degrees is expected to be sharp based on the low-degree prediction, since $e^{O(\log^2 N)}$ time suffices to find maximum independent sets via brute-force search.

In this application, we will treat $\by \sim G(N,1/2)$ as an element of $\Ber(1/2)^{\otimes \binom{N}{2}}$, and will identify a set of vertices $S\subseteq [N]$ with the element $\ind_S \in \{0,1\}^N$ (rather than $\{-1,1\}^N$ as above).
Rather than randomly rounding the output of $\cA^\circ$ as in \eqref{eq:round-stable-alg}, we will use a more stringent deterministic rounding scheme similar to \cite{gamarnik2020optimization,wein2020independent,bresler2021ksat}.
(This is because the stability estimates from Subsection~\ref{subsec:stability} are tailored for macroscopic distances of order $\sqrt{N}$ between solutions, which is not the case here.)
Define the function $\round : \bbR \to \{0,1,*\}$ by
\[
	\round(z) = \begin{cases}
		0 & |z| \le 1/2, \\
		1 & |z| \ge 1, \\
		* & \text{otherwise}.
	\end{cases}
\]
Furthermore, for $\bz \in \bbR^N$, let $\round(\bz) \in \{0,1,*\}^N$ denote the vector obtained by applying $\round$ entrywise to $\bz$.
Let $\Delta(\bx,\bx') = |\{i \in [N] : x_i \neq x'_i\}|$ denote the Hamming distance between $\bx,\bx' \in \{0,1,*\}^N$.
Let $\cP([N])$ denote the power set of $[N]$.
\begin{definition}
	For $b > 0$, a function $F : \{0,1,*\}^N \times \{0,1\}^{\binom{N}{2}} \to \cP([N]) \cup \{\err\}$ is \emph{$b$-repairing} if the following holds.
	If $F(\bz,\by) = S \neq \err$, then $S$ is an independent set of $\by$ and $\Delta(\ind_S,\round(\bz)) \le b \log_2 N$.
\end{definition}
We then let
\begin{equation}
	\label{eq:round-sparse-algorithm}
	\cA(\by,\omega) = F(\round(\cA^\circ(\by,\omega)),\by),
\end{equation}
for $F$ an arbitrary $b$-repairing function.
In other words, we give $\cA$ the power to fix mistakes in $b \log_2 N$ entries of $\round(\cA^\circ(\by,\omega))$ in a computationally unbounded way, where all entries with value $*$ are considered mistakes.
(If there are too many mistakes, then $\cA$ must output $\err$ and fail.)
Our result in this setting is as follows.
\begin{corollary}[{Based on \cite[Theorem 1.9]{wein2020independent}}]
	\label{cor:wein-dense}
	For any $\iota > 0$, there exists a sufficiently small constant $b > 0$ such that the following holds.
	If $\cA^\circ$ is a degree $D = o(\log^2 N)$ function with $\bbE[\|\cA^\circ(\by,\omega)\|^2] \le C \log N$, then for any $\cA$ of the form \eqref{eq:round-sparse-algorithm} where $F$ is $b$-repairing,
	\begin{equation}
		\label{eq:wein-dense-conclusion}
		\bbP[\text{$\cA(\by,\omega)$ is an independent set of $\by$ of size at least $(1+\iota) \log_2 N$}] = o_N(1).
	\end{equation}
\end{corollary}
Throughout the below proof, we set $\iota$ sufficiently small and $b$ sufficiently small depending on $\iota$.
We then set parameters:
\begin{align*}
	\eps &= \fr{\log(\log^2 N / D)}{\log^2 N}, &
	K &= \lt\lceil \fr{1}{b\eps} \rt\rceil, &
	m &= 1 + \lt\lceil \fr{5}{\iota^2} \rt\rceil, &
	b_0 &= bm.
\end{align*}
We generate a Markovian sequence $\by^{(0)},\ldots,\by^{(K)}$ where $\by^{(t+1)}$ is obtained by resampling each entry of $\by^{(t)}$ with probability $\eps$.
Let $S^{(t)} = \cA(\by^{(t)},\omega)$ and $\bz^{(t)} = \cA^\circ(\by^{(t)},\omega)$.
The following three lemmas are analogous to Lemmas~\ref{lem:indset-ogp}, \ref{lem:indset-chaos}, and \ref{lem:indset-max}, and we omit their proofs.
\begin{lemma}
	\label{lem:indset-dense-ogp}
	With probability $1-e^{-c\log^2 N}$, there do not exist $0\le t_1 \le \cdots \le t_m \le K$ and subsets $S_1,\ldots,S_m \subseteq [N]$ satisfying the following:
	\begin{enumerate}[label=(\roman*)]
		\item $S_j$ is an independent set in $\by^{(t_j)}$ for each $1\le j\le m$.
		\item $|S_j| \ge (1+\iota) \log_2 N$ for each $1\le j\le m$.
		\item $|S_j \setminus (\cup_{\ell < j} S_\ell)| \in [\fr14 \iota \log_2 N, \fr12 \iota \log_2 N]$ for each $2\le j\le m$.
	\end{enumerate}
\end{lemma}
\begin{lemma}
	\label{lem:indset-dense-chaos}
	With probability $1-e^{-c\log^2 N}$ there does not exist $j\le m$ and a sequence $0\le t_1 \le \cdots \le t_j \le K$ with $t_j \ge t_{j-1} + \fr{1}{b_0 \eps}$ and set $S\subseteq [N]$ satisfying the following:
	\begin{enumerate}[label=(\roman*)]
		\item $S$ is an independent set in $\by^{(t_j)}$.
		\item $|S| \ge (1+\iota) \log_2 N$.
		\item $|\cup_{\ell < j} S^{(t_\ell)}| \ge \iota^{-3} \log_2 N$.
		\item $|S \cap (\cup_{\ell < j} S^{(t_\ell)})| \ge \iota \log_2 N$.
	\end{enumerate}
\end{lemma}
\begin{lemma}
	\label{lem:indset-dense-max}
	With probability $1-e^{-c\log^2 N}$, there is no independent set in any $\by^{(k)}$ of size larger than $(2+\iota) \log_2 N$.
\end{lemma}
Define analogously to above
\begin{align*}
	\Ssolve &= \lt\{
		\text{$S^{(t)}$ is an independent set of $\by^{(t)}$ of size $\ge (1+\iota) \log_2 N$, for all $0\le t\le K$}
	\rt\}\\
	\Sstab &= \lt\{
		\text{$|S^{(t)} \Delta S^{(t+1)}| \le \fr14 \iota \log_2 N$ for all $0\le t\le K-1$}
	\rt\},
\end{align*}
where by convention neither $\Ssolve$ nor $\Sstab$ holds if any of the $S^{(t)}$ are $\err$.
Let $\Sogp, \Schaos, \Smax$ be the events in Lemmas~\ref{lem:indset-dense-ogp}, \ref{lem:indset-dense-chaos}, and \ref{lem:indset-dense-max}.
Analogously to Lemma~\ref{lem:indset-no-intersection}:
\begin{lemma}
	\label{lem:indset-dense-no-intersection}
	We have $\Ssolve \cap \Sstab \cap \Sogp \cap \Schaos \cap \Smax = \emptyset$.
\end{lemma}
\begin{proof}[Proof of Corollary~\ref{cor:wein-dense}]
	Let $\psolve$ be the left-hand side of \eqref{eq:wein-dense-conclusion} and
	\[
		\punstable = \bbP\lt[|S^{(1)} \Delta S^{(2)}| \ge \fr14 \iota \log_2 N\rt],
	\]
	where the event in the probability does not hold if either of $S^{(1)}$ or $S^{(2)}$ is $\err$.
	Since $\bbP(\Sogp \cap \Schaos \cap \Smax) \ge 1-e^{-c\log^2 N}$, Lemmas~\ref{lem:bernoulli-correlation} and \ref{lem:indset-dense-no-intersection} imply
	\[
		(\psolve^2 - \punstable)^{2K} \le \bbP(\Ssolve \cap \Sstab) \le e^{-c\log^2 N} \qquad \implies \qquad
		\psolve \le \punstable^{1/2} + e^{-c\log^2 N / 4K}.
	\]
	Let $J_t$ be the set of coordinates where $\ind_{S^{(t)}}$ differs from $\round(\bz^{(t)})$.
	Since $F$ is $b$-repairing, if $S^{(t)} \neq \err$, then $|J_t| \le b\log_2 N$.
	Thus, (on the event $S^{(1)},S^{(2)} \neq \err$)
	\[
		|S^{(1)} \Delta S^{(2)}|
		\le \| (\round(\bz^{(1)}) - \round(\bz^{(2)}))_{[N]\setminus (J_1 \cup J_2)}\|_2^2 + |J_1| + |J_2|
		\le 4\|\bz^{(1)} - \bz^{(2)}\|_2^2 + 2b\log_2 N.
	\]
	We set $b \le \fr{1}{16} \iota$, so that $\fr14 \iota \log_2 N - 2b \log_2 N \ge \fr18 \log_2 N$, and thus
	\begin{align*}
		\punstable \le \bbP\lt[
			\|\bz^{(1)} - \bz^{(2)}\|_2^2
			\ge \fr{1}{32} \iota \log_2 N
		\rt]
		&\stackrel{\eqref{eq:low-degree-l2-stable}}{\le}
		\fr{2D\eps \cdot C \log N}{\fr{1}{32} \iota \log_2 N} \\
		&\asymp D\eps
		\le \sqrt{\fr{D}{\log^2 N}}
		= o_N(1).
	\end{align*}
	Furthermore, (for dimension-free but varying $c$)
	\[
		e^{-c\log^2 N / 4K}
		= e^{-c\eps \log^2 N }
		= (D/\log^2 N)^c
		= o_N(1).
	\]
	Combining concludes the proof.
\end{proof}

\subsection{Low Degree Upper Bound for Spin Glass Optimization}
\label{subsec:low-degree-UB-GS}

In this subsection we prove a sharp converse to Corollary~\ref{cor:GJW} for Ising mixed $p$-spin models, that degree $O(N)$ algorithms can approximate the ground state.

Fix a mixed Ising $p$-spin Hamiltonian with covariance function $\xi(t)=\sum_{k=2}^{\bar k}\gamma_k^2 t^k$, and let
\[
\GS(\xi)
=
\lim_{N\to\infty}
\bbE\lt[
\max_{\bsig\in\bSig_N} H_N(\bsig)/N
\rt]
\]
denote its asymptotic ground-state energy \cite{auffinger2017parisi}.

\begin{theorem}
\label{thm:low-degree-upper-bound-GS}
For every $\eps>0$ there exists $C=C(\eps,\xi)<\infty$ and a sequence of randomized algorithms $\cA_N^\circ:(H_N,\omega)\mapsto \bbR^N$ such that, for each fixed realization of the auxiliary randomness $\omega$, the map $H_N\mapsto \cA_N^\circ(H_N,\omega)$ is a polynomial of degree at most $CN$ in the disorder coordinates of $H_N$.
With
\[
\cA_N(H_N,\omega)
\equiv
\sign\big(\cA_N^\circ(H_N,\omega)\big)\in \bSig_N
\]
defined coordinate-wise, one has
\[
\liminf_{N\to\infty}
\bbE\lt[
H_N(\cA_N(H_N,\omega))/N
\rt]
\geq
\GS(\xi)-\eps.
\]
\end{theorem}

\paragraph{Preparation for the proof.}
Given $\beta,\eta>0$, let the auxiliary randomness be a Gaussian vector
\[
\omega=\bh=(h_1,\ldots,h_N)\sim \cN(0,I_N),
\]
independent of $H_N$, and define the perturbed Hamiltonian
\[
H_N^\eta(\bsig)
=
H_N(\bsig)
+
\eta \sum_{i=1}^N h_i \sigma_i.
\]
Let $\la \cdot \ra_{\beta,\eta}$ denote expectation with respect to the Gibbs measure on $\bSig_N$ with density proportional to $e^{\beta H_N^\eta(\bsig)}$, and write
\[
\bm_{\beta,\eta}(H_N,\bh)
\equiv
\la \bsig \ra_{\beta,\eta}
\in [-1,1]^N
\]
for the corresponding Gibbs barycenter.

\begin{lemma}
\label{lem:hermite-truncation-noise-stability}
Let $\omega$ be any auxiliary randomness independent of $H_N$, and let $F(H_N,\omega)\in L^2(\Omega;\bbR^N)$.
For each fixed $\omega$, write the Hermite expansion of $F$ in the Gaussian coordinates of $H_N$ as
\[
F=\sum_{k\geq 0} F_k,
\qquad
F^{\leq D}=\sum_{k=0}^D F_k.
\]
If $(H_N,H_{N,p})$ is a $p$-correlated pair, then for every $D\geq 0$,
\begin{equation}
\label{eq:hermite-truncation-noise-stability}
\bbE\lt[
\|F-F^{\leq D}\|^2
\rt]
\leq
\frac{1}{2(1-p^{D+1})}
\bbE\lt[
\|F(H_N,\omega)-F(H_{N,p},\omega)\|^2
\rt].
\end{equation}
\end{lemma}

\begin{proof}
Decomposing $F(H_N,\omega)-F(H_{N,p},\omega)$ degree-by-degree and using Hermite orthogonality and \eqref{eq:hermite-correlation},
\[
\bbE\lt[
\|F(H_N,\omega)-F(H_{N,p},\omega)\|^2
\rt]
=
2\sum_{k\geq 0}
(1-p^k)
\bbE\lt[
\|F_k\|^2
\rt]
\geq
2(1-p^{D+1})
\sum_{k>D}
\bbE\lt[
\|F_k\|^2
\rt].
\]
As the degrees are orthogonal in $L^2$, we obtain
$\sum_{k>D}
\bbE\lt[
\|F_k\|^2
\rt]
=
\bbE\lt[
\|F-F^{\leq D}\|^2
\rt]$, yielding \eqref{eq:hermite-truncation-noise-stability}.
\end{proof}

\begin{proposition}[{\cite[Proposition~2.8]{alaoui2023shattering}}]
\label{prop:ams23-tv-stability}
Fix $\beta,\eta>0$ and let $\tau>0$.
Then there exist $c=c(\beta,\eta,\xi,\tau)>0$ and $N_0<\infty$ such that the following holds for every $N\geq N_0$ and every
\[
p\in \lt[1-\frac{c}{N},1\rt].
\]
If $(H_N,H_{N,p})$ is a $p$-correlated pair, the same Gaussian field $\bh$ is used in both perturbed Hamiltonians, and
\[
\mu_N
\equiv
\mu_{N,\beta,\eta}^{H_N,\bh},
\qquad
\mu_{N,p}
\equiv
\mu_{N,\beta,\eta}^{H_{N,p},\bh}
\]
be the corresponding Gibbs measures on $\bSig_N$.
Then
\[
\bbE\lt[
\TV(\mu_N,\mu_{N,p})
\rt]
\leq
\tau.
\]
\end{proposition}

\begin{proof}[Proof of Theorem~\ref{thm:low-degree-upper-bound-GS}]
    By Proposition~\ref{prop:gradients-bounded}, there exists $C_1(\xi) > 0$ such that with probability $1-e^{-cN}$,
    \[
    \sup_{\bx\in [-1,1]^N}\|\nabla H_N(\bx)\|
    \leq
    C_1(\xi)\sqrt{N}
    \]
    Let $\cG_N$ be the event that this holds and $\|\bh\| \leq 2\sqrt{N}$.
    Gaussian concentration then implies $\bbP(\cG_N)\geq 1-e^{-cN}$.
    On $\cG_N$, there exists $C_2(\xi,\eta) > 0$ such that
    \[
        \sup_{\bx\in[-1,1]^N}\|\nabla H_N^\eta(\bx)\|
        \leq
        C_2(\xi,\eta)\sqrt{N}.
    \]
    Fix $\eps>0$.
    We now choose the parameters successively.
    First choose $\eta>0$ so small that $2\eta \sqrt{\frac{2}{\pi}}
    \leq
    \frac{\eps}{4}$.
    Next choose $\beta$ so large, depending on $(\eps,\eta,\xi)$, that
    \[
    \frac{\log 2}{\beta}
    \leq
    \frac{\eps}{4},
    \qquad
    C_2(\xi,\eta)\sqrt{\frac{1}{\beta\eta}}
    \leq
    \frac{\eps}{8},
    \qquad
    C_1(\xi)\sqrt{\frac{4}{\beta\eta}\sqrt{\frac{2}{\pi}}}
    \leq
    \frac{\eps}{8}.
    \]
    Finally choose $\delta>0$ so small, depending on $(\eps,\beta,\eta,\xi)$, that 
    $C_1(\xi)\sqrt{\frac{2\delta}{1-e^{-2}}}
    \leq
    \frac{\eps}{8}$
    and finally take $N$ sufficiently large depending on these choices.
    For the moment let
    \[
    \bm_N
    \equiv
    \bm_{\beta,\eta}(H_N,\bh)
    =
    \la \bsig \ra_{\beta,\eta}.
    \]
    We first show that $\bm_N$ is already a near-ground state for the original Hamiltonian $H_N$, up to the fact that it does not lie in $\bSig_N$. Let
    \[
    q_N
    \equiv
    \|\bm_N\|^2/N
    =
    \la \la \bsig^1,\bsig^2\ra/N \ra_{\beta,\eta},
    \]
    where $\bsig^1,\bsig^2$ are independent Gibbs samples from $H_N^\eta$.
    Since $\partial_{h_i} H_N^\eta(\bsig)=\eta \sigma_i$, Gaussian integration by parts in the field variables gives
    \[
    \frac{\eta}{N}
    \sum_{i=1}^N
    \bbE\big[h_i \la \sigma_i \ra_{\beta,\eta}\big]
    =
    \beta \eta^2
    \bbE\lt[
    1-q_N
    \rt].
    \]
    On the other hand,
    \[
    \frac{\eta}{N}
    \sum_{i=1}^N
    \bbE\big[h_i \la \sigma_i \ra_{\beta,\eta}\big]
    \leq
    \frac{\eta}{N}
    \sum_{i=1}^N
    \bbE[|h_i|]
    =
    \eta \sqrt{\frac{2}{\pi}},
    \]
    and therefore
    \begin{equation}
    \label{eq:barycenter-self-overlap-near-one}
    \bbE[1-q_N]
    \leq
    \frac{1}{\beta\eta}\sqrt{\frac{2}{\pi}}.
    \end{equation}
    Next consider the free energy function
    \[
    \psi_N(\beta)
    \equiv
    \frac{1}{N}
    \log\sum_{\bsig\in \bSig_N} e^{\beta H_N^\eta(\bsig)}.
    \]
    Since $\psi_N$ is convex and $\psi_N(0)=\log 2$, we have
    \[
    \psi_N(\beta)
    \leq
    \beta \psi_N'(\beta) + \log 2,
    \]
    while trivially $\psi_N(\beta)\geq \beta \max_{\bsig} H_N^\eta(\bsig)/N$.
    Recalling that $\psi_N'(\beta)=\la H_N^\eta(\bsig)\ra_{\beta,\eta}/N$, it follows that
    \begin{equation}
    \label{eq:gibbs-average-near-max}
    \frac{1}{N}\la H_N^\eta(\bsig)\ra_{\beta,\eta}
    \geq
    \max_{\bsig\in\bSig_N} H_N^\eta(\bsig)/N
    -
    \frac{\log 2}{\beta}.
    \end{equation}
    On the event $\cG_N$,
    \begin{align}
    \label{eq:barycenter-close-to-gibbs-average}
    \Big|
    H_N^\eta(\bm_N)
    -
    \la H_N^\eta(\bsig)\ra_{\beta,\eta}
    \Big|
    &\leq
    C_2(\xi,\eta)\sqrt{N}\la \|\bsig-\bm_N\|\ra_{\beta,\eta}
    \notag\\
    &\leq
    C_2(\xi,\eta)N\sqrt{1-q_N},
    \end{align}
    because
    \[
    \la \|\bsig-\bm_N\|^2 \ra_{\beta,\eta}
    =
    N-\|\bm_N\|^2
    =
    N(1-q_N).
    \]
    Using \eqref{eq:barycenter-self-overlap-near-one}, \eqref{eq:gibbs-average-near-max}, \eqref{eq:barycenter-close-to-gibbs-average}, and the exponentially small probability of $\cG_N^c$, we obtain
    \begin{equation}
    \label{eq:barycenter-near-perturbed-ground-state}
    \liminf_{N\to\infty}
    \bbE\lt[
    H_N^\eta(\bm_N)/N
    \rt]
    \geq
    \liminf_{N\to\infty}
    \bbE\lt[
    \max_{\bsig\in\bSig_N} H_N^\eta(\bsig)/N
    \rt]
    -
    \frac{\log 2}{\beta}
    -
    C_2(\xi,\eta)\sqrt{\frac{1}{\beta\eta}}.
    \end{equation}
    Finally, for every $\bx\in[-1,1]^N$,
    \[
    |H_N^\eta(\bx)-H_N(\bx)|
    \leq
    \eta \sum_{i=1}^N |h_i|,
    \]
    so dividing by $N$ and taking expectations shows
    \begin{equation}
    \label{eq:perturbation-small}
    \limsup_{N\to\infty}
    \bbE\lt[
    \sup_{\bx\in[-1,1]^N}
    |H_N^\eta(\bx)-H_N(\bx)|/N
    \rt]
    \leq
    \eta \sqrt{\frac{2}{\pi}}.
    \end{equation}
    Combining \eqref{eq:barycenter-near-perturbed-ground-state}, \eqref{eq:perturbation-small}, and the definition of $\GS(\xi)$ gives
    \begin{equation}
    \label{eq:barycenter-near-original-ground-state}
    \liminf_{N\to\infty}
    \bbE\lt[
    H_N(\bm_N)/N
    \rt]
    \geq
    \GS(\xi)
    -
    2\eta \sqrt{\frac{2}{\pi}}
    -
    \frac{\log 2}{\beta}
    -
    C_2(\xi,\eta)\sqrt{\frac{1}{\beta\eta}}.
    \end{equation}

    We now approximate $\bm_N$ by a low-degree polynomial in the disorder.
    For fixed $\bh$, let $\bm_N^{\leq D}$ denote the degree-$\leq D$ Hermite truncation of $\bm_N$ in the Gaussian coordinates of $H_N$.
    By Proposition~\ref{prop:ams23-tv-stability} with $\tau=\delta/4$, there exists $c=c(\beta,\eta,\xi,\delta)>0$ such that for
    \[
    p_N
    =
    1-\frac{c}{N}
    \]
    and all sufficiently large $N$, if $\mu_N,\mu_{N,p_N}$ denote the Gibbs measures associated to $(H_N,\bh)$ and $(H_{N,p_N},\bh)$, then
    \[
    \bbE\lt[
    \TV(\mu_N,\mu_{N,p_N})
    \rt]
    \leq
    \frac{\delta}{4}.
    \]
    If $(\bsig,\brho)$ is a maximal coupling of $(\mu_N,\mu_{N,p_N})$, then $\bbP[\bsig\neq\brho\,|\,H_N,H_{N,p_N},\bh]=\TV(\mu_N,\mu_{N,p_N})$ and $\|\bsig-\brho\|^2\leq 4N\,\ind\{\bsig\neq\brho\}$, so
    \begin{equation}
    \label{eq:barycenter-L2-stable}
    \limsup_{N\to\infty}
    \frac{1}{N}
    \bbE\lt[
    \|\bm_N(H_N,\bh)-\bm_N(H_{N,p_N},\bh)\|^2
    \rt]
    \leq
    4\limsup_{N\to\infty}
    \bbE\lt[
    \TV(\mu_N,\mu_{N,p_N})
    \rt]
    \leq
    \delta.
    \end{equation}
    We now choose
    \[
    D_N
    =
    \Big\lceil \frac{2N}{c}\Big\rceil,
    \qquad
    \cA_N^\circ(H_N,\bh)
    \equiv
    \bm_N^{\leq D_N}(H_N,\bh).
    \]
    Since $1-p_N^{D_N+1}\geq 1-e^{-2}$, Lemma~\ref{lem:hermite-truncation-noise-stability} and \eqref{eq:barycenter-L2-stable} show
    \begin{equation}
    \label{eq:low-degree-barycenter-approximation}
    \limsup_{N\to\infty}
    \frac{1}{N}
    \bbE\lt[
    \|\cA_N^\circ(H_N,\bh)-\bm_N\|^2
    \rt]
    \leq
    \frac{\delta}{2(1-e^{-2})}.
    \end{equation}

    Finally let $\wh\bsig_N=\sign(\cA_N^\circ(H_N,\bh))$.
    Since $(1-|x|)_+\leq 1-|y|+|x-y|$ for all $x,y\in \bbR$, we have
    \begin{align}
    \|\wh\bsig_N-\bm_N\|^2
    &\leq
    2\|\wh\bsig_N-\cA_N^\circ(H_N,\bh)\|^2
    +
    2\|\cA_N^\circ(H_N,\bh)-\bm_N\|^2
    \notag\\
    &\leq
    4\|\cA_N^\circ(H_N,\bh)-\bm_N\|^2
    +
    4\sum_{i=1}^N (1-|\bm_{N,i}|)^2
    \notag\\
    &\leq
    4\|\cA_N^\circ(H_N,\bh)-\bm_N\|^2
    +
    4(N-\|\bm_N\|^2),
    \label{eq:round-low-degree-barycenter}
    \end{align}
    because $(1-|a|)^2\leq 1-a^2$ for $|a|\leq 1$.
    On $\cG_N$, the mean value theorem therefore yields
    \[
    |H_N(\wh\bsig_N)-H_N(\bm_N)|
    \leq
    C_1(\xi)\sqrt{N}\|\wh\bsig_N-\bm_N\|.
    \]
    Combining this with \eqref{eq:barycenter-self-overlap-near-one}, \eqref{eq:low-degree-barycenter-approximation}, \eqref{eq:round-low-degree-barycenter}, and the exponentially small failure probability of $\cG_N$, we conclude
    \[
    \limsup_{N\to\infty}
    \frac{1}{N}
    \bbE\lt[
    |H_N(\wh\bsig_N)-H_N(\bm_N)|
    \rt]
    \leq
    C_1(\xi)
    \sqrt{
    \frac{2\delta}{1-e^{-2}}
    +
    \frac{4}{\beta\eta}\sqrt{\frac{2}{\pi}}
    }.
    \]
    Together with \eqref{eq:barycenter-near-original-ground-state}, this gives
    \begin{align*}
    \liminf_{N\to\infty}
    \bbE\lt[
    H_N(\cA_N(H_N,\omega))/N
    \rt]
    \geq
    \GS(\xi)
    -
    2\eta \sqrt{\frac{2}{\pi}}
    -
    \frac{\log 2}{\beta}
    -
    C_2(\xi,\eta)\sqrt{\frac{1}{\beta\eta}}
    -
    C_1(\xi)
    \sqrt{
    \frac{2\delta}{1-e^{-2}}
    +
    \frac{4}{\beta\eta}\sqrt{\frac{2}{\pi}}
    }.
    \end{align*}
    Since $\sqrt{a+b}\leq \sqrt a + \sqrt b$, our choice of parameters guarantees that the total error on the right-hand side is at most $\eps$.
    This proves the theorem.
\end{proof}

\begin{remark}\label{rmk:low-deg-upper-bound-sphere}
If one uses a full generic perturbation instead of only the external field $\bh$, then the results from \cite[Chapter 3]{panchenko2013parisi} imply that
\[
\|\bm_N\|^2/N
=
\la R_{1,2}\ra
\]
converges to the first moment of the corresponding Parisi measure.
However this is not needed for our argument.
\end{remark}

\begin{remark}
The same argument gives a degree $O(N)$ upper bound for spherical spin glasses.
One replaces the sign map by radial projection $\bx\mapsto \sqrt{N}\bx/\|\bx\|$.
The analogue of \eqref{eq:gibbs-average-near-max} is then obtained by a standard argument: on the event that $H_N$ has Lipschitz constant $O(\sqrt{N})$ on $\cS_N$, a fixed-radius cap around a maximizer has measure $e^{-O(N)}$, so the spherical free energy is within $O(\beta^{-1})$ of the maximum; see e.g.\ \cite[Section~2]{chen2017parisi}.
\end{remark}

\section{Spherical Langevin Dynamics Does Not Find Wells}

In this section, we prove Theorem~\ref{thm:langevin-fails-pseudo-wells}, that Langevin dynamics does not find wells.
Below, $H_N$ will now be a mixed $p$-spin Hamiltonian from \eqref{eq:def-hamiltonian}, with fixed correlation function $\xi(t)=\sum_{k=2}^{\bar k} \gamma_k^2 t^k$.
In Subsection~\ref{subsec:main-langevin-proof} we state a few useful lemmas and use them to deduce Theorem~\ref{thm:langevin-fails-pseudo-wells}.
We then verify these lemmas in the following two subsections.

\subsection{Main Argument}
\label{subsec:main-langevin-proof}

Our proof of Theorem~\ref{thm:langevin-fails-pseudo-wells} studies multiple trajectories of Langevin dynamics on the same Hamiltonian $H_N$, but with correlated initialization and driving Brownian motion.
We first explain the pair-wise correlation structure that will be used.
Take $\bg,\wt\bg\stackrel{IID}{\sim}\cN(0,I_N)$ and $\bB_{[0,T]},\wt\bB_{[0,T]}$ to be IID Brownian motions on $\bbR^N$, and set:
\begin{equation}
\label{eq:langevin-interpolation}
\begin{aligned}
    \bg_{p}&=p\,\bg+\sqrt{1-p^2}\,\wt\bg,
    \\
    \bB_{t,p}&=p\,\bB_t + \sqrt{1-p^2}\,\wt\bB_{t}.
\end{aligned}
\end{equation}
We then let $\bx_{[0,T],p}$ be the trajectory of Langevin dynamics initialized at $\frac{\bg_{p}\sqrt{N}}{\|\bg_{p}\|}$ and driven by $\bB_{t,p}$.
Then $\bx_{[0,T],p}$ has the same law as $\bx_{[0,T]}$, even after conditioning on $H_N$.
We say this pair of Langevin trajectories is $p$-correlated.

The first lemma we will use ensures concentration of the overlaps between these correlated trajectories, and of the gradient norm for a single trajectory.
This is proved in Subsection~\ref{subsec:langevin-concentration}.

\begin{lemma}
\label{lem:quantities-concentrate}
For any $T$ the following quantities concentrate exponentially in the sense that for any $\delta > 0$, they lie in an interval $I_N(T,\delta)\subseteq \bbR$ of length $\delta$ with probability $1-e^{-cN}$ for some $c(T,\delta)$, when $N$ is sufficiently large. Furthermore the value of $c$ is uniform in the value $p\in [0,1]$.
\begin{enumerate}[label=(\Roman*)]
    \item
    \label{it:real-langevin-overlap-concentration}
    $\la\bx_{T},\bx_{T,p}\ra/N$.
    \item
    \label{it:real-langevin-gradient-concentration}
    $\|\nabla H_N(\bx_{T})\|/\sqrt{N}$.
\end{enumerate}
\end{lemma}

Next for each $N$ and time $T$, define the correlation function 
\[
\chi_{N,T}(p)
=
\bbE[\la\bx_{T},\bx_{T,p}\ra/N]
\]
We will use the following two lemmas on its behavior.
Lemma~\ref{lem:overlap-monotone} is proved here, while Lemma~\ref{lem:langevin-orthogonal} is proved in Subsection~\ref{subsec:orthogonal-langevin}.

\begin{lemma}
\label{lem:overlap-monotone}
For any $T>0$, the function $p\mapsto \chi_{N,T}(p)$
is increasing.
\end{lemma}

\begin{proof}[Proof of Lemma~\ref{lem:overlap-monotone}]
    We can identify $(H_N,\bg,\bB_{[0,T]})$ with a countably infinite sequence of IID standard Gaussian variables, as follows.
    We identify $H_N$ with its disorder coefficients $g_{i_1,\ldots,i_k}$ and $\bg$ with its Gaussian entries.
    We then identify each coordinate $(\bB_{[0,T]})_i$ (a one-dimensional standard Brownian motion on $[0,T]$) of $\bB_{[0,T]}$ with the countably infinite sequence $((\bB_T)_i / \sqrt{T}, G_{i,0}, G_{i,1}, \ldots)$ for $G_{i,j}$ defined as follows.
    Define the Brownian bridge
    \[
    	(\widehat \bB_t)_i = (\bB_t)_i - \fr{t}{T} (\bB_T)_i,
    \]
    and let $G_{i,j}$ be the normalized Fourier coefficient
    \[
    	G_{i,j} = Z_{i,j} / \bbE[Z_{i,j}^2]^{1/2}, \qquad
    	Z_{i,j} = \int_0^T \sin\lt(\fr{j \pi t}{T}\rt) (\widehat \bB_t)_i \,\de t.
    \]
    The multivariate Hermite polynomials corresponding to these Gaussians form an orthonormal basis of the space of $L^2$ functions of $(H_N,\bg,\bB_{[0,T]})$ (see e.g. \cite[Theorem 8.1.7]{lunardi2015infinite}).

    Each coordinate $(\bx_T)_i$ is a bounded, and thus $L^2$, function of $(H_N,\bg,\bB_{[0,T]})$.
    We can thus write $\bx_T$ in this Hermite basis.
    We write $\bx_{T,p}$ in the analogous Hermite basis for $(H_N,\bg_{p},\bB_{[0,T],p})$, where $\bg_{p},\bB_{[0,T],p}$ are $p$-correlated with $\bg,\bB_{[0,T]}$.
    The identity \eqref{eq:hermite-correlation} gives the desired monotonicity.
\end{proof}

\begin{lemma}
\label{lem:langevin-orthogonal}
For any fixed $T$ and $\xi$, we have $\lim_{N\to\infty} \chi_{N,T}(0)=0$.
\end{lemma}

We note that Lemma~\ref{lem:langevin-orthogonal} is the only place in this section where the absence of external field is used.
Additionally when $\xi$ is an even polynomial, one has $\chi_{N,T}(0)=0$ exactly at finite $N$ because $\bbE[\bx_T|H_N]=0$ by antipodal symmetry.
The proof for general $\xi$ requires more work, and is based on state evolution analysis of an associated approximate message passing iteration.

Now we turn to proving Theorem~\ref{thm:langevin-fails-pseudo-wells}.
We will consider $N$ simultaneously correlated trajectories (a large constant depending on other dimension-free parameters would also suffice).
Thus for $1\leq n\leq N$ let $\bx_{t,p,n}$ be a conditionally independent (given $(\bg,\bB_{[0,T]},H_N)$) copy of the correlated dynamics constructed previously:
\begin{align}
    \bg_{p,i}&=p\bg+\sqrt{1-p^2}\wt\bg_{n},
    \\
    \bB_{t,p,n}&=p\bB_t + \sqrt{1-p^2}\wt\bB_{t,n}.
\end{align}
Note that $(\bx_{T,p,n},\bx_{T,p,m})$ are $p^2$-correlated for $n\neq m$.


\begin{proof}[Proof of Theorem~\ref{thm:langevin-fails-pseudo-wells}]
    We choose $\eps>0$ small depending on $(\gamma,\xi)$, and $\delta$ small depending on $(\eps,\gamma,\xi)$.
    We then let $\beta,T$ be arbitrary and take $N$ large depending on all of these quantities.
    In light of Lemma~\ref{lem:langevin-orthogonal} and the trivial identity $\chi_{N,T}(1)=1$, for large $N$ there exists $p\in (0,1)$ such that
    \[
    \chi_{N,T}(p)=1-\eps.
    \]
    Then Lemma~\ref{lem:overlap-monotone} implies that
    \begin{equation}
    \label{eq:almost-monotone-consequence}
    \chi_{N,T}(p^2)\leq 1-\eps.
    \end{equation}
    We fix this value of $p$ below and consider the behavior of $\bx_{T}$ as well as the correlated outputs $\breve\bx_{T,n}\equiv \bx_{T,p,n}$ for each $1\leq n\leq N$.
    We assume the conclusions of Proposition~\ref{prop:gradients-bounded} and Lemma~\ref{lem:spectrum-approx} hold, and the conclusion of Lemma~\ref{lem:quantities-concentrate} holds for all $\binom{N+1}{2}$ pairs among $(\bx_T,\breve\bx_{T,1},\ldots,\breve\bx_{T,N})$ (with correlation value $p$ for pairs $(\bx_T,\breve \bx_{T,n})$ and $p^2$ for pairs $(\bx_{T,m},\breve \bx_{T,n})$.)
    By a union bound, this occurs with probability $1-e^{-cN}$. (Note that the conclusion of Lemma~\ref{lem:spectrum-approx} is monotone in $\delta$, and therefore holds with probability $1-e^{-cN}$ if $\delta$ is taken sufficiently small.)

    We will show that on this event,  $\bx_{T}$ is a not $(\gamma,\delta)$-well.
    Assume otherwise for the sake of contradiction.


    For each $1\le n \le N$, let $\by_{n}\in \bx_{T}^{\perp}$ be the unit tangent vector to $\cS_N$ at $\bx_{T}$ in the direction pointing from $\bx_{T}$ to $\breve\bx_{T,n}$ (i.e. such that the shortest path geodesic from $\bx_{T}$ in direction $\by_n$ passes through $\wt\bx_{T,n}$ within distance $O(\sqrt{\eps N})$).
    We first argue that $\nabla_{\sph}^2 H_N(\bx_{T})\by_n$ must have small norm for each $n$.
    Indeed, an elementary Taylor expansion of the gradient (using Proposition~\ref{prop:gradients-bounded} to control the error term) shows
    \[
    \|\nabla_{\sph} H_N(\breve\bx_{T,n})\|
    \geq
    \sqrt{\eps N}
    \|\nabla_{\sph}^2 H_N(\breve\bx_{T,n})\cdot \by_n\|
    -
    O(\eps \sqrt{N})
    -
    \|\nabla_{\sph} H_N(\bx_{T})\|.
    \]
    As we assumed the conclusions of Lemma~\ref{lem:quantities-concentrate} hold and that $\bx_{T}$ is a $(\gamma,\delta)$-well, we have also
    \[
    \|\nabla_{\sph} H_N(\breve \bx_{T,n})\|
    \leq
    \|\nabla_{\sph} H_N(\bx_{T})\|+\delta\sqrt{N}
    \leq
    2\delta\sqrt{N}.
    \]
    Since $\delta$ is small depending on $\eps$, rearranging the preceding two displays shows that
    \begin{equation}
    \label{eq:yn-near-hessian-kernel}
    \|\nabla_{\sph}^2 H_N(\breve\bx_{T,n})\cdot \by_n\|
    \leq
    O(\sqrt{\eps}).
    \end{equation}
    Next we use Lemma~\ref{lem:spectrum-approx}, with the same constant $\delta$ as above.
    Combined with the definition of $(\gamma,\delta)$-wells, it implies that for a constant $K(\gamma,\eps,\delta)$ we have
    \[
    \lambda_K(\nabla_{\sph}^2 H_N(\bx_{T}))
    \leq -\gamma/2.
    \]
    Thus, let us define $S_K\subseteq \bx_{T}^{\perp}$ to be the span of the top $K$ eigenvectors of $\nabla_{\sph}^2 H_N(\bx_{T})$.
    Then we may decompose each $\by_n$ into
    \[
    \by_n = \bv_n + \bw_n
    \]
    with $\bv_n\in S_K$ and $\bw_n\in S_K^{\perp}\cap \bx_{T}^{\perp}$.
    Then \eqref{eq:yn-near-hessian-kernel} implies that for each $n$, we have:
    \[
    \|\bw_n\|\leq O(\sqrt{\eps}/\gamma).
    \]
    Since $N$ is much larger than $K$, the pigeonhole principle implies that there exist $n,m$ with $\|\bv_n-\bv_m\|\leq \eps$.
    This means
    \begin{equation}
    \label{eq:pigeonhole-tangent}
        \|\by_n-\by_m\|\leq O(\eps + \sqrt{\eps}/\gamma)\leq O(\sqrt{\eps}/\gamma)
    \end{equation}
    By the assumption that Lemma~\ref{lem:quantities-concentrate} is in force, the distances $\|\bx_{T}-\breve\bx_{T,n}\|$ and $\|\bx_{T}-\breve\bx_{T,m}\|$ differ by at most $\eps\sqrt{N}$, and are at each at most $O(\sqrt{\eps N})$.
    It follows from \eqref{eq:pigeonhole-tangent} that
    \begin{equation}
    \label{eq:pigeonhole-outputs}
        \|\breve\bx_{T,n}-\breve\bx_{T,m}\|\leq O(\eps\sqrt{N}).
    \end{equation}
    However this means $\la \breve\bx_{T,n},\breve\bx_{T,m}\ra/N \geq 1-O(\eps^2)$.
    This contradicts \eqref{eq:almost-monotone-consequence}, completing the proof.
\end{proof}

\subsection{Concentration Properties of Langevin Dynamics}
\label{subsec:langevin-concentration}

Here we prove Lemma~\ref{lem:quantities-concentrate} on concentration of overlaps and gradient norms.
For our analysis, it will be useful to consider an approximation to the Langevin dynamics which enjoys additional Lipschitz properties.
Given $\bg\in\bbR^N$ we let
\[
\hat \bg
=
\frac{\bg}{\max(1/2,\|\bg\|/\sqrt{N})}.
\]
We take $\bg\sim\cN(0,I_N)$ and use it to couple the uniformly initialized Langevin dynamics with a proxy process $\bx_{[0,T]}^{(K)}$ using the same driving $\bB_{[0,T]}$ and initializations:
\begin{equation}
\label{eq:proxy-coupling-init}
\begin{aligned}
    \bx_0&=\bg\sqrt{N}/\|\bg\|,
    \\
    \bx_0^{(K)}&=\hat\bg.
\end{aligned}
\end{equation}
The point is simply that $\bx_0^{(K)}$ is a Lipschitz function of $\bg$ such that $\bx=\bx_0^{(K)}$ with probability $1-e^{-cN}$.
Here the constant $K>0$ is a parameter in the latter process. With $\bar k$ as in \eqref{eq:def-hamiltonian}, we define:
\[
f_K(r)=K(r-1)^2 + (r^2-1)^{\bar k}.
\]
Then the auxiliary full-space diffusion $\bx_{[0,T]}^{(K)}$ is defined by the SDE:
\begin{equation}
\label{eq:full-space-langevin}
\de \bx_t^{(K)}
=
\Big(
\beta \nabla H_N(\bx_t^{(K)})
-
f_K'(\|\bx_t^{(K)}\|^2/N)\bx_t^{(K)}
\Big)\de t
+
\sqrt{2}\,\de \bB_t.
\end{equation}

We will view $\bx_T^{(K)}$ as a function of $(\bg,\bG^{(2)},\dots,\bG^{(\bar k)},\bB_{[0,T]})$, typically abbreviated $(\bg,H_N,\bB_{[0,T]})$
We metrize $\bB_{[0,T]}$ by the supremum norm $\|\bB_{[0,T]}-\bB'_{[0,T]}\|=\sup_{t\in [0,T]}\|\bB_t-\bB'_t\|$, and the remaining arguments via the (un-normalized) Euclidean norm so that such triples lie within a metric space $\cM=\cM_N$.
We endow $\cM$ with the product measure $\mu$, which is given by the usual Wiener measure on the last component and standard Gaussian measure on the other components.

\begin{proposition}
\label{prop:langevin-approx-lipschitz}
For positive $T,\eps$, there exist constants $L,K,c$ such that for $N$ sufficiently large, there exists $\cM^{\circ}=\cM^{\circ}_N(T,\eps,L,K,c)\subseteq \cM$ with $\mu(\cM^{\circ})\geq 1-e^{-cN}/2$ such that the following hold.
\begin{enumerate}[label=(\alph*)]
    \item
    \label{it:langevin-approx-lipschitz-1}
    The restriction $\bx_T^{(K)}\big|_{\cM^{\circ}}:\cM^{\circ}\to\bbR^N$ is an $L$-Lipschitz function.
    \item
    \label{it:langevin-approx-lipschitz-2}
    For the coupling of \eqref{eq:proxy-coupling-init}, then
    $\|\bx_T^{(K)}-\bx_T\|/\sqrt{N}\leq \eps$ holds for all $(\bx_0,H_N,\bB_{[0,T]})\in\cM^{\circ}$.
\end{enumerate}
\end{proposition}

\begin{proof}
    We choose $K$ large depending on $(T,\eps)$ and $L$ large depending on $(K,T,\eps)$ and finally $c$ small depending on $(L,K,T,\eps)$.
    Note that $\bg\mapsto \bx_0^{(K)}$ is $2$-Lipschitz.
    By composition, \cite[Lemma 2.6]{ben2006cugliandolo} shows that $\bx_t^{(K)}$ is $L$-Lipschitz on a subset of $\cM$ with $\mu$-measure at least $1-e^{-cN}/4$.
    Noting that $\bx_0=\bx_0^{(K)}$ with probability $1-e^{-cN}$, \cite[Lemma 3.1]{sellke2023threshold} proves that
    \[
    \bbP[\|\bx_T^{(K)}-\bx_T\|/\sqrt{N}\leq \eps/2]\geq 1-e^{-cN}/4.
    \]
    Defining $\cM^{\circ}$ by the intersection of these two events completes the proof.
\end{proof}

Let $\ocM$ be the space of tuples $(\bg,\wt\bg,\bB_{[0,T]},\wt\bB_{[0,T]},H_N)$, endowed with the natural product measure $\omu$.
Fixing $p\in [0,1]$ and a choice of constants $(T,\eps,L,K,c)$ defining $\cM^{\circ}$, we let $\ocM_{p}^{\circ}(T,\eps,L,K,c)\subseteq\ocM$ consist of those tuples such that
\[
(\bg,\bB_{[0,T]},H_N),
(\bg_{p},\bB_{[0,T],p},H_N),
\in \cM^{\circ}(T,\eps,L,K,c).
\]
Since $(\bg_{p},\bB_{t,p},H_N)\stackrel{d}{=}(\bg,\bB_{[0,T]},H_N)$ for any $p$, we have $\omu(\ocM_{p}^{\circ})\geq 1-e^{-cN}$ when the relevant parameters are as in Proposition~\ref{prop:langevin-approx-lipschitz}.


\begin{proposition}
\label{prop:langevin-concentration}
For any $T,\eps,K$, there exist $L,c$ such that for $N$ sufficiently large, uniformly in $(p,p')$, the following quantities are $L/\sqrt{N}$-Lipschitz functions of $(\bg, \wt\bg, \bB_{[0,T]},\wt\bB_{[0,T]},H_N)$ on $\ocM^\circ_{p}(T,\eps,L,K,c)$:
\begin{enumerate}[label=(\roman*)]
    \item
    \label{it:langevin-overlap-concentration}
    $\la\bx_{T}^{(K)},\bx_{T,p}^{(K)}\ra/N$.
    \item
    \label{it:langevin-gradient-concentration}
    $\|\nabla H_N(\bx_{T}^{(K)})\|/\sqrt{N}$.
\end{enumerate}
\end{proposition}

\begin{proof}
    We focus on \ref{it:langevin-overlap-concentration}, explaining the small changes for \ref{it:langevin-gradient-concentration} at the end.
    Note that the map $(\bx,\bx')\mapsto \la\bx,\bx'\ra/N$ is $O(1/\sqrt{N})$-Lipschitz on the set $\{(\bx,\bx'):\|\bx\|,\|\bx'\|\leq 2\sqrt{N}\}$.
    By Proposition~\ref{prop:langevin-approx-lipschitz}\ref{it:langevin-approx-lipschitz-2}, this condition holds for $(\bx_{T,p_1,1}^{(K)},\bx_{T,p_2,2}^{(K)})$ for data within the set $\ocM^{\circ}_{p_1,p_2}$.
    By composing with Proposition~\ref{prop:langevin-approx-lipschitz} and \eqref{eq:langevin-interpolation}, we find that
    \[
    (\bg, \wt\bg, \bB_{[0,T]},\wt\bB_{[0,T]},H_N) \mapsto
    \la\bx_{T,p_1,1}^{(K)},\bx_{T,p_2,2}^{(K)}\ra/N
    \]
    is $O(1/\sqrt{N})$ Lipschitz on $\ocM^\circ_{p}(T,\eps,L,K,c)$, proving \ref{it:langevin-overlap-concentration}.

    For case \ref{it:langevin-gradient-concentration}, the only change is that we may without loss of generality assume $\cM^{\circ}$ in Proposition~\ref{prop:langevin-approx-lipschitz} includes only those $H_N$ obeying Proposition~\ref{prop:gradients-bounded}.
    Then $\bx\mapsto \|\nabla H_N(\bx)\|/\sqrt{N}$ is $O(1/\sqrt{N})$-Lipschitz on $\|\bx\|\leq 2\sqrt{N}$.
    This implies a similar Lipschitz guarantee on $\ocM^{\circ}_{p}$, again by composition.
\end{proof}

We now deduce Lemma~\ref{lem:quantities-concentrate} from Proposition~\ref{prop:langevin-concentration}.

\begin{proof}[Proof of Lemma~\ref{lem:quantities-concentrate}]
    Since both quantities \ref{it:langevin-overlap-concentration}, \ref{it:langevin-gradient-concentration} in Proposition~\ref{prop:langevin-concentration} are scalar functions, by Kirszbraun's extension theorem they admit extensions to all of $\ocM$ with the same Lipschitz constant.
    Recalling Proposition~\ref{prop:gradients-bounded}, the estimate from Proposition~\ref{prop:langevin-approx-lipschitz} shows that \ref{it:langevin-overlap-concentration}, \ref{it:langevin-gradient-concentration} are respectively within $\delta$ of \ref{it:real-langevin-overlap-concentration}, \ref{it:real-langevin-gradient-concentration} with probability $1-e^{-cN}$.
    Since $(\ocM,\omu)$ obeys a logarithmic Sobolev inequality with constant $O(1)$ (see \cite{ben2006cugliandolo} for more discussion), applying concentration of measure for Lipschitz functions on $\ocM$ completes the proof.
\end{proof}

\subsection{Orthogonality of Independent Langevin Trajectories}
\label{subsec:orthogonal-langevin}

Here we prove Lemma~\ref{lem:langevin-orthogonal}.
The method uses a different auxiliary algorithm known as approximate message passing (AMP) to approximate the Langevin dynamics.
This technique was introduced by \cite{celentano2021high} to study noise-less continuous-time dynamics of a similar flavor (in fact it is suggested therein that the method should also suffice to analyze Langevin dynamics).
However we do not derive the equations governing the dynamics, but only verify the far simpler Lemma~\ref{lem:langevin-orthogonal}.
We note that AMP can be shown to have dimension-free Lipschitz constant in various parameters (see \cite[Section 6]{gamarnik2019overlap} 
), and thus should also suffice to imply Lemma~\ref{lem:quantities-concentrate}.
However we felt the proof in the previous subsection was more approachable.

Fixing $\eta>0$, the relevant AMP is initialized at $\bw_0^{\eta}\sim \cN(0,I_N)$ and defined for recursively determined ($N$-independent) constants $(A_{j,k}^{\eta})_{0\leq j\leq k}$ and $(B_{j,k}^{\eta})_{k\geq 0}$ by
\begin{equation}
\label{eq:discrete-time-AMP}
\begin{aligned}
\by^{\eta}_{k}
&=
\beta \eta\nabla H_N(\bx^{\eta}_k)
-
\ons^{\eta}_k;
\\
\bx^{\eta}_{k}
&=
\frac{\bw^{\eta}_{k}}{\sqrt{C^\eta_{k}}};
\\
\bw^{\eta}_{k}
&=
\bx_{k-1}^{\eta}
+
\beta \eta\nabla H_N(\bx^{\eta}_{k-1})
+
\sqrt{2\eta} \bg_{k-1}
\\
&=
\bx^{\eta}_{k-1}
+
\by^{\eta}_{k-1}
+
\ons^{\eta}_{k-1}
+
\sqrt{2\eta} \bg_{k-1}
\\
&=
\sum_{j=0}^{k-1}
\Big(
A_{j,k}^{\eta} \by^{\eta}_{j}
+
B_{j,k}^{\eta} \bg_j
\Big)
.
\end{aligned}
\end{equation}
Here
\[
C_{k}^{\eta}=\plim_{N\to\infty}\la \bw^{\eta}_{k},\bw^{\eta}_{k}\ra/N
\]
is determined by the state evolution recursion, and the Onsager terms $\ons^{\eta}_k$ are described below.
AMP algorithms of this type have been studied since \cite{bolthausen2014iterative,donoho2009message,bayati2011dynamics,javanmard2013state}; in these works, $H_N$ is quadratic and so $\nabla H_N(\cdot)$ is just multiplication by the corresponding random matrix.
The specific iteration above falls under the framework of \cite[Theorem 2]{huang2024optimization} which incorporates both random tensors and external Gaussian noise $\bg_k$.
Said state evolution result characterizes the iteration above in terms of an auxiliary centered Gaussian process $(Y^{\eta}_0,Y^{\eta}_1,\dots)$ with recursively defined covariance:
\begin{equation}
\label{eq:state-evolution}
    \bbE[Y^{\eta}_{k+1}Y^{\eta}_{j+1}]
    =
    \beta^2 \eta^2 \xi'\big(\bbE[X^{\eta}_{k}X^{\eta}_{j}]\big).
\end{equation}
Here the relevant random variables are defined as follows.
$G_0,G_1,\ldots \sim \cN(0,1)$ are IID standard Gaussians.
$W^{\eta}_k$ and $X^{\eta}_k$ are defined by:
\begin{align}
\notag
W_{k+1}^{\eta}
&=
X_k^{\eta}+Y_{k}^{\eta}+\ONS^{\eta}_k+\sqrt{2\eta} G_k,
\\
\label{eq:X-amp-def}
X_{k+1}^{\eta}
&=
\frac{W_{k+1}^{\eta}}{\sqrt{\bbE[(W_{k+1}^{\eta})^2]}},
\\
\label{eq:Onsager-formula}
\ONS_{k}^{\eta}
&=
\sum_{i=0}^{k-1}
\lt[
\xi''\Big(
\bbE[X_{k}^{\eta} X_i^{\eta}]
\Big)
\cdot
\lt(\frac{A_{i,k}^{\eta}}{\sqrt{C_{k}^{\eta}}}\rt)
\cdot
X_i^{\eta}
\rt]
\\
\label{eq:coefficients-determined}
&\equiv
\sum_{j=0}^k
\Big(
A_{j,k+1}^{\eta} Y^{\eta}_{j}
+
B_{j,k+1}^{\eta} G_j,
\Big)
.
\end{align}
Here the upper-case $\ONS$ random variables are defined on the same probability space as the $X,Y$ variables.
The formula \eqref{eq:Onsager-formula} is most easily read off from \cite[Equation (3.2)]{ams20} (which does not technically allow for the auxiliary randomness variables $G_i$), and can also be recovered from \cite[Theorem 2]{huang2024optimization}.
This uniquely determines the coefficients in \eqref{eq:coefficients-determined}.
One can easily check that the recursions \eqref{eq:state-evolution} and below close, thus uniquely defining a joint distribution on scalar random variables.
The state evolution result from \cite[Theorem 2]{huang2024optimization} then asserts that these recursively defined variables describe the behavior of the AMP iterates in a typical coordinate.

To model independent pairs of Langevin trajectories, let $\bx^{\eta}_i,\wt\bx^{\eta}_i$ be defined by the above AMP iteration, but with independent $\bx^{\eta}_0,\wt\bx^{\eta}_0$ and independent external Gaussian vectors $\bg_i,\wt\bg_i$.
The general formulation of \cite[Theorem 2]{huang2024optimization} encompasses this iteration as well\footnote{In the notation there, one can initialize $\bw_0^{\eta}=\be^0$ and $\wt\bw_0^{\eta}=\be^1$, and then e.g.\ alternate computations on the $\bw$ variables and $\wt\bw$ variables on even and odd steps. After initialization, the external randomness $\be^t$ play the role of the vectors $\bg_k$ in \eqref{eq:discrete-time-AMP}. The multiple species there are not needed in the present formulation, i.e.\ one can take $|\sS|=1$.}.
The resulting state evolution description of this pair of correlated AMP iterations is as follows.
Let $(G_1,\dots,G_k,\wt G_1,\dots,\wt G_k)$ be IID standard Gaussians, and define recursively the centered Gaussian process $(Y^{\eta}_i,\wt Y^{\eta}_i)_{i\geq 0}$ by the generalization of \eqref{eq:state-evolution}:
\begin{equation}
\label{eq:state-evolution-pair}
\begin{aligned}
    \bbE[Y^{\eta}_{k+1}Y^{\eta}_{j+1}]
    &=
    \beta^2 \eta^2 \xi'\big(\bbE[X^{\eta}_{k}X^{\eta}_{j}]\big),
    \\
    \bbE[Y^{\eta}_{k+1}\wt Y^{\eta}_{j+1}]
    &=
    \beta^2 \eta^2 \xi'\big(\bbE[X^{\eta}_{k}\wt X^{\eta}_{j}]\big),
    \\
    \bbE[\wt Y^{\eta}_{k+1}\wt Y^{\eta}_{j+1}]
    &=
    \beta^2 \eta^2 \xi'\big(\bbE[\wt X^{\eta}_{k}\wt X^{\eta}_{j}]\big)
    .
\end{aligned}
\end{equation}
Here $\wt X_j^{\eta}$ is described by the analogs of e.g.\ \eqref{eq:X-amp-def}.

\begin{proposition}[State Evolution]
For any fixed $\eta$ and Lipschitz $\psi:\bbR^{k+1}\to\bbR$, we have the $N\to\infty$ convergence in probability in the space $\bbW_2(\bbR^{4(k+1)})$:
\begin{equation}
\label{eq:w-W-convergence}
\begin{aligned}
&\frac{1}{N}
\sum_{i=1}^N
\delta_{(\by_0^{\eta})_i,\dots,(\by_k^{\eta})_i,(\bg_0)_i,\dots,(\bg_k)_i,(\wt\by_0^{\eta})_i,\dots,(\wt\by_k^{\eta})_i,(\wt\bg_0)_i,\dots,(\wt\bg_k)_i}
\\
&\to
\cL\big(Y^{\eta}_0,\dots,Y^{\eta}_k,G_1,\dots,G_k,\wt Y^{\eta}_0,\dots,\wt Y^{\eta}_k,\wt G_1,\dots,\wt G_k\big)
.
\end{aligned}
\end{equation}
Here $\delta_{(\cdot)}$ denotes a Dirac delta mass.
\end{proposition}

The next lemma is the desired orthogonality statement for the AMP iteration. It follows readily from the state evolution equations.

\begin{lemma}
\label{lem:AMP-orthogonal}
Then for any $\eta,K>0$ we have for all $0\leq j,k\leq K$:
\[
\plim_{N\to\infty}
\la \bx^{\eta}_j,\wt\bx^{\eta}_k\ra/N
=
\plim_{N\to\infty}
\la \by^{\eta}_j,\wt\by^{\eta}_k\ra/N
=
0.
\]
\end{lemma}

\begin{proof}
    We induct on $K$, with the base case being clear as $\plim_{N\to\infty} \la\bx_0,\wt\bx_0\ra/N=0$.
    For the inductive step, we consider $0\leq j\leq K$.
    Since $\xi'(0)=0$, the state evolution recursion \eqref{eq:state-evolution-pair} implies:
    \[
    \plim_{N\to\infty}
    \la\by^{\eta}_{j+1},\wt\by^{\eta}_{K+1}\ra/N
    =
    0.
    \]
    Given the expansion \eqref{eq:discrete-time-AMP}, we see that
    \begin{align*}
    \plim_{N\to\infty}
    \la\bw^{\eta}_{j+1},\wt\bw^{\eta}_{K+1}\ra/N
    &=
    0
    \\
    \implies
    \plim_{N\to\infty}
    \la\bx^{\eta}_{j+1},\wt\bx^{\eta}_{K+1}\ra/N
    &=
    0.
    \end{align*}
    In this implication, we use that
    \[
    \plim_{N\to\infty}
    \la\bg_{j+1},\wt\bg_{K+1}\ra/N
    =
    \plim_{N\to\infty}
    \la\bg_{j+1},\wt\bx^{\eta}_{K+1}\ra/N
    =
    \plim_{N\to\infty}
    \la\bx^{\eta}_{j+1},\wt\bg_{K+1}\ra/N
    =
    0.
    \]
    The first is clear while the latter two follow by independence, and because $\plim_{N\to\infty} \|\wt\bx^{\eta}_{k}\|/\sqrt{N}$ exists and is finite for each $k$ (again by state evolution).
    This closes the induction and completes the proof.
\end{proof}

We next define an intermediate approximation
\begin{equation}
\label{eq:discrete-time-Langevin}
\bx^{\eta,\aux}_{k+1}
=
\frac{\bx^{\eta,\aux}_k
+\beta \eta\nabla H_N(\bx^{\eta,\aux}_k)+\sqrt{2\eta} \bg_k}
{\|\bx^{\eta,\aux}_k
+\beta \eta\nabla H_N(\bx^{\eta,\aux}_k)+\sqrt{2\eta} \bg_k\|}\cdot \sqrt{N}.
\end{equation}

We couple this iteration with the AMP by using the same initialization $\bx_0^{\eta,\aux}=\bx_0^{\eta}$, and the same Gaussian vectors $\bg_k$.
We then couple with Langevin dynamics by setting
\begin{equation}
	\label{eq:langevin-coupling}
	\sqrt{\eta}\bg_k=\bB_{(k+1)\eta}-\bB_{k\eta}.
\end{equation}
Lemma~\ref{lem:langevin-orthogonal} follows directly from Lemma~\ref{lem:AMP-orthogonal} above and Lemmas~\ref{lem:AMP-approximates-discrete-langevin}, \ref{lem:discrete-langevin-approx-Langevin} below.
The latter two show that the auxiliary iteration \eqref{eq:discrete-time-Langevin} accurately approximates both AMP and the continuous-time spherical Langevin dynamics.

\begin{lemma}
\label{lem:AMP-approximates-discrete-langevin}
For any $T>0$ we have
\[
\plimsup_{N\to\infty}
\sup_{0\leq k\leq T/\eta}
\|\bx^{\eta,\aux}_k-\bx^{\eta}_k\|/\sqrt{N}
=
0.
\]
\end{lemma}

\begin{proof}
    This is easily shown by induction on $k$, the base case being trivial.
    In the inductive step, we assume the result for the $k$-th iterates and show it for the $k+1$ iterates.
    Using Proposition~\ref{prop:gradients-bounded}, we see that
    \[
    \plimsup_{N\to\infty}
    \sup_{0\leq k\leq T/\eta}
    \|\bx^{\eta,\aux}_k+\beta\eta \nabla H_N(\bx^{\eta,\aux}_k)
    -\bx^{\eta}_k-\beta\eta \nabla H_N(\bx^{\eta}_k)\|/\sqrt{N}=0.
    \]
    The independence of $\bg_k$ from both iterates implies it has overlap $o(1)$ with each in probability.
    Therefore
    \[
    \plimsup_{N\to\infty}
    \frac{\|\bx^{\eta,\aux}_{k+1}\|^2
    -
    \|\bx^{\eta}_{k+1}\|^2}{N}
    =
    0.
    \]
    Therefore the denominator in \eqref{eq:discrete-time-Langevin} satisfies
    \begin{align*}
    &\plim_{N\to\infty}
    \|\bx^{\eta,\aux}_k
    +\beta \eta\nabla H_N(\bx^{\eta,\aux}_k)+\sqrt{2\eta} \bg_k\|/\sqrt{N} \\
    &=\plim_{N\to\infty}
    \|\bx^{\eta}_k
    +\beta \eta\nabla H_N(\bx^{\eta}_k)+\sqrt{2\eta} \bg_k\|/\sqrt{N} \\
    &=\plim_{N\to\infty}
    = \|\bw_{k+1}^\eta\| / \sqrt{N}
    = \sqrt{C_{k+1}^\eta}
    \end{align*}
    in probability.
    This easily completes the inductive step and hence the proof.
\end{proof}

The following elementary estimate will be useful to prove Lemma~\ref{lem:discrete-langevin-approx-Langevin} (we will take $\bx_{k\eta}$ to be the disorder-dependent ``initialization'' when applying it).

\begin{proposition}[{\cite[Lemma 2.1]{sellke2023threshold}}]
\label{prop:langevin-continuity-bound}
    Suppose $H_N$ obeys Proposition~\ref{prop:gradients-bounded}.
    For fixed $\beta$ and small enough $\eta\in (0,\eta_0(\beta))$, for any (possibly $H_N$-dependent) initialization $\bx_0\in\cS_N$ of Langevin dynamics:
    \[
    \bbP[
    \sup_{0\leq t\leq \eta}
    \|\bx_t-\bx_0\|
    \leq
    C(\beta)\sqrt{\eta N}]
    \geq 1-e^{-cN}.
    \]
\end{proposition}

\begin{lemma}
\label{lem:discrete-langevin-approx-Langevin}
Let $T>0$ be fixed.
For $\bx^{\eta,\aux}$ as in \eqref{eq:discrete-time-Langevin} and $\bx_t$ the solution to the Langevin dynamics, coupled by \eqref{eq:langevin-coupling},
\[
\lim_{\eta\to 0}
\plimsup_{N\to\infty}
\sup_{0\leq k\leq T/\eta}
\|\bx^{\eta,\aux}_k-\bx_{k\eta}\|/\sqrt{N}
=
0.
\]
\end{lemma}

\begin{proof}
    Define the random variable
    $D_{k,\eta}= \|\bx^{\eta,\aux}_k-\bx_{k\eta}\|/\sqrt{N}$.
    We will show that with probability $1-o(1)$ as $N\to\infty$, the following recursion holds for all $0\le k\le T/\eta$:
    \begin{equation}
    \label{eq:discrete-time-gronwall}
    D_{k+1,\eta}
    \leq
    (1+O(\eta))
    D_{k+1,\eta}
    +
    O(\eta^{3/2})+o_N(1).
    \end{equation}
    We will prove this by induction on $k$, with the base case $k=0$ being trivial.
    Define
    \[
    \hat\bx^{\eta}_{k+1}
    =
    \frac{\bx_{k\eta}
    +\beta \eta\nabla H_N(\bx_{k\eta})+\sqrt{2\eta} \bg_k}
    {\|\bx_{k\eta}
    +\beta \eta\nabla H_N(\bx_{k\eta})+\sqrt{2\eta} \bg_k\|}\cdot \sqrt{N}.
    \]
    We then have:
    \begin{equation}
    \label{eq:triangle-ineq-first-step-discrete-langevin}
    D_{k+1,\eta}
    \leq
    \|\bx^{\eta,\aux}_{k+1}-\hat\bx^{\eta}_{k+1}\|/\sqrt{N}
    +
    \|\hat\bx^{\eta}_{k+1}-\bx_{(k+1)\eta}\|/\sqrt{N}
    \end{equation}
    For the first term, on the event of Proposition~\ref{prop:gradients-bounded} we have
    \begin{equation}
    \label{eq:first-term-discrete-langevin-numerator}
    \big\|\big(\bx_{k\eta}
    +\beta \eta\nabla H_N(\bx_{k\eta})\big)
    -
    \big(\bx^{\eta,\aux}_k
    +\beta \eta\nabla H_N(\bx^{\eta,\aux}_k)
    \big)
    \big\|/\sqrt{N}
    \leq
    (1+O(\eta))
    D_{k,\eta}.
    \end{equation}
    We claim that
    \begin{equation}
    	\label{eq:first-term-discrete-langevin-denominator}
    	\fr{\|\bx_{k\eta}+\beta \eta\nabla H_N(\bx_{k\eta})+\sqrt{2\eta} \bg_k\|}{\sqrt{N}},
    	\fr{\|\bx^{\eta}_k+\beta \eta\nabla H_N(\bx^{\eta}_k)+\sqrt{2\eta} \bg_k\|}{\sqrt{N}}
    	= 1 + O(\eta).
    \end{equation}
    This is because $\|\bx_{k\eta}\|, \|\bx^{\eta}_k\| = \sqrt{N}$, while $\|\nabla H_N(\bx_{k\eta})\|, \|\nabla H_N(\bx^{\eta}_k)\| = O(\sqrt{N})$ by Proposition~\ref{prop:gradients-bounded}, and $\bg_k$ is independent of and thus asymptotically orthogonal to everything else, with probability $1-e^{-cN}$.
    Combining \eqref{eq:first-term-discrete-langevin-numerator}, \eqref{eq:first-term-discrete-langevin-denominator} with Lemma~\ref{lem:normalize-almost-contraction} below (with $C = 1 + O(\eta)$) implies
    \[
    	\|\bx^{\eta,\aux}_{k+1}-\hat\bx^{\eta}_{k+1}\|/\sqrt{N} \le (1+O(\eta)) D_{k,\eta}.
    \]
    It remains to estimate the second term in \eqref{eq:triangle-ineq-first-step-discrete-langevin}.
    To this end, for $\by\in\bbR^N$ we write
    $\by=\by^{\parallel}+\by^{\perp}$ for its decomposition into components parallel and orthogonal to $\bx_{k\eta}$.
    Since $\bx_{k\eta}\in\cS_N$, we have
    \begin{align}
    \notag
    \|\hat\bx^{\eta}_{k+1}-\bx_{(k+1)\eta}\|/\sqrt{N}
    &\leq
    \|(\hat\bx^{\eta}_{k+1})^{\perp}-(\bx_{(k+1)\eta})^{\perp}\|/\sqrt{N}
    +
    \|(\hat\bx^{\eta}_{k+1})^{\parallel}-(\bx_{(k+1)\eta})^{\parallel}\|/\sqrt{N}
    \\
    \notag
    &=
    \|(\hat\bx^{\eta}_{k+1})^{\perp}-(\bx_{(k+1)\eta})^{\perp}\|/\sqrt{N} \\
    \notag 
    &\qquad +
    \lt|
    \sqrt{1-\|(\hat\bx^{\eta}_{k+1})^{\perp}\|^2/N}
    -
    \sqrt{1-\|(\bx_{(k+1)\eta})^{\perp}\|^2/N}
    \rt|
    \\
    \label{eq:last-main-estimate-discrete-langevin}
    &\leq
    O(\|(\hat\bx^{\eta}_{k+1})^{\perp}-(\bx_{(k+1)\eta})^{\perp}\|/\sqrt{N})
    .
    \end{align}
    Here the last inequality holds so long as both $\|(\hat\bx^{\eta}_{k+1})^{\perp}\|$ and $\|(\bx_{(k+1)\eta})^{\perp}\|$ are at most $\sqrt{N/10}$, as the function $x \mapsto \sqrt{1-x^2}$ is $O(1)$-Lipschitz on $[0,1/\sqrt{10}]$.
    For small $\eta$, the bound $\|(\hat\bx^{\eta}_{k+1})^{\perp}\| \le \sqrt{N/10}$ follows from Proposition~\ref{prop:gradients-bounded}, while $\|(\bx_{(k+1)\eta})^{\perp}\| \le \sqrt{N/10}$ follows from Proposition~\ref{prop:langevin-continuity-bound} (and these propositions show both are $O(\sqrt{\eta N})$ with probability $1-e^{-cN}$).
    To estimate \eqref{eq:last-main-estimate-discrete-langevin}, we write
    \begin{align*}
    	\|(\hat\bx^{\eta}_{k+1})^{\perp}-(\bx_{(k+1)\eta})^{\perp}\|
    	&\le \big\|
    		(\bx_{(k+1)\eta})^{\perp}-
    		\big(\bx_{k\eta} +\beta \eta\nabla H_N(\bx_{k\eta})+\sqrt{2\eta} \bg_k\big)^{\perp}
    	\big\| \\
    	&+ \big\|
    		(\hat\bx^{\eta}_{k+1})^{\perp}-
    		\big(\bx_{k\eta} +\beta \eta\nabla H_N(\bx_{k\eta})+\sqrt{2\eta} \bg_k\big)^{\perp}
		\big\|.
    \end{align*}
    Then,
    \begin{align*}
    &(\bx_{(k+1)\eta})^{\perp}
    -
    \lt(\bx_{k\eta}
    +\beta \eta\nabla H_N(\bx_{k\eta})+\sqrt{2\eta} \bg_k\rt)^{\perp}
    \\
    &=
    \int_{k\eta}^{(k+1)\eta}
    \beta
    \Big(
    \big(\nabla_{\sph} H_N(\bx_t)\big)^{\perp}
    -
    \big(\nabla_{\sph} H_N(\bx_{k\eta})\big)^{\perp}
    \Big)
    +
    \frac{(N-1)\bx_t^{\perp}}{2N}
    \de t
    \\
    &\quad
    +
    \lt(
    \sqrt{2}
    \int_{k\eta}^{(k+1)\eta}(P^{\perp}_{\bx_t}-1)
    \de \bB_t
    \rt)^{\perp}.
    \end{align*}
    Proposition~\ref{prop:langevin-continuity-bound} ensures that $\|\bx_t-\bx_{k\eta}\|\leq O(\sqrt{\eta N})$ for all $t\in [k\eta,(k+1)\eta]$ with probability $1-e^{-cN}$.
    It easily follows that the integral of the first and second terms are both $O(\eta^{3/2}\sqrt{N})$, using Proposition~\ref{prop:gradients-bounded} in the former case.
    The last term is by definition
    \[
    -
    \sqrt{2}
    \lt(\int_{k\eta}^{(k+1)\eta}
    \frac{\bx_t \bx_t^{\top}}{N}\de \bB_t\rt)^{\perp}.
    \]
    The stochastic integrand has Frobenius norm $1$ almost surely for each $t$, so the resulting stochastic integral has average $L^2$ norm $\eta$, and thus vanishes in probability upon division by $\sqrt{N}$.
    Finally we claim that
    \[
    \|(\hat\bx^{\eta}_{k+1})^{\perp}
    -
    \lt(\bx_{k\eta}
    +\beta \eta\nabla H_N(\bx_{k\eta})+\sqrt{2\eta} \bg_k\rt)^{\perp}\|\leq O(\eta^{3/2}\sqrt{N}).
    \]
    Recall from \eqref{eq:first-term-discrete-langevin-denominator} that with probability $1-e^{-cN}$,
    \[
    \|\bx_{k\eta}
    +\beta \eta\nabla H_N(\bx_{k\eta})+\sqrt{2\eta} \bg_k\|^2/N
    =
    1+O(\eta).
    \]
    Therefore
    \begin{align*}
    \|(\hat\bx^{\eta}_{k+1})^{\perp}
    -
    \lt(\bx_{k\eta}
    +\beta \eta\nabla H_N(\bx_{k\eta})+\sqrt{2\eta} \bg_k\rt)^{\perp}\|
    &\leq
    O(\eta)
    \cdot \lt\|\lt(\bx_{k\eta}
    +\beta \eta\nabla H_N(\bx_{k\eta})+\sqrt{2\eta} \bg_k\rt)^{\perp}\rt\|
    \\
    &\leq
    O(\eta^{3/2}\sqrt{N})
    \end{align*}
    by Proposition~\ref{prop:langevin-continuity-bound}.
    Combining the above, we conclude that the left-hand side of \eqref{eq:last-main-estimate-discrete-langevin} is $O(\eta^{3/2}\sqrt{N})$.
    This yields \eqref{eq:discrete-time-gronwall}.
    Finally, iterating \eqref{eq:discrete-time-gronwall} gives that for all $0\le k\le T/\eta$,
    \begin{align*}
    	D_{k,\eta}
    	&\le \sum_{\ell=0}^{k-1} (1+O(\eta))^\ell O(\eta^{3/2}) + o_N(1) \\
    	&\le (1+O(\eta))^k O(\eta^{1/2}) + o_N(1)
    	\le O(\eta^{1/2}) + o_N(1),
    \end{align*}
    where the last inequality uses that $(1+O(\eta))^k \le e^{O(T)}$ is bounded by a constant.
\end{proof}

\begin{lemma}
	\label{lem:normalize-almost-contraction}
	Suppose $\bx,\by \in \bbR^N$ and $\|\bx\|, \|\by\| \ge \sqrt{N} / C$.
	Let $\widehat \bx = \sqrt{N} \cdot \bx / \|\bx\|$, $\widehat \by = \sqrt{N} \cdot \by / \|\by\|$.
	Then, $\|\widehat \bx - \widehat \by\| \le C \|\bx-\by\|$.
\end{lemma}
\begin{proof}
	We directly calculate
	\[
		\lt\|\widehat \bx - \widehat \by\rt\|^2
		= \fr{\|\bx-\by\|^2 - (\|\bx\|-\|\by\|)^2}{\|\bx\|\|\by\|} \cdot N
		\le \fr{\|\bx-\by\|^2}{\|\bx\|\|\by\|} \cdot N
		\le C^2 \|\bx-\by\|^2. \qedhere
	\]
\end{proof}

\subsection*{Acknowledgement}

Thanks to Ahmed El Alaoui, David Gamarnik, Cris Moore, Mehtaab Sawhney, Alex Wein, and Lenka Zdeborov{\'a} for helpful discussions.

\footnotesize
\bibliographystyle{alpha}
\bibliography{bib}

\end{document}

%% file: commands.tex
\newcommand{\bea}{\begin{eqnarray}}
\newcommand{\eea}{\end{eqnarray}}
\newcommand{\<}{\langle}
\renewcommand{\>}{\rangle}

\newcommand{\wt}{\widetilde}
\newcommand{\op}{\text{op}}
\newcommand{\wh}{\widehat}

\newcommand\eg{{\text{\eg~}}}

\def\Unif{{\sf Unif}}

\def\eps{{\varepsilon}}

\def\bh{\boldsymbol{h}}

\def\ind{{\mathbbm 1}}

\def\bSig{{\boldsymbol{\Sigma}}}

\def\bg{{\boldsymbol{g}}}

\def\bx{{\boldsymbol{x}}}

\def\bone{{\mathbf 1}}

\def\cG{{\mathcal G}}
\def\cT{{\mathcal T}}

\def\blambda{{\boldsymbol \lambda}}

\def\sS{{\mathscr S}}

\def\op{{\rm op}}

\def\bsig{{\boldsymbol {\sigma}}}

\def\brho{{\boldsymbol \rho}}

\def\by{{\boldsymbol y}}

\def\bv{{\boldsymbol{v}}}
\def\bz{{\boldsymbol{z}}}
\def\bx{{\boldsymbol{x}}}

\def\bA{\boldsymbol{A}}
\def\bB{\boldsymbol{B}}

\def\bm{\boldsymbol{m}}

\def\ons{\mathbf{ons}}
\def\ONS{\mathbf{ONS}}

\def\de{{\rm d}}

\def\<{\langle}
\def\>{\rangle}

\def\aux{{\rm aux}}
\def\sign{{\rm sign}}

\def\cM{{\cal M}}
\def\ocM{{\overline{\cal M}}}
\def\omu{{\overline{\mu}}}
\def\cN{{\cal N}}

\def\cL{{\cal L}}

\def\by{{\boldsymbol{y}}}

\def\bw{{\boldsymbol{w}}}

\def\be{{\boldsymbol{e}}}
\def\blambda{{\boldsymbol{\lambda}}}

\def\b0{{\boldsymbol{0}}}

\def\Bin{{\sf Bin}}

\def\rd{{\mathrm {rad}}}

\def\bG{{\boldsymbol G}}

\DeclareMathOperator*{\plim}{p-lim}
\DeclareMathOperator*{\plimsup}{p-limsup}

\def\cA{{\mathcal A}}

\def\cS{{\mathcal S}}

\def\err{{\sf err}}

\renewcommand{\b}{\mathbf{b}}

\def\fr{\frac}
\def\lt{\left}
\def\rt{\right}

\def\la{\langle}
\def\ra{\rangle}

\def\eps{\varepsilon}

\def\bbE{{\mathbb{E}}}

\def\bbN{{\mathbb{N}}}
\def\bbP{{\mathbb{P}}}
\def\bbR{{\mathbb{R}}}

\def\bbT{{\mathbb{T}}}
\def\bbW{{\mathbb{W}}}
\def\bbZ{{\mathbb{Z}}}

\def\cA{{\mathcal{A}}}

\def\cN{{\mathcal{N}}}
\def\cP{{\mathcal{P}}}

\def\round{{\mathsf{round}}}

\def\sH{{\mathscr{H}}}

\def\bg{{\mathbf{g}}}
\def\bh{{\boldsymbol{h}}}

\def\OGP{{\mathsf{OGP}}}

\def\TV{{\mathrm{TV}}}
\def\GS{{\mathrm{GS}}}

\def\Ssolve{S_{\mathrm{solve}}}
\def\psolve{p_{\mathrm{solve}}}
\def\punstable{p_{\mathrm{unstable}}}
\def\Sstab{S_{\mathrm{stab}}}
\def\Sstable{S_{\mathrm{stab}}}
\def\Sall{S_{\mathrm{all}}}

\def\Sogp{S_{\mathrm{ogp}}}
\def\Smax{S_{\mathrm{max}}}

\def\Schaos{S_{\mathrm{chaos}}}
\def\Sramsey{S_{\mathrm{Ramsey}}}

\def\sph{\mathrm{sp}}

\newcommand{\norm}[1]{{\lt\|#1\rt\|}}
\newcommand{\tnorm}[1]{{\|#1\|}}

\newcommand{\He}{\mathrm{He}}

\def\ons{\mathbf{ons}}
\def\ONS{\mathsf{ONS}}

\def\sym{{\rm sym}}

\def\err{{\sf err}}

\def\Ber{{\mathsf{Ber}}}

\def\ttT{{\mathtt{T}}}
\def\ttF{{\mathtt{F}}}
\def\sat{{\mathrm{sat}}}

\def\err{{\mathsf{err}}}